\documentclass[pra, aps, twocolumn, notitlepage, nofootinbib, showpacs, floatfix, nobalancelastpage, superscriptaddress, longbibliography, 10pt]{revtex4-1}
\usepackage{amsmath,amssymb,amsthm,mathrsfs,amsfonts,dsfont}
\usepackage{graphicx}
\usepackage{xspace}
\usepackage{epstopdf}
\usepackage[nodayofweek]{datetime}
\usepackage{cryptocode}
\usepackage{tcolorbox}
\usepackage{bbold}
\usepackage{color}
\usepackage{braket}
\usepackage[hyperfootnotes=false]{hyperref}
\usepackage{xcolor}
\usepackage{bm}
\usepackage[10pt]{moresize}
\usepackage{amscd}
\usepackage{url}
\usepackage{enumerate}
\usepackage{diagbox}
\usepackage[utf8]{inputenc}
\usepackage{multirow}
\usepackage{subfigure,epsfig}
\usepackage[misc]{ifsym}
\definecolor{mylinkcolor}{rgb}{0,0,0.8} 
\hypersetup{unicode=true,%
  bookmarksnumbered=false,bookmarksopen=false,bookmarksopenlevel=1, %
  breaklinks=true,pdfborder={0 0 0},colorlinks=true}%
\hypersetup{%
  anchorcolor=mylinkcolor,citecolor=mylinkcolor, %
  filecolor=mylinkcolor,linkcolor=mylinkcolor, %
  menucolor=mylinkcolor,runcolor=mylinkcolor, %
  urlcolor=mylinkcolor}

\newtheorem{theorem}{Theorem}
\newtheorem{lemma}{Lemma}
\newtheorem{corollary}{Corollary}
\theoremstyle{definition}
\newtheorem{definition}{Definition}
\newtheorem{remark}{Remark}

\DeclareFontShape{OT1}{cmr}{bx}{sc}{<-> cmbcsc10}{}
\DeclareMathAlphabet\mathbfcal{OMS}{cmsy}{b}{n}

\newcommand{\freq}{\mathrm{freq}}

\newcommand{\ketbra}[2]{\ket{#1}\!\!\bra{#2}}

\newcommand{\tr}{\mathrm{Tr}}
\newcommand{\Max}{\mathrm{Max}}
\newcommand{\Min}{\mathrm{Min}}
\newcommand{\var}{\mathrm{Var}}
\newcommand{\varq}{\mathrm{Var}_p}
\newcommand{\rate}{\mathrm{rate}}
\newcommand{\rateopt}{\mathrm{rate}_{\mathrm{opt}}}
\newcommand{\DIRNG}{\textsc{DIRNG}\xspace}
\newcommand{\DIRNE}{\textsc{DIRNE}\xspace}

\newcommand{\A}{A}
\newcommand{\B}{B}
\newcommand{\X}{X}
\newcommand{\Csys}{D}
\newcommand{\ceil}[1]{\lceil#1\rceil}
\newcommand{\floor}[1]{\lfloor#1\rfloor}
\newcommand{\wexp}{\omega_{\mathrm{exp}}}

\newcommand{\epsmooth}{\epsilon_{h}}
\newcommand{\epsound}{\epsilon_{\mathcal{S}}}
\newcommand{\epcomplete}{\epsilon_{\mathcal{C}}}

\newcommand{\epext}{\epsilon_{\mathrm{EXT}}}
\newcommand{\epeat}{\epsilon_{\mathrm{EAT}}}

\newcommand{\Q}{\mathcal{Q}}

\newcommand{\Qgamma}{\Q^{\gamma}}
\renewcommand{\P}{\mathcal{P}}
\newcommand{\M}{\mathcal{M}}
\newcommand{\N}{\mathcal{N}}

\renewcommand{\S}{\mathcal{S}}

\begin{document}

\title{Device-independent randomness expansion against quantum side information}

\author{Wen-Zhao Liu}
\author{Ming-Han Li}
\affiliation{Hefei National Laboratory for Physical Sciences at Microscale and Department of Modern Physics, University of Science and Technology of China, Hefei 230026, P.~R.~China}
\affiliation{Shanghai Branch, CAS Center for Excellence and Synergetic Innovation Center in Quantum Information and Quantum Physics, University of Science and Technology of China, Shanghai 201315, P.~R.~China}
\affiliation{Shanghai Research Center for Quantum Sciences, Shanghai 201315, China}

\author{Sammy Ragy}
\affiliation{Department of Mathematics, University of York, Heslington, York YO10 5DD, United Kingdom}
\author{Si-Ran Zhao}
\author{Bing Bai}
\author{Yang Liu}
\affiliation{Hefei National Laboratory for Physical Sciences at Microscale and Department of Modern Physics, University of Science and Technology of China, Hefei 230026, P.~R.~China}
\affiliation{Shanghai Branch, CAS Center for Excellence and Synergetic Innovation Center in Quantum Information and Quantum Physics, University of Science and Technology of China, Shanghai 201315, P.~R.~China}
\affiliation{Shanghai Research Center for Quantum Sciences, Shanghai 201315, China}

\author{Peter J. Brown}
\affiliation{Department of Mathematics, University of York, Heslington, York YO10 5DD, United Kingdom}

\author{Jun Zhang}
\affiliation{Hefei National Laboratory for Physical Sciences at Microscale and Department of Modern Physics, University of Science and Technology of China, Hefei 230026, P.~R.~China}
\affiliation{Shanghai Branch, CAS Center for Excellence and Synergetic Innovation Center in Quantum Information and Quantum Physics, University of Science and Technology of China, Shanghai 201315, P.~R.~China}
\affiliation{Shanghai Research Center for Quantum Sciences, Shanghai 201315, China}

\author{Roger Colbeck}
\affiliation{Department of Mathematics, University of York, Heslington, York YO10 5DD, United Kingdom}

\author{Jingyun Fan}
\affiliation{Hefei National Laboratory for Physical Sciences at Microscale and Department of Modern Physics, University of Science and Technology of China, Hefei 230026, P.~R.~China}
\affiliation{Shanghai Branch, CAS Center for Excellence and Synergetic Innovation Center in Quantum Information and Quantum Physics, University of Science and Technology of China, Shanghai 201315, P.~R.~China}
\affiliation{Shanghai Research Center for Quantum Sciences, Shanghai 201315, China}
\affiliation{Shenzhen Institute for Quantum Science and Engineering and Department of Physics, Southern University of Science and Technology, Shenzhen, 518055, P.~R.~China }

\author{Qiang Zhang}
\author{Jian-Wei Pan}
\affiliation{Hefei National Laboratory for Physical Sciences at Microscale and Department of Modern Physics, University of Science and Technology of China, Hefei 230026, P.~R.~China}
\affiliation{Shanghai Branch, CAS Center for Excellence and Synergetic Innovation Center in Quantum Information and Quantum Physics, University of Science and Technology of China, Shanghai 201315, P.~R.~China}
\affiliation{Shanghai Research Center for Quantum Sciences, Shanghai 201315, China}

\begin{abstract}
The ability to produce random numbers that are unknown to any outside party is crucial for many applications. Device-independent randomness generation~\cite{ColbeckThesis,CK2,PAMBMMOHLMM,ADFRV} does not require trusted devices and therefore provides strong guarantees of the security of the output, but it comes at the price of requiring the violation of a Bell inequality for implementation. A further challenge is to make the bounds in the security proofs tight enough to allow randomness expansion with contemporary technology. Although randomness has been generated in recent experiments~\cite{Liu1,Shen,QRNG1,Liu2,QRNG2}, the amount of randomness consumed in doing so has been too high to certify expansion based on existing theory. Here we present an experiment that demonstrates device-independent randomness expansion~\cite{ColbeckThesis,CK2,PAMBMMOHLMM,Fehr13,coudron2014infinite,MS1,miller2017universal,VV,BRC}. By developing a Bell test setup with a single-photon detection efficiency of around $84\%$ and by using a spot-checking protocol, we achieve a net gain of $2.57\times10^8$ certified bits with a soundness error $3.09\times10^{-12}$. The experiment ran for $19.2$~h, which corresponds to an average rate of randomness generation of $13,527$ bits per second. By developing the entropy accumulation theorem~\cite{DFR,ADFRV,DF}, we establish security against quantum adversaries. We anticipate that this work will lead to further improvements that push device-independence towards commercial viability.
\end{abstract}
\maketitle

According to quantum theory, measurement outcomes are in general unpredictable, even to observers who possess quantum devices. Quantum processes have therefore been studied extensively as a source of randomness~\cite{acin2016certified,herrero2017quantum}. In a typical quantum random number generator, the user relies on the device working in a particular way, for instance, by detecting single photons after they pass through a 50:50 beam splitter. Deviations in the device behaviour affect the randomness of the outputs and are difficult to detect. Furthermore, any real device will be too complicated to model in its entirety, which leaves open the possibility that an adversary can exploit a feature of the device outside the model, as has been seen in quantum key distribution~\cite{GLLSKM}. To circumvent this, device-independent protocols have been introduced, which are proven to be secure without any assumptions about the devices used. This leads to a substantially higher level of security because any problems caused by unmodelled features are removed.

Experimental device-independent randomness generation (\DIRNG) has been improving at a rapid pace. Some previous studies required additional assumptions~\cite{PAMBMMOHLMM,Liu1,Shen,QRNG1}, and even the most advanced to date~\cite{QRNG2,Liu2} consumed more randomness than they generated. Therefore, randomness expansion, which is a quantum feature without classical counterpart, remained elusive and technically challenging. For example, with our previous experimental setup~\cite{Liu2}, almost $118,000$ experimental hours (at a repetition rate of $200$~kHz) would be required to achieve randomness expansion with the protocol presented below, putting it out of reach in practice.

In this Letter we report the experimental realization of device-independent randomness expansion (\DIRNE) with high statistical confidence, the success of which is based on substantial improvements on both the theoretical and the experimental sides. We derive a tighter bound on entropy accumulation in the randomness generation process and construct a photonic entanglement platform to violate the Clauser-Horne-Shimony-Holt (CHSH) inequality~\cite{CHSH} as much as possible. The significance of this work is twofold in that it advances both our understanding of randomness and our experimental quantum optical capabilities. Such improvements bring us closer to being able to realise a number of other critical quantum information tasks such as device-independent quantum key distribution~\cite{murta2019towards}.

The entropy accumulation theorem (EAT)~\cite{DFR,ADFRV,DF} provides relatively tight bounds on the amount of randomness that can be extracted against an adversary that is limited only by quantum theory. Roughly speaking, the EAT shows that in an $n$-round protocol that achieves a CHSH game score of $\omega$, the amount of output randomness is lower-bounded by
\begin{align}\label{eq:1}
\mathrm{rand}_{\mathrm{out}}\ge n h(\omega)-\sqrt{n} v\,,
\end{align}
where $h(\omega)$ is the worst-case von Neumann entropy of an individual round of the protocol with expected score $\mathrm{\omega}$, and $v$ is a correction factor accounting for the finite statistics. The score on round $i$ is $\frac{1}{2}(1+(-1)^{A_i\oplus B_i\oplus(X_i\cdot Y_i)})$, where $A_i$ and $B_i$ are measurement outcomes and $X_i$ and $Y_i$ are measurement setting choices at the two sites, with $A_i$, $B_i$, $X_i$ and $Y_i\in \{0,1\}$ (Figure~\ref{fig:CHSHprot}). By using ideas from the improved EAT~\cite{DF}, we derive a tighter lower bound on the accumulated entropy (Methods). This allows us to use a spot-checking protocol to experimentally realise randomness expansion with a state-of-art experimental quantum optical technique.

We provide a conceptual drawing of our spot-checking device-independent protocol (Figure~\ref{fig:CHSHprot}) in which the assumptions are outlined. The underlying idea is to check whether devices situated in a secure lab violate a Bell inequality; it is therefore important to ensure that the devices at both sites (labelled Alice and Bob) cannot signal to one another or to the outside of the lab. If a Bell inequality is violated while satisfying our assumptions, then the devices must be generating randomness, even relative to an adversary who may share entanglement with the devices. The generated randomness can be extracted by appropriate post-processing. In this protocol (Box~1), the initial randomness is required to decide whether a round is a test round, $T_i=1$ (with probability $\gamma$), or a generation round, $T_i=0$ (with probability $1-\gamma$). $T_i$ is then communicated to two separate sites (but not to the measurement devices). In a test round, an independent uniform random number generator at each site generates the input to each device to perform the CHSH game. A test round consumes 2 bits of randomness. In a generation round, the devices at the two sites are given the input `0'. Crucially, each measurement device learns only its own input and not whether a round was a test round or generation round.

\begin{figure}[htbp]
\centering
\resizebox{7.5cm}{!}{\includegraphics{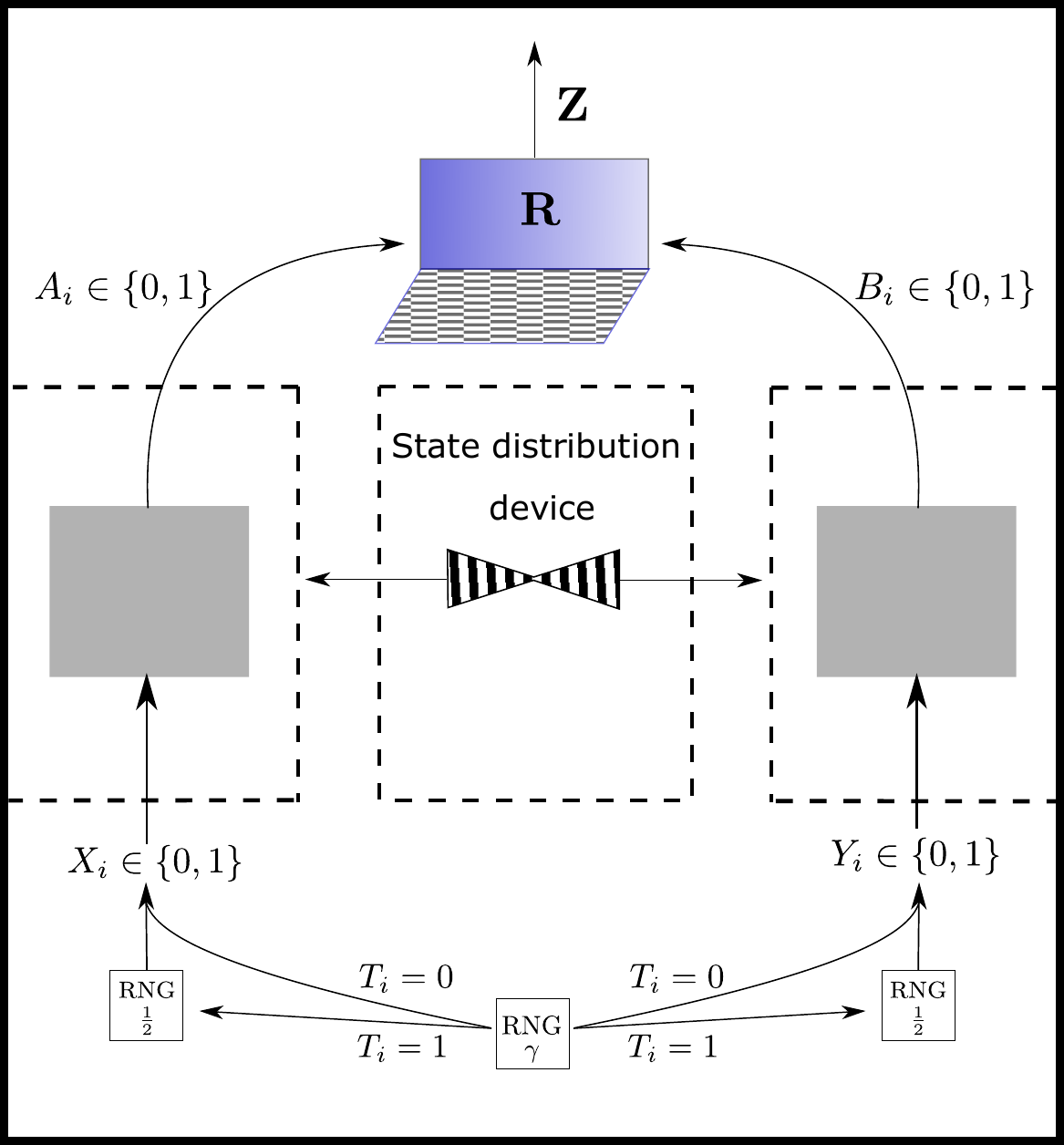}}
\caption{{\bf Conceptual sketch of the \DIRNE protocol setup.} The protocol (Box~1) takes place in a secure lab, which is shielded from direct communication to the outside. The lab contains two black-box devices that accept inputs and yield outputs from the binary alphabet $\lbrace 0,1\rbrace$ and these can be shielded from communicating at will. In particular, we assume that the user can completely control the flow of classical communication in and out of these regions (indicated by the dashed lines). In our experiment, the secure lab contains two sites, Alice and Bob. They share a pair of entangled particles which may be distributed from a central station. (If we had good enough quantum storage, then all entanglement could be pre-shared.) Alice's and Bob's respective inputs are $X_i$ and $Y_i$ and their outputs are $A_i$ and $B_i$. The user also possesses a trusted classical computer (with which to process the classical data) and sources of initial randomness. In our experiment the initial randomness is shown by an extractor seed $\mathbf{R}$ and three random number generators (RNGs) that determine the inputs to the devices. These RNGs output either $0$ or $1$, where the number in the box ($\gamma$ or $1/2$) denotes the probability of $1$. The central RNG determines the type of round ($T_i=1$ meaning test and $T_i=0$ meaning generate) and the peripheral RNGs determine the inputs if a test round is chosen. The final randomness output is denoted by $\mathbf{Z}$.}
\label{fig:CHSHprot}
\end{figure}

\begin{figure*}
\begin{tcolorbox}[title = \textbf{Box 1 : CHSH-based \DIRNE Protocol},colback=white,colframe=blue!50!white]
\begin{flushleft}
 \textbf{Arguments:}\\
 $n\in\mathbb{N}$ --  number of rounds \\
 $\gamma\in(0,1)$ --  test probability \\
 $\omega_{\mathrm{exp}}$ --  expected CHSH score given a test round \\
 $\delta>0$ -- width of the statistical confidence interval for the CHSH score \\
 $\textbf{R}$ -- random seed for the extractor
\end{flushleft}
\tcblower
\leftline{\textbf{Protocol:}}

\vspace{0cm}
\begin{enumerate}\setlength{\itemsep}{0ex}\setlength{\parsep}{0ex}
  \item For every round $i\in\lbrace 1,\dots,n\rbrace$, do steps $2-4$.
  \item Set $U_i=\perp$. Choose $T_i\in \lbrace 0,1\rbrace$ such that ${\rm Pr}(T_i=1)=\gamma$.
  \item If $T_i=0$, use the devices with inputs $(X_i,Y_i)=(0,0)$, record $A_i$, replace $B_i$ with $0$ and set $U_i=\bot$.
  \item If $T_i=1$, choose the inputs $X_i$ and $Y_i$ uniformly at random from $\{0,1\}$, record $A_i$ and $B_i$ and set $U_i=\frac{1}{2}(1+(-1)^{A_i\oplus B_i\oplus(X_i\cdot Y_i)})$.
  \item If $|\{U_i:U_i=0\}|> n\gamma(1-(\omega_{\mathrm{exp}}-\delta))$, then abort the protocol.
  \item Apply a strong quantum-proof randomness extractor to get output randomness $\mathbf{M}=\text{Ext}(\mathbf{A}\mathbf{B},\mathbf{R})$. (Because we use a strong extractor, $\mathbf{M}$ can be concatenated with $\textbf{R}$ to give $\mathbf{Z}=(\mathbf{M},\mathbf{R})$.)
\end{enumerate}
\end{tcolorbox}
\end{figure*}

We implement the protocol on a quantum optical platform (Figure~\ref{fig:setup}). Pairs of polarization-entangled photons with a wavelength of 1,560~nm are generated via spontaneous parametric down-conversion and are delivered through spatial optical paths to two sites where polarization-dependent measurements are conducted. Previously, and with space-like separation between Alice and Bob, this platform proved to be robust enough to realise loophole-free violation of a Bell inequality and \DIRNG, in which the CHSH game scores $\omega$ violated the classical bound $\omega_{\mathrm{class}}=3/4$ by $0.00027$ (ref.~\citenum{Liu2}). Under these conditions and with the same error parameters as elsewhere in this paper, it would take about $8.52\times10^{13}$ rounds of the experiment to witness randomness expansion according to our revised EAT theory (Figure~\ref{fig:EAT2}a, open square). To go beyond this, in the present work we reduced the distance between Alice and Bob by replacing the fibre links with spatial optical paths to achieve single-photon detection efficiencies of $83.40\pm0.32\%$ for Alice and $84.80\pm0.31\%$ for Bob, which enables the detection loophole to be closed in the CHSH game. Following the spot-checking protocol, a biased quantum random number generator (QRNG) is used to decide whether to test or not. Its output $T_i$ is transmitted to Alice and Bob to determine whether to use the local unbiased QRNGs in each round. When $T_i=1$, the setting choices $A_i$ and $B_i$ are determined randomly, whereas when $T_i=0$, the local unbiased QRNGs are turned off and fixed measurements are made.

Before the start of the main experiment, a systematic experimental calibration is implemented and some calculations are performed to predetermine several parameters that are mentioned in the protocol. The calibration yielded a CHSH game score of $0.752487$, and we compute that for $\gamma_{\mathrm{opt}}=3.393\times10^{-4}$, which corresponds to an average input entropy rate of $0.0049$ bits per round, randomness expansion with a soundness error (Methods) of $3.09\times10^{-12}$ can be witnessed after at least $8.951\times10^{10}$ rounds (Figure~\ref{fig:EAT2}a, cross), that is, the randomness produced in the experiment surpasses the consumed entropy after this number of rounds (Appendix Section~C.1).

In the main experiment, we set $\omega_{\exp}=0.752487$, $\delta=3.52\times10^{-4}$ and $\gamma=3.264\times10^{-4}$, and conservatively set the number of rounds to $n=1.3824\times10^{11}$, which is slightly larger than the $8.951\times10^{10}$ rounds required (see Section~C.1 of the Appendix for computation of the latter). We complete all the rounds of the experiment in $19.2$~h at a repetition rate of $2$~MHz, which is much shorter than $118,000$~h (which would have been required by the previous experiment~\cite{Liu2}). The resulting CHSH game score is $\omega_{\mathrm{CHSH}}=0.752484$, which is consistent with the value we expect (and the protocol did not abort), $|\omega_{\mathrm{CHSH}}-\omega_{\mathrm{exp}}|<\delta$. The raw experimental output has a size of $0.138 \text{Tb}$. According to the development of the EAT presented in the Appendix, the output contains at least $9.350\times10^8$ quantum-certified bits of randomness, which exceeds the amount of entropy ($6.778\times10^8$ bits) required for its generation (Methods and Appendix Section~C.1). We use a personal computer to perform a Toeplitz matrix ($0.935~\text{Gb}\times0.138~\text{Tb}$) multiplication to extract the quantum-certified random bits from the raw output. The soundness error of the final output is $3.09\times10^{-12}$.

Because a quantum-proof strong extractor is applied, the seed required for the extraction remains random after its use and is therefore not consumed~\cite{KR2011}. (Technically, the seed degrades by a very small amount, which is accounted for in the soundness error given above (Appendix Section~A.3)). Note that we do not consider the other randomness required in the protocol (that is, for choosing the test rounds and the inputs on the test rounds) to be reusable because of the possibility of some subtle attacks (see section~4.2 of ref.~\citenum{CK2}). Overall we achieve \DIRNE, gaining $2.57\times10^8$ net bits with a net rate of $1.86\times10^{-3}$ bits per round against an eavesdropper that is limited by quantum theory (Figure~\ref{fig:EAT2}b, red cross).

\begin{figure*}[htbp]
\centering
\resizebox{13cm}{!}{\includegraphics{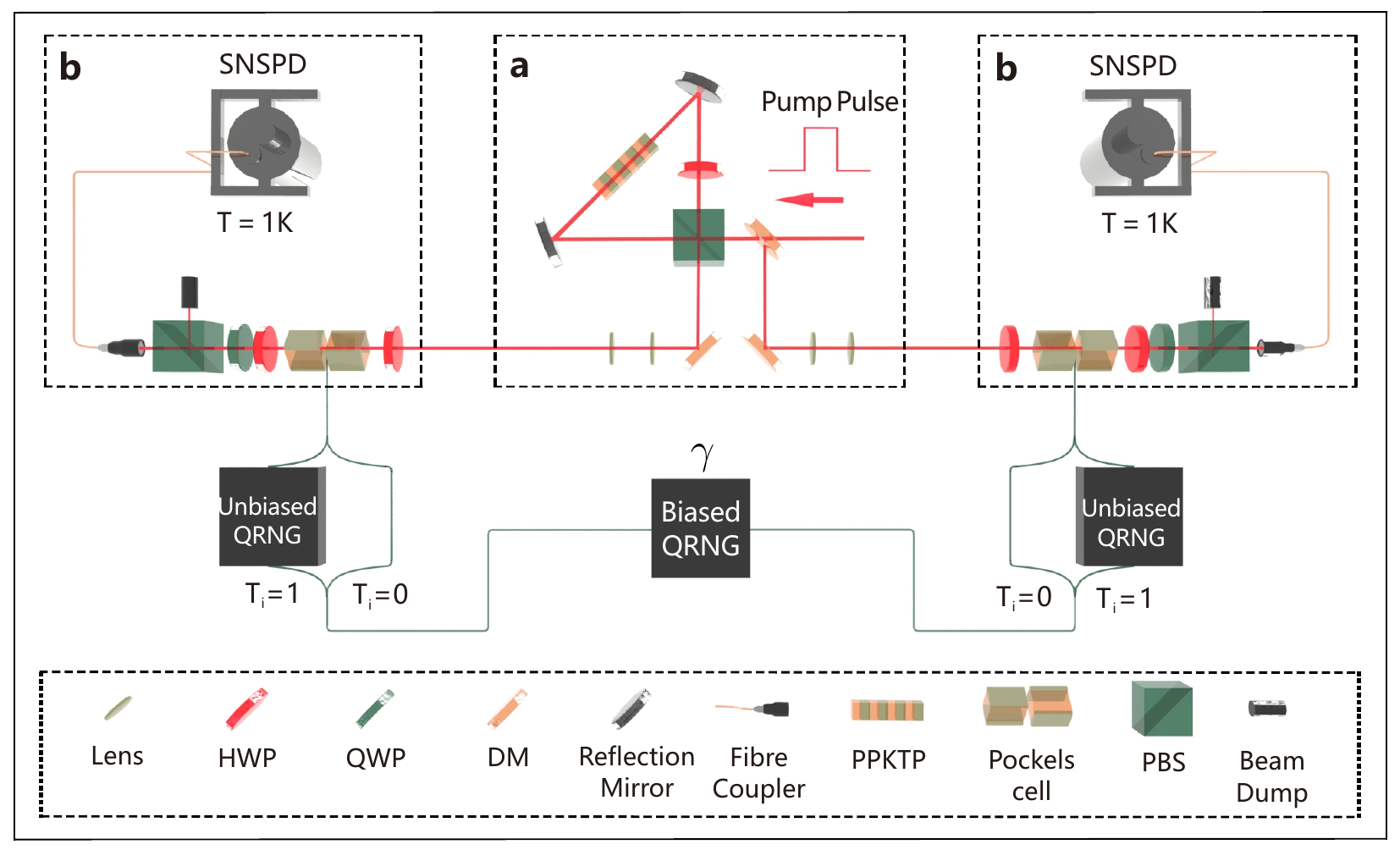}}\vspace{-0.5cm}
\caption{{\bf Schematic of the experiment.} {\bf a}, Entanglement Source, and creation of pairs of entangled photons. Light pulses of 10~ns are injected at a repetition rate of 2~MHz into a periodically poled potassium titanyl phosphate (PPKTP) crystal in a Sagnac loop to generate polarization-entangled photon pairs\cite{Liu2}. The two photons of an entangled pair at 1,560~nm travel in opposite directions to two sites, Alice and Bob, where they are subject to polarization projection measurements. {\bf b}, Alice and Bob, and single-photon polarization measurement. In the measurement sites, Alice (Bob) uses a Pockels cell to project the single photon into one of two predetermined measurement bases, and then detects single photons with a superconducting nanowire single-photon detector (SNSPD) operating at 1K. In each round, a biased QRNG in the lab creates a random bit $T_i$ with probability distribution $(\gamma, 1-\gamma)$ to determine in advance whether this round will be a test round or generation round. In test rounds, Alice and Bob each receive a random bit ``0'' or ``1'' from a local quantum random number generator (QRNG) to set the Pockels cell to zero and half-wave voltage accordingly (in generation rounds, they always use zero). DM, dichroic mirror; HWP, half-wave plate; PBS, polarizing beam splitter; QWP, quarter-wave plate.}
\label{fig:setup}
\end{figure*}

\begin{figure*}[htbp]
\centering
\resizebox{16cm}{!}{\includegraphics{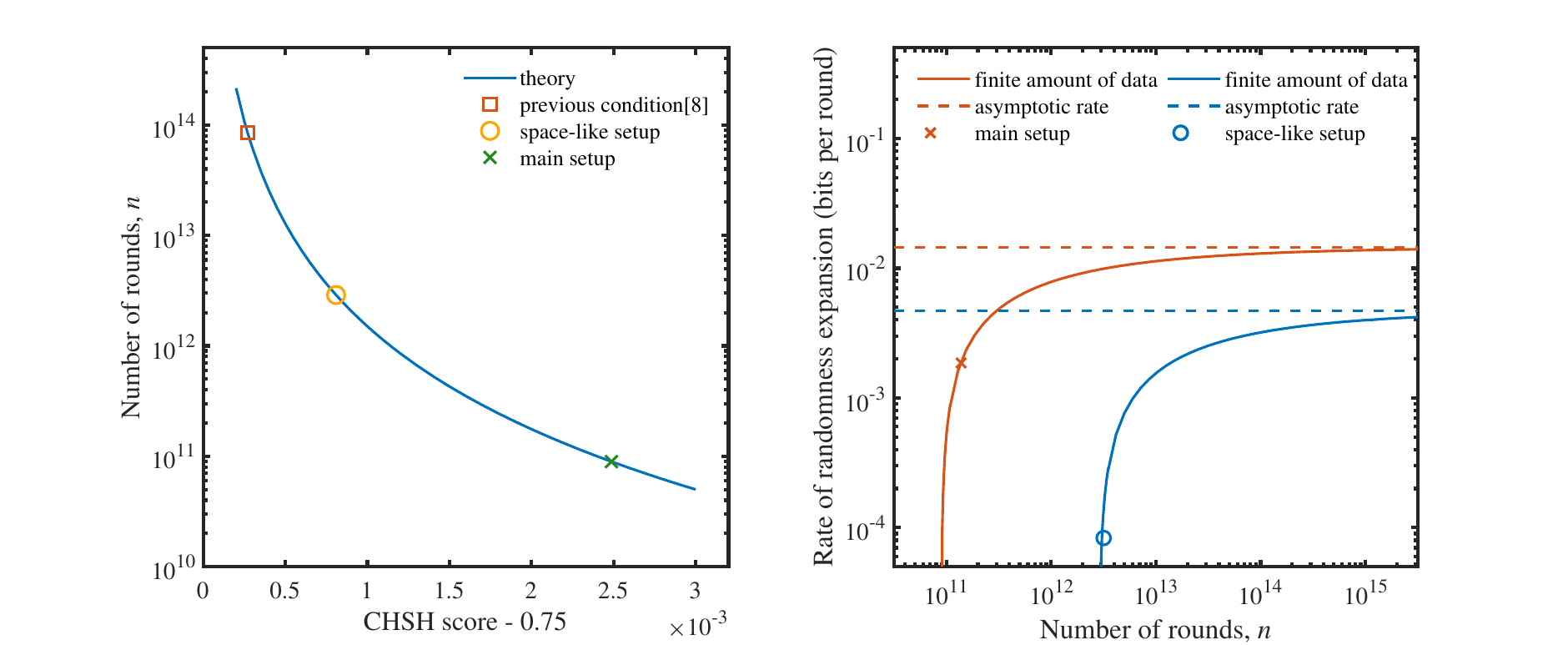}}
\caption{{\bf Rounds needed and expansion rate for the main and space-like experiments.} {\bf a}, We estimate the minimum number of experimental runs with our revised EAT theory to witness randomness expansion as a function of CHSH violation (smooth curve) with soundness error $3.09\times10^{-12}$. The red square, yellow circle and green cross indicate the previous~\cite{Liu2}, space-like and main experimental conditions, respectively. {\bf b}, We estimate the randomness expansion rate based on our revised EAT theory as a function of number of rounds (smooth line) and the asymptotic rate (dashed line) with a soundness error of $3.09\times10^{-12}$. The cross and circle indicate the experimental parameters used, red indicates the main experiment and blue indicates the space-like experiment.}
\label{fig:EAT2}
\end{figure*}

When playing the CHSH game, we close the detection loophole. Given the assumption that the devices are well shielded
, it is not necessary to close the locality loophole (Appendix Section~A.2). Considering the demanding experimental requirements to close both loopholes~\cite{hensen2015loophole,shalm2015strong,giustina2015Significant,rosenfeld2017event,LiPRL2018}, the use of shielding assumptions instead of space-like separation improves efficiency and brings \DIRNE closer to commercialization. We also remark that the randomness we generate is secure according to a composable security definition (Methods) and therefore can be used in any application that requires random numbers. Strictly, because of an issue with the composability of device-independent protocols~\cite{bckone}, without further assumption, ongoing security of the output randomness relies on the devices not being reused.

We also upgraded our previous platform~\cite{Liu2} where we closed the locality loophole to use higher-efficiency detectors, resulting in overall efficiencies of $80.41\pm0.34\%$ for Alice and $82.24\pm0.32\%$ for Bob (slightly lower than those of the main experiment). As a comparison, we analyse the performance of that setup by using the EAT framework with the same error parameters as our main experiment (Table~\ref{tab:params}).
Because of the need for additional rounds, we increased the repetition rate to $4$~MHz for this experiment. A comparison of the results (Table~\ref{tab:params}) shows that overall, we obtained $6.496\times10^{9}$ random bits within $3.168\times10^{12}$ experimental rounds, exceeding the amount of entropy ($6.233\times10^9$ bits) consumed during the protocol (Methods and Appendix Section~C.1) and gaining $2.63\times10^8$ net bits with a net rate of $8.32\times10^{-5}$ bits per round against an eavesdropper that is limited by quantum theory (Figure~\ref{fig:EAT2}b, blue open circle). The outputs of both experiments are available online \href{https://www.dropbox.com/sh/hae9ht1cc426i5g/AABcuGgGyuNJMC0zOxIhOQBGa?dl=0} {{\tt https://tinyurl.com/qssxxaq}}.

Beyond the work here, we would like protocols to have an improved rate. Robust protocols that achieve up to 2 bits of randomness per entangled qubit pair are known~\cite{BRC}. However, using such protocols to gain an experimental advantage requires improvement in the detection efficiency, which is challenging with a photonic setup. On the theory side, better rates could be achieved by developing tighter bounds on the output randomness. It would also be interesting to put into practice a protocol for randomness amplification~\cite{CR_free}, to reduce the assumption on the input randomness.\bigskip

\begin{table*}[htbp]
  \caption{{\bf Comparison between the two experiments.} 
    Here, 
    $\gamma_{\mathrm{opt}}$ is the optimal test probability to witness expansion in the minimum number of rounds $n_{\min}$ with the error parameters chosen. 
  }
  \vspace{0.2cm}
  \begin{tabular}{c|c|c}
  \hline
                                &Main experiment             &Space-like experiment       \\\hline
    Expected CHSH score ($\omega_{\exp}$)             &$0.752487$             &$0.750809$             \\
    $\gamma_{\mathrm{opt}}$     &$3.393\times10^{-4}$   &$9.851\times10^{-5}$   \\
    $n_{\min}$                  &$8.951\times10^{10}$   &$2.888\times10^{12}$   \\
    Repetition rate             &$2$~MHz                &$4$~MHz                \\
    Time taken                  &$19.2$ h           &$220$ h           \\
    Observed CHSH score ($\omega_{\mathrm{CHSH}}$)    &$0.752484$             &$0.750805$             \\
    Test probability ($\gamma$)                    &$3.264\times10^{-4}$   &$1.194\times10^{-4}$   \\
    Number of rounds ($n$)                         &$1.3824\times10^{11}$  &$3.168\times10^{12}$   \\
    Confidence width ($\delta$)                    &$3.52\times10^{-4}$    &$1.22\times10^{-4}$    \\
    Soundness error ($\epsound$)            &$3.09\times10^{-12}$   &$3.09\times10^{-12}$   \\
    Completeness error ($\epcomplete$)         &$1\times10^{-6}$       &$1\times10^{-6}$       \\
    Entropy in output           &$9.350\times10^8$ bits &$6.496\times10^9$ bits \\
    Randomness generation rate  &$13,527$ bits${\cdot}s^{-1}$ &$8,202$ bits${\cdot}s^{-1}$\\
    Entropy in input            &$6.778\times10^8$ bits &$6.233\times10^9$ bits \\
    Net gain                    &$2.57\times10^8$ bits  &$2.63\times10^8$ bits  \\\hline
  \end{tabular}
\label{tab:params}
\end{table*}

\emph{Acknowledgements.} We thank C.-L.\ Li for experimental assistance and J.-D.\ Bancal and E.\ Tan for comments on an earlier draft. This work was supported by the National Key R\&D Program of China (grant nos. 2017YFA0303900 and 2017YFA0304000), the National Natural Science Foundation of China, the Chinese Academy of Sciences, the Shanghai Municipal Science and Technology Major Project (grant no. 2019SHZDZX01), the Anhui Initiative in Quantum Information Technologies, the Guangdong Innovative and Entrepreneurial Research Team Program (grant no. 2019ZT08X324), the Key Area R\&D Program of Guangdong Province (grant no. 2020B0303010001), the Quantum Communication Hub of the Engineering and Physical Sciences Research Council (EPSRC)(grant nos. EP/M013472/1 and EP/T001011/1) and an EPSRC First Grant (grant no. EP/P016588/1). We are grateful for computational support from the University of York High Performance Computing service, {\sc Viking}, which was used for the randomness extraction.

\section{Methods}
\subsection{Security definition}
In this work we use a composable security definition~\cite{Canetti,Ben-OrMayers,PR2014}.

\begin{definition}[security]
\label{def:Security}
A protocol with an output $\mathbf{Z}$ is called $(\epsilon_{\mathcal{S}},\epsilon_{\mathcal{C}})$-secure if it satisfies the followings:

\begin{enumerate}
\item (Soundness) For an implementation of the protocol that produces $m$ bits of output we have
\begin{align}
\frac{1}{2}p_{\Omega}\|\rho_{\mathbf{Z}E|\Omega}-\tau_m\otimes\rho_{E|\Omega}\|_1\leq \epsilon_{\mathcal{S}}\,,
\end{align}
where $\tau_m$ represents a completely mixed state on $m$ qubits, $E$ represents all systems that could be held by an adversary (Eve), $\Omega$ is the event that the protocol does not abort, and $p_{\Omega}$ is the probability of this occurring. $\|\cdot \|_1$ is the trace norm. \item (Completeness) There exists an honest implementation such that $p_\Omega\geq 1-\epsilon_{\mathcal{C}}$.
\end{enumerate}
\end{definition}
The soundness error bounds the distance between the output of the protocol and that of an idealized protocol in which Eve's marginal is the same as in the real protocol, but the output is perfectly uniform and independent of Eve.

In general, the raw output of a protocol can have a lot of randomness while being easily distinguished from uniform. However, with the application of an appropriate randomness extractor, which is a classical function that takes a random seed and the raw output, an almost uniform output can be recovered. The length of this output can be taken to be roughly equal to the smooth min-entropy of the raw string that is conditioned on the side information held by Eve~\cite{Renner,KRS}.

\begin{definition}[Smooth min-entropy]
\label{def:minent}
For any classical-quantum density matrix $\rho_{AE}=\sum_a p(a)\ketbra{a}{a}\otimes\rho_E^a$ that acts on the joint Hilbert space $\mathcal{H}_{AE}$, the $\epsilon_h$-smooth min-entropy is defined by
\begin{align}
H_{\mathrm{min}}^{\epsilon_h}(A|E)_{\rho_{AE}}=\max_{\tilde{\rho}_{AE}}\left( -\log\max_{\lbrace\Pi_a\rbrace}\sum_{a} \tilde{p}(a) \tr{(\Pi^a \tilde{\rho}_E^a)}\right),
\end{align}
where the outer maximisation is over the set $\mathcal{B}^{\epsilon_h}(\rho_{AE})$ of all sub-normalized states $\tilde{\rho}_{AE}=\sum_a \tilde{p}(a)\ketbra{a}{a}\otimes \tilde{\rho}_E^a$ within purified distance~\cite{TCR2} $\epsilon_h$ of $\rho_{AE}$. Note that $\max_{\lbrace\Pi_a\rbrace}\sum_{a} \tilde{p}(a) \tr{(\Pi^a \tilde{\rho}_E^a)}$ can be interpreted as the maximum probability of guessing $A$ given access to the system $E$.
\end{definition}
The interpretation in terms of guessing probability makes clear that this quantity is a measure of unpredictability. Bounding the smooth min-entropy for a device-independent protocol is challenging. We do this by means of the entropy accumulation theorem and state an informal version that is applicable to the CHSH game below.

\subsection{Theoretical details about the protocol}
In the protocol, the user has two devices that are prevented from communicating with one another and with which the CHSH game can be played. To do so, each device is supplied with a uniformly chosen input denoted by $X,Y\in\{0,1\}$, and each produces an output, denoted by $A,B\in\{0,1\}$, respectively. The CHSH game is scored according to the function $\frac{1}{2}(1+(-1)^{A\oplus B\oplus(X\cdot Y)})$. In other words, the game is won (with a score of $1$) if $A\oplus B=X\cdot Y$ and is lost (with a score of $0$) otherwise.

At the end of the protocol the number of rounds in which the CHSH game was lost is counted and compared to $n\gamma(1-(\omega_{\mathrm{exp}}-\delta))$. The challenge in a randomness expansion protocol is to go from this to the amount of extractable randomness. For this we use the EAT, which we state informally here (note that the version we use is a development of ref.~\citenum{DF}; for more details, see the Appendix).

\begin{theorem}[Entropy accumulation, informal]
Suppose that the protocol of Box 1 is performed and that devices are such that $p_{\Omega}$ is the probability that the protocol does not abort. Let $\alpha\in(1,2)$, $\epsilon_h\in(0,1)$ and $f(s)$ be an affine lower bound on the single-round von Neumann entropy for any strategy that achieves an expected score of $s$. If the protocol does not abort, we can assume that
\begin{align}
H^{\epsilon_h}_{\min}(\mathbf{AB}|E) \geq &\ nf(\omega)+n\Delta(f,\omega)-n(\alpha-1)V(f,\gamma,\omega)\nonumber\\&-\frac{\alpha}{\alpha-1}\log\left(\frac{1}{p_{\Omega} (1-\sqrt{1-\epsilon^2_h})}\right)\nonumber\\&+n(\alpha-1)^2K_\alpha(f,\gamma),
\end{align}
where $\omega=\omega_{\mathrm{exp}}-\delta$ and the explicit forms of the functions $\Delta$, $V$ and $K_\alpha$ can be found in the Appendix.
\end{theorem}
When $\alpha-1\propto\frac{1}{\sqrt{n}}$, is set the subtracted terms scale with $\sqrt{n}$ whereas the leading rate term scales with $n$, which leads to the relation in equation~\eqref{eq:1} when $f(\omega)$ is a good approximation to $h(\omega)$, the worst-case von Neumann entropy for the observed score.

To produce the output string $\mathbf{M}$, we apply a strong quantum-proof randomness extractor. The reason we use a strong extractor is that the random seed, $\mathbf{R}$, required for the extractor remains random even when conditioned on the extractor's output and is therefore not consumed. This means that $\mathbf{M}$ can be concatenated with the extractor seed $\mathbf{R}$ to give output $\mathbf{Z}=(\mathbf{M},\mathbf{R})$. (see Appendix for a more detailed discussion of the extraction). Importantly, the length of the output (excluding the recycled seed) will be roughly $\mathrm{rand}_{\mathrm{out}}\approx H^{\epsilon_h}_{\min}(\mathbf{A}\mathbf{B}|E)$. We need this to be greater than the randomness consumed.

\begin{remark}[Input randomness]
The expected input randomness, $\mathrm{rand}_{\mathrm{in}}$ of the protocol in Box~1 is
\begin{align}
\mathrm{rand}_{\mathrm{in}}=n(H_{\mathrm{bin}}(\gamma)+2\gamma)+2\,,
\end{align}
where $H_{\mathrm{bin}}$ denotes the binary Shannon entropy. The contribution $H_{\mathrm{bin}}(\gamma)$ comes from the selection of the test rounds and $2\gamma$ from the selection of the input bits for the CHSH game. The interval algorithm~\cite{HaoHoshi} can be used to turn uniform random bits to biased ones at the claimed rate.
\end{remark}

We do not include the randomness that is necessary for seeding the extractor in the above because it is not consumed, although it is needed to run the protocol.

Suppose that a protocol has some fixed expected score $\omega_{\mathrm{exp}}$. To demonstrate randomness expansion, that is, $\mathrm{rand}_{\mathrm{out}}-\mathrm{rand}_{\mathrm{in}}>0$, at this performance we have to choose the parameters $n$ and $\gamma$ appropriately. Increasing $n$ leads to an improvement in the rate, but takes longer and increases the experimental difficulty. The trade-off with $\gamma$ appears in the $\mathrm{rand}_{\mathrm{out}}$ and $\mathrm{rand}_{\mathrm{in}}$ terms. The input randomness evidently decreases as $\gamma$ shrinks, which is favourable because the input randomness is subtracted to calculate the randomness gained. However, the min-entropy also decreases because the error term scales roughly as $\frac{1}{\sqrt{\gamma}}$ (ref.~\citenum{DF}). Moreover, the statistical confidence decreases with less frequent testing, and as such the threshold score for successful parameter estimation must be lowered (that is, $\delta$ increased) to obtain a small completeness error (see Appendix for how to calculate this error). This also has a negative impact on the randomness produced.

\onecolumngrid
\appendix
\section{Theory of Device-Independent Quantum Randomness Expansion}

\subsection{Notation}
We summarize the main notation used in this paper. We use $\log$ for the logarithm base 2 and $\ln$ for the natural logarithm. The function $\operatorname{sign}(x) := \frac{x}{|x|}$ for $x \in \mathbb{R}$ is the sign function with $\operatorname{sign}(0) = 0$. We generally denote states by lower case Greek letters, such as $\rho$ and $\sigma$. These will often be subscripted by capital Latin letters denoting the Hilbert spaces upon which the state acts, such as a state $\rho_{ABE}$ denoting a state on Alice, Bob and Eve's joint system, although when clear from context the system subscripts may be omitted. We denote the set of states for a Hilbert space $A$ by $\S(A)$. We consider classical-quantum (CQ) states, which take the form
\begin{align*}
\rho_{AE}=\sum_{a\in\mathcal{A}}p(a)\ketbra{a}{a}\otimes \rho_E^{(a)}.
\end{align*}
The classical system, or register, $A$ contains the letters, $a$, of
some alphabet, $\mathcal{A}$. The quantum system $E$ carries states $\rho_E^{(a)}$ where the bracketed superscript indicates the value of the classical system upon which the state depends.

We account for multiple rounds by encoding the result of each round to
a different system. We use a subscript, e.g., $A_i$ to denote Alice's
system for round $i$. We denote quantum channels by calligraphic
letters, such as $\mathcal{M}_i$, which we also subscript with the
round-number to which the channel relates when appropriate. When a
channel is defined on some system $A$, we extend it to composite systems
implicitly, hence
$\mathcal{M}(\rho_{AR})=\mathcal{M}\otimes\mathcal{I}_R(\rho_{AR})$,
where $\mathcal{I}_R$ is the identity channel on the $R$ system. When
there are $n$ systems in total, we will sometimes use bold face, i.e.,
$\mathbf{A}=A_1\dots A_n$, and likewise for values on these $n$
systems, so
$\ketbra{\mathbf{a}}{\mathbf{a}}_{\mathbf{A}}=\ketbra{a_1}{a_1}_{A_1}\otimes \dots \otimes \ketbra{a_n}{a_n}_{A_n}$. We
also use capital letters to refer to the random variables on the
classical registers of CQ states.

We use slightly different notation from the main body to refer to the inputs and outputs for the EAT. In the main body, we have chosen notation which is consistent with the usual notation for cryptography, $A$ and $B$ denote outputs for each device and $X$ and $Y$ inputs. Here, however, we choose for our notation to be consistent with that of~\cite{DFR,DF} for readability since we adapt their proof. To translate from the appendix to the main text, we can make the substitutions $A\to AB$, $B\to XY$ and $U\to X$, i.e., $A$ is now the register containing the joint outputs, $B$ is the register containing the joint inputs and $X$ is the register containing the test scores.

\subsection{Protocol Assumptions}\label{SupSec:Protocol}
We recall the assumptions necessary for randomness expansion (see e.g.~\cite{ColbeckThesis}):

\begin{enumerate}
\item \label{Assum1} The user has a secure laboratory and can prevent any devices from receiving or sending communication at will.
\item \label{Assum2} The user possesses a trusted classical computer.
\item \label{Assum3} The user has some initial trusted randomness.
\item Quantum theory is correct and complete (see~\cite{CR_ext,CR_book1} for a connection between the two).
\end{enumerate}
The first assumption is necessary to prevent the devices from communicating directly to an adversary, which would trivially compromise security. We also use it to ensure that the devices do not learn whether a particular round is a test round or not (if they could learn this, they could behave honestly only on test rounds, which could violate the security).  Moreover, to genuinely violate a Bell inequality the devices must not be able to communicate their inputs to each other before producing their outputs, which can either be ensured by shielding or space-like separation.

The second assumption is required for correct data processing. In the extreme case, an untrusted computer could simply substitute the real key with a compromised one.  It may seem unreasonable to allow trust in classical computers but not in quantum devices, but classical computations are repeatable, so trust in them can be gained by repeating a calculation on different computers.

The third assumption is needed because the inputs to the devices must be random in order to reliably test that a Bell inequality is violated and randomness is also required to seed the extractor. The final assumption constrains the devices and the adversary to operate according to quantum mechanics, and is needed for the theoretical arguments to work.  The initial randomness is often thought of as a pre-existing seed. In our experiment for convenience RNGs were used to produce the inputs to the devices.  Although we have used RNGs in the protocol, the key point is that more randomness is generated than is used. At the same time we have demonstrated a Bell inequality violation on the test rounds while closing the detection and locality loopholes.

Because of Assumption 1, it is not strictly necessary to close the locality loophole in our experiment.  The important thing is to ensure that, before giving its output, each device only learns its own input and not that of the other device.  We can ensure this by shielding the devices.  (This is sufficient because the aim of the present experiment is not to disprove local realism.) Since we have chosen to implement a spot-checking protocol, even in the setup where Alice and Bob are far apart, we do not have full space-like separation over all rounds since the selection of a generation round is jointly communicated to the device locations by a central QRNG.  The shielding assumption is hence necessary to ensure that the devices do not alter their behaviour between generation rounds and test rounds. Furthermore, thinking of randomness expansion as a process that extends a pre-existing seed requires that this seed be shielded from the devices, a task that is not connected to their space-like separation. We also remark that it is undesirable to use space-like separated measurements when developing commercial devices implementing DIRNE, where compactness is important. Hence, our main experiment with shielding assumptions is more relevant for assessing the state-of-the-art performance such a commercial device could have.

\subsection{Extraction} \label{app:extraction}
Given a bound on the min-entropy, a seeded extractor is a deterministic function which takes as arguments the output string and a random seed and produces a string that is almost indistinguishable from uniform.
\begin{definition}[Strong extractor~\cite{NZ96,KR2011,DPVR}]
\label{def:ext} A function $\mathrm{Ext:\, }\lbrace 0,1\rbrace^n
\times \lbrace 0,1\rbrace^d\to \lbrace 0,1\rbrace^m$ is a
\textit{quantum-proof} $(k,\epext)$-strong extractor
with uniform seed, if for any classical-quantum state $\rho_{XE}$ with $H_{\min}(X|E)_{\rho}\geq k$ and for uniform seed $Y$ we have
\begin{align*}
\frac{1}{2}||\rho_{\mathrm{Ext}(X,Y)YE}-\tau_m\otimes \tau_d \otimes \rho_E||\leq \epext\,,
\end{align*}
where $\tau_d$ is the maximally mixed state of dimension $d$.
\end{definition}
A quantum-proof $(k,\epext)$-strong extractor thus
provides $m$ bits of randomness (except with probability
$\epext$) given a guarantee of $k$ bits of
min-entropy in $X$.
Extractors can also use the smooth min-entropy instead, which often
leads to more randomness (the smooth min-entropy can be much larger
than the min-entropy) in the output with a relatively small penalty
in the error term.
\begin{lemma}[~\cite{Renner,DPVR}]\label{lem:ext}
If $\mathrm{Ext:\, }\lbrace 0,1\rbrace^n  \times \lbrace 0,1\rbrace^d\to \lbrace 0,1\rbrace^m$ is a \textit{quantum-proof} $(k,\epext)$-strong extractor, then for any classical-quantum state $\rho_{XE}$ with $H_{\min}^{\epsmooth}(X|E)_{\rho_{XE}}\geq k$
\begin{align*}
\frac{1}{2}||\rho_{\mathrm{Ext}(X,Y)YE}-\tau_m\otimes \tau_d\otimes \rho_E||\leq \epext + 2\epsmooth\,.
\end{align*}
\end{lemma}

The crucial feature of Definition~\ref{def:ext} and Lemma~\ref{lem:ext} is that they involve the state $\rho_{\mathrm{Ext}(X,Y)YE}$ as opposed to $\rho_{\mathrm{Ext}(X,Y)E}$.  This is what makes the extractor a \emph{strong} extractor: the seed randomness $Y$ remains random and is uncorrelated with the extractor output and any information held by an adversary (to within distance $\epext+2\epsmooth$). This is why the seed is not considered to be consumed in the process.

In this work we use Toeplitz matrices for the extraction, a procedure that was also followed in~\cite{Liu2} (see also~\cite{Ma13}). A Toeplitz matrix $T$ is one for which $T_{i,j}=T_{i+1,j+1}$ whenever both elements exist.  Thus, to choose a random Toeplitz matrix we can randomly choose the first row and column, from which all other entries are fixed.  The set of all binary Toeplitz matrices form a set of two-universal hash functions~\cite{MNT,Krawczyk}, and hence can be used as quantum-proof $(k,\epsilon_{\mathrm{EXT}})$-strong randomness extractors with $\epsilon_{\mathrm{EXT}}=2^{-(k-m)/2}$~\cite{Renner}. Given a raw string of length $n$ with min-entropy $k$, we can use a randomly chosen $m\times n$ binary Toeplitz matrix $T$ to extract the randomness by multiplication with the raw string modulo 2.  Given that $\epsilon_{\mathrm{EXT}}$ drops exponentially in $k-m$ one can ensure a small error without sacrificing much in the size of the output.

This theoretically simple computation runs into difficulty due to the size of the data considered.  To reduce the computational requirements we perform several simplifications that reduce the problem size and increase its efficiency. These simplifications give the largest advantage for the experiment with space-like separation, and we give the figures used for this case in the remainder of this section.

Firstly, in the protocol we replace Bob's output, $B_i$, with $0$ in the generation rounds. This means that none of the entropy accumulated during the protocol is contained within the outputs recorded by Bob during the generation rounds and so we do not need to include these in our extraction. The extraction from ${\bf AB}$ can be performed instead as an extraction from the binary string $v=\mathbf{AB}_{\text{test}}$, where $\mathbf{B}_{\text{test}}$ is the string of outputs recorded by Bob during test rounds. By removing the outputs of Bob in the generation rounds the input to the extractor is roughly half of its original size, making the computation easier.

After this initial simplification we are left with the problem of multiplying an $m\times n = (6.496\times10^9) \times (3.17\times10^{12}$) Toeplitz matrix $T$ to the raw data vector $v$ with a length $n=3.17\times10^{12}$ yielding our extracted randomness ${\bf M}=Tv$ of approximately $6.496\times10^{9}$ bits. To reduce the memory requirements we split this matrix-vector multiplication into several smaller matrix-vector multiplications. For some $l \leq m$ we can split $T$ into $\ceil{n/l}$ blocks of size $m\times l$, each comprising $l$ columns, i.e., $T = (T_{1}\,T_{2}\,\dots T_{\ceil{n/l}})$ where $T_1$ is the matrix consisting of the first $l$ columns of $T$, $T_2$ is the matrix consisting of the second $l$ columns and so on.\footnote{The final block will have fewer columns if $n/l$ is not an integer.} Splitting $v = (v_1\,v_2\,\dots\,v_{\ceil{n/l}})$ also into blocks of the same length $l$ we can rewrite the original matrix vector multiplication as
\begin{equation}
Tv =  T_1 v_1 \oplus T_2 v_2 \oplus \dots \oplus T_{\ceil{n/l}} v_{\ceil{n/l}},
\end{equation}
where $\oplus$ denotes elementwise addition modulo $2$. In our case we split the problem into $1618$ blocks.

Evaluating $T_iv_i$ requires an $m \times l$ matrix-vector multiplication where $m = 6.496\times10^9$ and $l = 1.959 \times 10^9$. By exploiting the structure of Toeplitz matrices we can use fast Fourier transforms (FFTs) to reduce the time complexity of this operation from $O(m^2)$ to $O(m\log m)$~\cite{gohberg1994complexity}. For completeness we now detail the FFT based algorithm.

An $m\times n$ Toeplitz matrix takes the form
\begin{equation}\label{}
T =
\left(
\begin{array}{cccccc}
a_0     & a_{-1}  & a_{-2}& \cdots & a_{-(n-2)}   & a_{-(n-1)} \\
a_1     & a_0     & a_{-1}&\ddots &              & a_{-(n-2)} \\
a_2     & a_1     & a_0&\ddots & \ddots       & \vdots     \\
\vdots  & \vdots  &\vdots& \ddots       & \ddots       & a_{-(n-1)+(m-2)} \\
a_{m-1} & a_{m-2} & a_{m-3}&\cdots & a_{-n+(m-1)} & a_{-(n-1)+(m-1)} \\
\end{array}
\right)\,,
\end{equation}
and so is uniquely specified by the vector $a = (a_{-(n-1)}, a_{-(n-2)}, \dots, a_{m-1})$. Consider the vector $b = (v\,\, 0_m)$ which is the vector $v$ appended with $m$ zeros. Then, the matrix-vector product $Tv$ may be computed using the identity
\begin{equation}\label{eq:FFT}
(c_{n-1}\,\,Tv) = \mathcal{F}^{-1}\left( \mathcal{F}(a) \circledcirc \mathcal{F}(b) \right),
\end{equation}
where $c_{n-1}$ is some $n-1$ dimensional vector, $\mathcal{F}$ denotes the discrete Fourier transform, $\mathcal{F}^{-1}$ denotes its inverse and $\circledcirc$ denotes the elementwise (Hadamard) product of vectors. Equation \eqref{eq:FFT} can be derived by noting that the circular convolution of the two vectors $a$ and $b$ gives $a \star b = (c_{n-1}\,\, Tv)$ and furthermore, the Fourier convolution theorem states that $a \star b = \mathcal{F}^{-1}\left( \mathcal{F}(a) \circledcirc \mathcal{F}(b) \right)$. Finally, noting that each of the blocked matrices $T_i$ is itself a Toeplitz matrix we can perform the FFT technique to speed up each of the $T_iv_i$ computations individually.

To implement the extraction procedure we utilized the {\sc Viking} research cluster and the FFTs were implemented using the {\sc Fftw3} package~\cite{FrigoJo98}. The total computation time was around 249 hours, split across 32 cores and required around 400GB of memory.\footnote{In principle we could have utilized more memory and reduced the number of blocks. However, significant slowdowns occur to the FFT algorithm when the vectors are not stored in a contiguous block of memory.} In summary, by applying a quantum-proof strong extractor (Toeplitz hashing) to the output bit-string ($3.17 \times 10^{12}$ bits) we obtained a shorter bit-string $(6.496 \times 10^{9}$ bits) which is almost indistinguishable from uniform randomness. Finally, the seed of the extractor $a$, which in our implementation has size $m+n-1$ bits, is not expended by the protocol (this is a condition for a \emph{strong} extractor), and can thus be reused for some other purpose. Alternatively, a public (but trusted) source of randomness can be used without compromising security (provided it is unknown to the devices before the protocol is run).

\begin{remark}
We chose to use Toeplitz hashing for the extraction because it was efficient enough to implement on our large output data. In principle, other extractors that require shorter seeds would be desirable to use instead, such as Trevisan's extractor~\cite{trevisan,DPVR}.
\end{remark}

\subsection{Error parameters}\label{app:Error}
Various error parameters feature in the security statements of device-independent randomness expansion. In this section we reprise the discussion of some of the parameters introduced in the main body of the text, starting with a calculation of the completeness error. We use the following theorem, which gives tight bounds for the cumulative distribution function (CDF) of the binomial distribution.

\begin{theorem}[\cite{AD,SZ}]\label{thm:binomial_CDF}
Let $n \in \mathbb{N}$, $p \in (0,1)$ and let $X$ be a random variable distributed according to $X\sim\mathrm{Binomial}(n,p)$. Then, for every
$k = 0,\dots,n-1$ we have
\begin{align*}
C(n,p,k)\leq\mathbb{P}[X\leq k]\leq C(n,p,k+1)\,,
\end{align*}
where
\begin{align*}
  G\left(x,p\right)&=x \ln\frac{x}{p}+(1-x)\ln \frac{1-x}{1-p}\\
  \Phi(x)&=\frac{1}{\sqrt{2 \pi}}\int_{-\infty}^x\mathrm{e}^{-\frac{u^2}{2}}\mathrm{d}u\\
  C(n,p,k)&=\Phi\left(\mathrm{sign}\left(\frac{k}{n}-p\right)\sqrt{2 n G\left(\frac{k}{n},p\right)}\right).
\end{align*}
\end{theorem}

We call an implementation of our DIRNE protocol \emph{honest} if on each round, the state shared between Alice and Bob and the measurements performed for their respective inputs remain the same. In particular, this implies that the CHSH score for each test round is distributed in an i.i.d.\ manner. We can use Theorem~\ref{thm:binomial_CDF} to upper-bound the completeness error of an honest implementation of the protocol.

\begin{corollary}[Completeness error]
Let $\wexp$ be the expected CHSH score achieved by some honest implementation of Protocol DIRNE. Then, the probability that the protocol aborts is no larger than
\begin{equation}\label{eq:completeness_error}
1-C(n,\gamma(1-\wexp),\floor{n\gamma(1-(\wexp-\delta))}),
\end{equation}
where $\delta > 0$ is the selected confidence threshold. In particular, this means the completeness error for the protocol, $\epcomplete$, is no larger than~\eqref{eq:completeness_error}.
\end{corollary}
\begin{proof}
 For an honest implementation of the protocol the score registers $U_i$ are distributed according to
\begin{equation*}
\mathbb{P}[U_i=u] = \begin{cases}
\gamma (1-\wexp) & \text{for }u=0 \\
\gamma \wexp & \text{for }u=1, \\
(1-\gamma) &\text{for }u=\perp
\end{cases}.
\end{equation*}
for each $i=1,\dots,n$. Recall that the protocol aborts when $|\{U_i:U_i=0\}|>n\gamma(1-(\wexp-\delta))$. The quantity $|\{U_i:U_i=0\}|$ is a random variable following the binomial distribution $\mathrm{Binomial}(n,\gamma(1-\wexp))$. Applying Theorem~\ref{thm:binomial_CDF}, we find the probability that the protocol aborts is upper bounded by
\begin{align*}
\mathbb{P}\bigl[|\{U_i:U_i=0\}|>n\gamma(1-(\wexp-\delta))\bigr]&\leq \mathbb{P}\bigl[|\{U_i:U_i=0\}|>\floor{n\gamma(1-(\wexp-\delta))}\bigr] \\
&=1-\mathbb{P}\bigl[|\{U_i:U_i=0\}|\leq \floor{n\gamma(1-(\wexp-\delta))}\bigr]\\                                                             &\leq 1-C(n,\gamma(1-\wexp),\floor{n\gamma(1-(\wexp-\delta))}).\qedhere
\end{align*}
\end{proof}

The second parameter in the security definitions of randomness expansion is the \emph{soundness error} $\epsound$. Following~\cite{BRC}, we may bound the soundness error of our randomness expansion protocol as
\begin{equation}
\epsound \leq \max \{\epeat, \epext + 2 \epsmooth\}.
\end{equation}
The first term $\epeat$, which we refer to as the \emph{device-independent error}, is roughly our tolerance of encountering `lucky' adversaries. In Theorem~\ref{Thm:EAT} we see that the error terms depend explicitly on the probability, $p_{\Omega}$, that the protocol does not abort. Since we work in the device-independent setting we cannot assume to know the value of $p_{\Omega}$ (this can be set by the adversary). Instead, we can replace $p_{\Omega}$ with $\epeat$ in the error terms and consider the two possible scenarios. If $p_{\Omega} \geq \epeat$ then by making the replacement $p_{\Omega} \mapsto \epeat$ the error terms only increase and we have genuine lower bound on the accumulated entropy. Otherwise, we have $p_{\Omega} \leq \epeat$. However, the chances that the protocol passes but the entropy bound is not valid is less than $\epeat$ which can be chosen to be negligibly small. In summary, either the protocol aborts with probability greater than $1-\epeat$ or Theorem~\ref{Thm:EAT} with $p_{\Omega} \mapsto \epeat$ gives a valid lower bound on the entropy we produce.

The second term comes from the result of applying a strong extractor to the raw output of our experiment (cf.\ Lemma~\ref{lem:ext}). It consists of the smoothing parameter $\epsmooth$ and the extractor error $\epext$. The extractor error has far less impact on our results than the smoothing error
and for our calculations we choose to set $\epext=10^{-5}\,\epsmooth$. The smoothing parameter, $\epsmooth$, plays a more prominent role in both the soundness error and the lower bound on quantity of certifiable smooth min-entropy (cf.\ Theorem~\ref{Thm:EAT}). By reducing the smoothing parameter we may decrease the soundness error at the expense of smaller bound on the entropy certified by the entropy accumulation theorem. This relationship is the same for the device-independent error. To simplify our calculations we set $\epeat=\epext+2\epsmooth$.

\subsection{Entropies}
In the entropy accumulation theorem, a lower bound on the min-entropy
is obtained by way of the $\alpha$-entropy. In this work, we take
$\alpha\in(1,2]$, although some definitions and identities hold for
wider ranges of $\alpha$. We begin with two definitions
of the $\alpha$-entropy.

\begin{definition}[R\'enyi $\alpha$-entropy~\cite{MDSFT}]
We define the $\alpha$-entropy by means of the sandwiched R\'enyi divergence,
\begin{align*}
D_\alpha(\rho||\sigma)=\frac{1}{\alpha-1}\log \tr\left[\left(\sigma^{\frac{1-\alpha}{2\alpha}} \rho \sigma^{\frac{1-\alpha}{2\alpha}}\right)^\alpha\right],
\end{align*}
where $\rho$ and $\sigma$ are positive semidefinite operators on the same Hilbert space. Then, on a bipartite state $\rho_{AB}$,
\begin{align*}
H^\uparrow_\alpha(A|B)_{\rho}&=\sup_{\sigma_B} -D_\alpha(\rho_{AB}||\mathbb{1}_A\otimes \sigma_B)\\
H_\alpha(A|B)_{\rho}&= -D_\alpha(\rho_{AB}||\mathbb{1}_A\otimes \rho_B)\,.
\end{align*}
When it is clear from the context we will omit the subscript $\rho$.
\end{definition}
The smooth min-entropy can be then be lower bounded by~\cite{TCR,DFR,Tomamichel_book}
\begin{align}
\label{eq:minvalpha}
H^{\epsmooth}_{\min}(A|B)_\rho \geq H^\uparrow_\alpha(A|B)_\rho - \frac{1}{\alpha-1}\log\frac{1}{1-\sqrt{1-\epsmooth^2}}.
\end{align}
Additionally, we need to condition upon observing a pass-event $\Omega$ on a classical register. The state can be written as $\rho=p_{\Omega} \rho_{\Omega} + (1-p_{\Omega}) \rho_{\Omega^\perp}$. The entropy of the conditioned state and unconditioned state can be related by
\begin{align}
\label{eq:Conditioning}
H^\uparrow_\alpha(A|B)_{\rho_\Omega}\geq H^\uparrow_\alpha(A|B)_\rho-\frac{\alpha}{\alpha-1}\log \frac{1}{p_{\Omega}}.
\end{align}
This result is proven in Lemma~B.5 of~\cite{DFR}. We can combine  Equation~\ref{eq:minvalpha} and Equation~\ref{eq:Conditioning} to obtain
\begin{align*}
H^{\epsmooth}_{\min}(A|B)_\rho\geq H^\uparrow_\alpha(A|B)_\rho - \frac{\alpha}{\alpha-1}\log\left(\frac{1}{p_{\Omega}(1-\sqrt{1-\epsmooth^2})}\right),
\end{align*}
where we have used that $\alpha>1$.

\subsection{Entropy Accumulation}
In this section we outline further improvements to the theory of~\cite{DF,DFR} which produce improved rates and make randomness expansion experimentally accessible\footnote{The quantum probability estimation framework~\cite{QPE} provides an alternative way to prove bounds on the amount of extractable randomness, but we do not use this in this work.}. We consider a set of channels, $\lbrace \mathcal{M}_i\rbrace_i$, for $i\in\lbrace 1,\dots,n\rbrace$ with $\mathcal{M}_i: \S(R_{i-1})\to \S(\A_i \B_i \X_i R_i)$, sometimes dubbed \textit{EAT channels}\footnote{The definition of EAT channels supplied here is slightly less general than that in~\cite{DFR} but the proof is easily adapted to the more general definition.}.

\begin{definition}[EAT channels]\label{def:EAT-channels}
Let $\{\M_i\}_{i}$ be a collection of trace preserving completely positive maps with $\M_i: \S(R_{i-1}) \to \S(A_iB_iX_iR_i)$ for each $i=1,\dots,n$. Let $\rho_{\mathbf{ABX}R_nE}=(\M_n\circ\dots\circ\M_1)\rho_{R_0E}$ be a state obtained by sequential application of the channels to an initial state $\rho_{R_0E}$. Then the collection $\{\M_i\}_i$ is called a set of EAT channels if for each $i=1,\dots,n$ the following both hold:
\begin{enumerate}
\item The systems $A_i$, $B_i$ and $X_i$ are all finite dimensional and classical and the state of the register $X_i$ is a deterministic function of $A_i$ and $B_i$. The system $R_i$ may be arbitrary.
\item For any initial state $\rho_{R_0E} \in \S(R_0E)$, the final state $\rho_{\mathbf{A}\mathbf{B}E}=\tr_{\mathbf{X}R_n}\left[\rho_{\mathbf{ABX}R_nE}\right]$ satisfies the conditional independence constraints $I(A_1^{i-1}:B_i|B_1^{i-1}E)=0$.
\end{enumerate}
[$I(A:B|C)$ is the conditional mutual information.]
\end{definition}

The collection of channels in the above definition represent the sequential interaction with the devices that occurs during the first four steps of the protocol from the main-text. In particular, as the inputs to our devices are chosen independently, the second condition trivially holds for any collection of channels we could use to implement the protocol. Similarly, the nature of the systems imposed by the first condition is also satisfied by any channel implementing our protocol, $A_i$ are the outputs for round $i$, $B_i$ are the inputs for round $i$ and $X_i$ is the recorded score for round $i$. We may also view EAT channels as \emph{quantum instruments}, i.e., for the channel $\M_i$ there is a collection of trace non-increasing completely positive maps $\{\M_i^{ab}\}_{ab}$ with $\M_i^{ab}:\S(R_{i-1}) \to \S(R_i)$ such that for a state $\rho \in \S(R_{i-1})$ we have
\begin{equation}\label{eq:EAT-instruments}
\M_i(\rho) = \sum_{a,b} \ketbra{a}{a}_{A_i} \otimes \ketbra{b}{b}_{B_i} \otimes\ketbra{x(a,b)}{x(a,b)}_{X_i} \otimes\M_i^{ab}(\rho),
\end{equation}
where $x(a,b)$ is a deterministic function of the inputs and outputs and $\tr[\M_i^{ab}(\rho)] = p(a,b)$.

\begin{definition}[Frequency distribution]
Let $\rho_{\mathbf{X} R}$ be a CQ state. Then we define
\begin{align*}
\mathrm{freq}_{\mathbf{X}}(x)=\frac{|\lbrace i\in\lbrace 1,\dots,n\rbrace: \X_i=x\rbrace|}{n}\,,
\end{align*}
and use $\freq_{\mathbf{X}}$ to refer to the induced probability distribution.
\end{definition}

Let $\P$ denote the set of all probability distributions over the possible outputs of the score register $X$. Given a set of channels $\mathfrak{S}$ whose outputs include a register $X$, we define $\Q_{\mathfrak{S}}$ to be the set of probability distributions over the possible outputs of the score register $X$ that could arise from the application of a channel from the set $\mathfrak{S}$ to some state. [Note that no other properties are assumed about $\mathfrak{S}$, so different channels in the set could have different input spaces, for instance.]
\begin{definition}[Achievable distributions]
For a set of channels, $\mathfrak{S}$ whose outputs include a register $X$, the set of achievable score distributions is defined by
\begin{equation}\label{eq:Q}
\Q_{\mathfrak{S}} = \left\{p_X: \exists\, \M\in\mathfrak{S},\rho \text{ such that } \M(\rho)_X = \sum_x p_{X}(x) \ketbra{x}{x}\right\},
\end{equation}
where $\M(\rho)_X$ denotes the state after applying $\M$ and tracing out all systems apart from the score system $X$.
\end{definition}

Given a distribution $q \in \Q_{\mathfrak{S}}$, we can identify the set $\Gamma_\mathfrak{S}(q)$ of states and channels that can achieve $q$, i.e., 
\begin{align}\label{eq:Gammaq}
\Gamma_\mathfrak{S}(q)=\left\{ (\omega_{RE},\M): \M(\omega_{RE})_{X}=\sum_{x} q(x)\ketbra{x}{x}\right\},
\end{align}
where the channel $\M\in\mathfrak{S}$ acts only on the system $R$ and the system $E$ is arbitrary (it can represent a quantum system held by an adversary).

We now define round-by-round lower bounds on the von Neumann entropy,
$f(q)$, known as min-tradeoff functions.
\begin{definition}[Min-tradeoff functions]
\label{def:MinTrade}
Given a set of channels $\mathfrak{S}$, whose outputs contain registers $A$, $B$ and $X$, a function from $\Q_\mathfrak{S}$ to $\mathbb{R}$ is called a \emph{rate function for $\mathfrak{S}$} if for all $q\in\Q_{\mathfrak{S}}$ we have
\begin{align}\label{eq:rate_function_def}
\rate(q)\leq \inf_{(\omega_{RE},\M)\in\Gamma_\mathfrak{S}(q)} H(\A|\B E)_{\M(\omega)}\,,
\end{align}
where the channel $\M$ acts on the $R$ system of $\omega$ (note that different channels in $\mathfrak{S}$ can have input spaces of different sizes).
When clear from the context we will omit mention of the set of
channels and just call such a function a rate function.  We refer to
affine rate functions as \emph{min-tradeoff functions}, and for these we
reserve the symbol $f$.
\end{definition}

\begin{remark}
This definition differs slightly from the definitions of min-tradeoff functions in~\cite{DFR,DF}. In these works, the min-tradeoff function is defined for an individual channel, and in subsequent proofs and theorem statements, it is specified that the function under consideration is a min-tradeoff function for each channel used in a protocol. We define our min-tradeoff functions with respect to a set of channels from the start, as it simplifies the writing of some of our proofs.


\end{remark}

Rate and min-tradeoff functions are lower bounds to the worst-case von Neumann entropy for any \textit{individual} round, given in terms of the score. Whilst rate functions are only defined on the set $\Q_\mathfrak{S}$ (the infimum is trivial for distributions not in $\Q_\mathfrak{S}$), we allow the domain of the min-tradeoff functions (affine rate functions) to be naturally extended to all probability distributions on $X$, i.e., to $\P$.

In order to use our min-tradeoff functions with the entropy accumulation theorem we require knowledge of several of their properties. Namely, if $f$ is a min-tradeoff function then we define:
\begin{itemize}
\item Maximum over all probability distributions:
\begin{equation}
\Max(f) = \max_{p \in \P} f(p).
\end{equation}
\item Minimum over distributions achievable by sending quantum states through channels in $\mathfrak{S}$:
\begin{equation}
\Min_{\Q_\mathfrak{S}}(f) = \inf_{p \in \Q_\mathfrak{S}} f(p).
\end{equation}
\item Variance
\begin{equation}
\varq(f) = \sum_x p(x) (f(\delta_x)-\mathbb{E}[f(\delta_x)])^2\,.
\end{equation}
\end{itemize}
In the final definition, $\delta_x$ denotes the probability distribution with $p(x)=1$. As $f$ is an affine function, the final definition is the statistical variance of the function $g(x)=f(\delta_x)$. In~\cite{DF} $\var(f)$ is defined to be $\var(f)=\sup_{p\in\mathcal{Q_\mathfrak{S}}} \varq(f)$, which gives a worst-case variance over all channels in $\mathfrak{S}$.

In order to explain our modification to the entropy accumulation theorem we reproduce a preliminary step in the proof, beginning Proposition V.3 in~\cite{DF}. Let $f$ be a min-tradeoff function for a set of EAT channels $\mathfrak{S}=\{\M_i\}_i$. Let $\rho_{\textbf{ABX}E}$ be the state generated by a sequential application of these channels as in Definition~\ref{def:EAT-channels}. We additionally apply for each round a channel $\mathcal{D}_i:\X_i\to \X_i \Csys_i$, which encodes the min-tradeoff function directly. In particular, we define
\begin{align*}
\mathcal{D}_i(\ketbra{x}{x}_{\X_i})=\ketbra{x}{x}_{\X_i}\otimes \tau(x)_{\Csys_i}\,,
\end{align*}
where $\tau(x)_{\Csys_i}$ is defined so that $H_\alpha(\tau(x)_{\Csys_i})=\Max(f)-f(\delta_x)$ for some fixed $\alpha$. We then apply the channels to the state to get $\rho_{\mathbf{ABXD} E}=\mathcal{D}_n\circ\dots\circ \mathcal{D}_1(\rho_{\mathbf{A B X} E})$. In a real protocol, the acceptance of a state will be conditioned upon some success event $\Omega$, determined by the values observed on the $\mathbf{X}$ registers (such as a minimal threshold for the CHSH score). Let $\rho_{\mathbf{ABX} E|\Omega}$ be the state conditioned on passing. It can then be shown that
\begin{align*}
H_\alpha^\uparrow(\textbf{A}|\textbf{B} E)_{\rho_{|\Omega}}\geq H_\alpha(\textbf{A} \mathbf{D}|\textbf{B} E)_{\rho_{|\Omega}}-n \mathrm{Max}(f)+n h,
\end{align*}
where $h$ is defined to be $\inf_{p\in\Omega} f(p)$ (see Equation~24 in~\cite{DF}). Then, using Equation~\ref{eq:Conditioning}, we obtain
\begin{align*}
H_\alpha^\uparrow(\textbf{A}|\textbf{B} E)_{\rho_{|\Omega}}\geq H_\alpha(\textbf{A} \mathbf{D}|\textbf{B} E)_{\rho}-\frac{\alpha}{\alpha-1}\log\frac{1}{p_{\Omega}}-n \mathrm{Max}(f)+n h,
\end{align*}
which relates the bound for the state conditioned on $\Omega$ to the
unconditioned state.

We now seek a lower bound to $H_\alpha(\textbf{A} \mathbf{D}|\textbf{B} E)$ in terms of the von Neumann entropy. The first step involves using the chain rules of~\cite{DF,DFR} to obtain
\begin{align*}
H_\alpha(\textbf{A} \mathbf{D}|\textbf{B} E)_{\rho}\geq \sum_i \inf_{\omega_{R_{i-1}R}} H_\alpha(\A_i \Csys_i|\B_i R)_{\mathcal{D}_i\circ\M_i(\omega)}.
\end{align*}
From this, we apply a continuity bound relating the von Neumann entropy to the $\alpha$-entropies which yields
\begin{align}
\label{eq:Error}
H_\alpha(\textbf{A} \mathbf{D}|\textbf{B} E)_{\rho}\geq \sum_i \inf_{\omega_{R_{i-1}R}}&( H(\A_i|\B_i R)_{\mathcal{D}_i\circ\M_i(\omega)}+H(\Csys_i|\X_i)_{\mathcal{D}_i\circ\M_i(\omega)}+\nonumber \\&-(\alpha-1)V(\A_i \Csys_i|\B_i R)_{\mathcal{D}_i\circ\M_i(\omega)}-(\alpha-1)^2 K_\alpha(\A_i \Csys_i|\B_i R)_{\mathcal{D}_i\circ\M_i(\omega)})\,,
\end{align}
where for CQ states\footnote{Note that the form of the error terms given in~\cite{DF} differs slightly in the dimension dependence. Since $A$ is  classical, we use the relevant statement in~\cite[Corollary III.5]{DF} to give our expression for $V$ and the statement in the proof of~\cite[Proposition~V.3]{DF} to give our expression for $K$, whereas in~\cite{DF} the statement given allows $A$ to be quantum.}
\begin{align}
V(\A_i \Csys_i|\B_i R)_{\mathcal{D}_i\circ\M_i(\omega)}&\leq \frac{\ln 2}{2} \left(\log(1+2 d_A) +\sqrt{2+\varq(f)}\right)^2\equiv V(f,p),\label{eq:V}\\
K_\alpha(\A_i \Csys_i|\B_i R)_{\mathcal{D}_i\circ\M_i(\omega)}&\leq \frac{1}{6(2-\alpha)^3 \ln 2}2^{(\alpha-1)(\log d_A +\mathrm{Max}(f)-\Min_{\mathcal{Q}}(f))}\ln^3 \left(2^{\log d_A +\mathrm{Max}(f)-\Min_{\mathcal{Q}}(f)}+\mathrm{e}^2\right)\nonumber\\
&\equiv K_\alpha(f)\,.\label{eq:K}
\end{align}
So far we have not deviated from the argument of~\cite{DF},
in which they then employ the following simplifications
\begin{align*}
V(f,q)\leq V(f) \equiv \frac{\ln 2}{2} \left(\log(1+2d_A) +\sqrt{2+\var(f)}\right)^2\,,
\end{align*}
where $\var(f)=\sup_{p\in \mathcal{Q}_{\mathfrak{S}}} \varq(f)$, and
\begin{align*}
H(\A_i|\B_i R)_{\mathcal{D}_i\circ\M_i(\omega)}+H(\Csys_i|\X_i)_{\mathcal{D}_i\circ\M_i(\omega)}&= H(\A_i|\B_i R)_{\mathcal{D}_i\circ\M_i(\omega)}+\Max(f)-f(p)\\
&\geq \Max(f)\,,
\end{align*}
which is seen by noting that
$H(\A_i|\B_i R)_{\mathcal{D}_i\circ\M_i(\omega)}=H(\A_i|\B_i
R)_{\M_i(\omega)}\geq f(p)$ by the definition of the channel
$\mathcal{D}$ and the min-tradeoff function (where $p$ is assumed to
be the distribution on the $\X_i$ register).  These simplifications
render the error term `statistics agnostic', in the sense that
Equation~\ref{eq:Error} becomes
\begin{align*}
H_\alpha(\textbf{A} \mathbf{D}|\textbf{B} E)_{\rho}&\geq \sum_i(\Max(f)-(\alpha-1)V(f)-(\alpha-1)^2 K_\alpha(f))\\
&= n(\Max(f)-(\alpha-1)V(f)-(\alpha-1)^2 K_\alpha(f))\,,
\end{align*}
which is independent of the actual statistics observed and this incurs a
significant loss of entropy in the regimes of interest to us. Instead, by not
taking the supremum in $V(f,q)$ and by permitting a tighter bound to
$H(\A_i|\B_i R)_{\M_i(\omega)}$, we find an improvement.

Given a rate function, we can instead use the bounds $H(\A_i|\B_i R)_{\mathcal{D}_i\circ\M_i(\omega)}=H(\A_i|\B_i
R)_{\M_i(\omega)}\geq \rate(p)$. Defining $\Delta(f,p):=\rate(p)-f(p)$, this leads instead to
\begin{equation}\label{eq:infimum_bound}
H_\alpha(\textbf{A} \mathbf{D}|\textbf{B} E)_{\rho}\geq n \left(\Max(f)+\inf_{p \in \Q_{\mathfrak{S}}}\left(\Delta(f,p)-(\alpha-1)V(f,p)\right)-(\alpha-1)^2 K_\alpha(f)\right).
\end{equation}

Summarizing, we have the following modified version of the EAT.
\begin{theorem}\label{thm:new_eat}
  Let $\epsmooth\in(0,1)$, $\alpha\in(1,2)$, $r\in\mathbb{R}$, $f$ be a min-tradeoff function for a set of EAT channels $\mathfrak{S}=\{\M_i\}_i$ and $\rho_{\mathbf{ABX}E}$ be the state generated by a sequential application of these channels as in Definition~\ref{def:EAT-channels}. Then for any event $\Omega$ on $\mathbf{X}$ that implies $f(\freq_{\mathbf{X}})\geq r$ we have
  \begin{align*}
H_{\min}^{\epsmooth}(\mathbf{A}|\mathbf{B}E)_{\rho_{\mathbf{ABX}E|\Omega}}>&n r-\frac{\alpha}{\alpha-1}\log\left(\frac{1}{p_\Omega(1-\sqrt{1-\epsmooth^2})}\right)+\\&n\inf_{p \in \Q_{\mathfrak{S}}}\big(\Delta(f,p)-(\alpha-1)V(f,p)\big)-n(\alpha-1)^2K_\alpha(f),
  \end{align*}
where $V(f,p)$ and $K_\alpha(f)$ are defined in~\eqref{eq:V} and~\eqref{eq:K}.
\end{theorem}

The presence of the $\Delta(f,p)$ term yields a two-fold
advantage. Not only is it a positive contribution to the final
entropy, but it effectively constrains the probability distribution on
the error terms to be close to the actual value of the observed
probability distribution.  While we still take a worst-case
optimisation over distributions, we do it in a way that depends
explicitly on the min-tradeoff function so that our optimisation is no
longer independent of the actual statistics. By taking min-tradeoff
functions that are tangent to some convex rate function, $\Delta(f,p)$
will grow to dominate the error terms if $q$ varies too far from the
tangent point. To ensure convergence of the minimisation to a globally optimal point and hence a true lower bound in~\eqref{eq:infimum_bound}, we show
that, given certain conditions, computing the infimum is a convex optimisation problem. In Section~\ref{sec:Nonlocalgames}, we show that these conditions are applicable to device-independent cryptography with non-local games.

\begin{lemma}[Convexity of the objective]
\label{lem:convexopt}
Let $f$ be a min-tradeoff function, $\rate$ be a convex rate function, $Q_{\mathfrak{S}}$ be a convex set and $\alpha \in (1,2)$. Then,
\begin{equation}
\inf_{p \in \Q_{\mathfrak{S}}}\left(\Delta(f,p)-(\alpha-1)V(f,p)\right)
\end{equation}
is a convex optimisation problem.
\end{lemma}
\begin{proof}
  The first term $\Delta(f,p)=\rate(p)-f(p)$ is convex since $\rate(p)$ is convex and $f(p)$ is affine. The second term can be written as $V(f,p)=c(d+\sqrt{2+\varq(f)})^2$ where $c$ and $d$ are positive constants. We note that $\varq(f)$ is the variance of the function $g(x)=f(\delta_x)$. Variance is concave in the probability distribution and the square root is an increasing concave function, hence $\sqrt{2+\varq(f)}$ is concave. Expanding $c(d+\sqrt{2+\varq(f)})^2$ we find that it is a non-negatively weighted sum of concave functions, hence also concave. This implies that $-(\alpha-1)V(f,p)$ is convex. Therefore the objective function is a non-negatively weighted sum of functions that are convex in $p$ and hence is itself a convex function. By assumption, the set being optimised over $Q_{\mathfrak{S}}$ is convex and minimising a convex function over a convex set is a convex optimisation problem.
\end{proof}

\subsection{Spot-checking protocols}
The spot-checking format of the protocol enforces a particular structure in the probability distributions over the score register that can be produced by the protocol. In particular, the expected distributions over $X_i$ for each round $i$ must now take the form
\begin{equation}\label{eq:structured_distributions}
\mathbb{P}[X_i = x] = \begin{cases}
\gamma q(x) & \text{for }x\neq\perp \\
(1-\gamma) &\text{for }x=\perp
\end{cases},
\end{equation}
where $q$ is a distribution over the test-round scores for round $i$. To capture this structure in the construction
of our min-tradeoff functions we follow~\cite[Section 5]{DF}, by first defining a min-tradeoff function that only takes the statistics of a test round and then extending the domain of the function to include the no-test symbol $\perp$. Formally, we
define infrequent sampling channels as follows.
\begin{definition}[Infrequent sampling channel]
\label{def:infreq}
An infrequent sampling channel, $\M_i: R_{i-1}\to \A_i \B_i \X_i R_i$ is a channel that decomposes according to
\begin{equation}
\M_i(\cdot)=\gamma \M^{\mathrm{test}}_i(\cdot)+ (1-\gamma)\M_i^{\mathrm{gen}}(\cdot)\otimes \ketbra{\bot}{\bot}_{X_i}
\end{equation}
where $\M^{\mathrm{test}}_i: R_{i-1}\to \A_i \B_i \X_i R_i$ with $\bra{\perp}\M^{\mathrm{test}}_i(\rho_{R_{i-1}E})_{X_i} \ket{\perp} = 0$ for any quantum state $\rho_{R_{i-1}E}$  and $\M^{\mathrm{gen}}_i: R_{i-1}\to \A_i \B_i R_i$.
\end{definition}
Thus, on round $i$ the channel $\M^{\mathrm{gen}}_i$ occurs with probability $1-\gamma$ and this event is recorded with $\bot$ in the $X_i$ register and otherwise a test channel acts with the score is written to the $X_i$ register. Now let $\mathfrak{S}$ be the set of channels from which $\M^{\mathrm{test}}$ can be drawn, and $\mathfrak{S}^\gamma$ be the set of associated infrequent sampling channels, where we suppress the round index for conciseness. We define the associated set of infrequent sampling distributions $\Q_{\mathfrak{S}}^\gamma$ as
\begin{equation}
\Q^\gamma_{\mathfrak{S}} = \left\{q_X: q_X(\bot)=(1-\gamma) \text{ and }q_X(x)=\gamma q'_X(x)\text{ otherwise, where }q_X'\in Q_{\mathfrak{S}}\right\}.
\end{equation}
Note that $\Qgamma_{\mathfrak{S}}$ inherits convexity from $\Q_{\mathfrak{S}}$ if the latter is convex. We also define $\Gamma^\gamma_\mathfrak{S}$ where $\M$ is an infrequent sampling channel according to
\begin{align}\label{eq:Gammaq2}
\Gamma^\gamma_\mathfrak{S}(q)=\left\{ (\omega_{R E},\M): \M(\omega_{RE})_{X}=\sum_{x} q(x)\ketbra{x}{x}\text{ and }\M^{\mathrm{test}}\in\mathfrak{S}\right\}.
\end{align}

The main result of this section is the following.
\begin{lemma}
Let $g$ be an affine function satisfying
\begin{equation}\label{eq:gq}
g(q)\leq\inf_{(\omega_{RE},\M)}\left\{H(A|BE)_{\M(\omega_{RE})}
:
\M^{\mathrm{test}}(\omega_{RE})_{X}= \sum_{x \neq \perp}q(x)\ketbra{x}{x}\right\},
\end{equation}
for all $q \in \Q_{\mathfrak{S}}$ and with $\M,\,\M^{\mathrm{test}}$ defined as in Definition~\ref{def:infreq}. Then for $c_{\perp} \in \mathbb{R}$, the function
\begin{equation}
\label{eq:Fp}
f(\delta_x)= \begin{cases}
\frac{1}{\gamma}g(\delta_x)+(1-\frac{1}{\gamma})c_{\perp} &\text{ if }x\neq\perp \\
c_\perp &\text{ if }x=\perp
\end{cases}
\end{equation}
is a min-tradeoff function for $\mathfrak{S}^\gamma$.
\end{lemma}
\begin{proof}
  Because of the structure of infrequent sampling channels the set $\Gamma^\gamma_{\mathfrak{S}}(p)$ (cf.\ \eqref{eq:Gammaq}) is non-empty if and only if $p \in \Qgamma_\mathfrak{S}$. Moreover, the infimum in the definition of min-tradeoff functions (cf.\ \eqref{eq:rate_function_def}) is finite only when $p \in \Qgamma_\mathfrak{S}$. In turn, we may restrict our attention to distributions of the form $\Qgamma_\mathfrak{S}$ when defining a min-tradeoff function. Now, for $p \in \Qgamma_\mathfrak{S}$ we have
\begin{align*}
\sum_x p(x) f(\delta_x)&=(1-\gamma)c_{\perp}+\gamma \sum_{x\neq\perp} q(x) \left(\frac{1}{\gamma}g(\delta_x) +\left(1-\frac{1}{\gamma}\right) c_{\perp}\right)\\
&=(1-\gamma) c_{\perp} +\sum_{x\neq\perp} q(x) g(\delta_x) - (1-\gamma)c_{\perp}\\
&=g(q)\,,
\end{align*}
with $q \in \Q_\mathfrak{S}$. Finally, we have
\begin{align*}
f(p) &= g(q) \\
&\leq \inf_{(\omega_{RE},\M)}\left\{H(A|BE)_{\M(\omega_{RE})}
:
\M^{\mathrm{test}}(\omega_{RE})_{X}= \sum_{x \neq \perp}q(x)\ketbra{x}{x}\right\} \\
&= \inf_{(\omega_{RE},\M)}\left\{H(A|BE)_{\M(\omega_{RE})}
:
\M(\omega_{RE})_{X}= \sum_{x \neq \perp}\gamma q(x)\ketbra{x}{x} + (1-\gamma) \ketbra{\perp}{\perp}\right\} \\
&= \inf_{(\omega_{RE},\M)\in\Gamma_{\mathfrak{S}}^\gamma(p)} H(\A|\B E)_{\M(\omega_{RE})}\,,
\end{align*}
where the final two lines hold because $\M$ is restricted, by spot-checking, to be an infrequent sampling channel. The function $f$ is therefore a min-tradeoff function for $\mathfrak{S}^\gamma$.
\end{proof}

We obtain expressions for $V(f,p)$ and $K_\alpha(f)$ by noting that
\begin{align*}
\Max(f)&=\max\left[\frac{1}{\gamma}\Max(g)+\left(1-\frac{1}{\gamma}\right)c_{\perp},\,c_{\perp}\right],\\
\Min_{\Qgamma_\mathfrak{S}}(f)&=\Min_{\Q_{\mathfrak{S}}}(g),\\
\varq(f)&\leq \sum_{x\neq\perp}\frac{q(x)}{\gamma}(c_{\perp}-g(\delta_x))^2.
\end{align*}
The first expression follows from the definition of $\Max(f)$, and the second from the fact that for any $p \in \Qgamma_\mathfrak{S}$ we have that $f(p) = g(q)$ for some $q \in \Q_\mathfrak{S}$. For $\varq(f)$, we explicitly calculate
\begin{align*}
\varq(f)&=\sum_x p(x)f(\delta_x)^2-f(p)^2\\
&=\sum_{x\neq\perp} \gamma q(x)\left(\frac{1}{\gamma}g(\delta_x)+\left(1-\frac{1}{\gamma}\right)c_{\perp}\right)^2+(1-\gamma)c_{\perp}^2-g(q)^2\,,
\end{align*}
where we have used that $f(p)=g(q)$ and substituted the explicit forms of $p(x)$ and $f(\delta_\bot)$. We can expand the term in the sum
\begin{align*}
\sum_{x\neq\perp} \gamma q(x)\left(\frac{1}{\gamma}g(\delta_x)+\left(1-\frac{1}{\gamma}\right) c_{\perp}\right)^2&=\sum_{x\neq\perp} \gamma q(x)\left(\frac{1}{\gamma^2}g(\delta_x)^2+2\frac{1}{\gamma}\left(1-\frac{1}{\gamma}\right)g(\delta_x)c_{\perp}+\left(1-\frac{1}{\gamma}\right)^2 c_{\perp}^2\right)\\
&=\sum_{x\neq\perp} \frac{q(x)}{\gamma} g(\delta_x)^2+2\left(1-\frac{1}{\gamma}\right)g(q)c_{\perp}+\frac{(\gamma-1)^2}{\gamma}c_{\perp}^2\,.
\end{align*}
Inserting the final line into the formula for $\varq(f)$ we can rearrange the expression to get
\begin{align*}
\varq(f)=\frac{1}{\gamma}\sum_{x\neq\perp} q(x)\left(c_{\perp}- g(\delta_x)\right)^2- (c_{\perp}-g(q))^2.
\end{align*}
As the final term is strictly negative we arrive at the desired result.

\subsection{Non-local games and convexity}
\label{sec:Nonlocalgames}
We have thus far focused on the EAT without specialising to non-local games (and in particular the CHSH game). Our notation has accordingly been consistent with the previous literature on the EAT~\cite{DFR,DF}. Unfortunately, there is some incompatibility between this notation and the standard notation for non-local games. In the subsequent we will alter the notation,  we replace $\textbf{X}$ with a new score register $\textbf{U}$, $\textbf{A}$ with the joint outputs $\textbf{AB}$, and $\textbf{B}$ with the joint inputs $\textbf{XY}$. In short, we make the following substitutions, $\textbf{X}\to \textbf{U}$, $\textbf{A}\to \textbf{AB}$ and $\textbf{B}\to \textbf{XY}$.

In what follows, we suppress the round index, since we will be interested in arbitrary states and measurements, and the classical registers for any round are isomorphic to those on any other. We define the set of quantum channels consistent with the no-signalling conditions. To do so, we partition the Hilbert space $R$ into Alice's and Bob's parts $R_{A}$ and $R_{B}$. We rewrite~\eqref{eq:EAT-instruments} in terms of the new notation by substituting $x\to u$, $a\to ab$ and $b\to xy$. Any channel $\M$ that is used in the protocol can then be decomposed as
\begin{align}\label{eq:EAT-instruments2}
\M(\rho_{R_AR_B}) = \sum_{a,b,x,y} \ketbra{abxy}{abxy}_{ABXY}\otimes\ketbra{u(a,b,x,y)}{u(a,b,x,y)}_U\otimes\M^{abxy}(\rho_{R_{A}R_{B}}).
\end{align}

\begin{definition}[No-signalling channels, $\mathfrak{N}$]
  Let $\M$ be a channel acting on $R_AR_B$ that implements a map of the form~\eqref{eq:EAT-instruments2}. $\M$ is a \emph{no-signalling channel} if there exists a distribution $p(xy)$ such that
  \begin{align}
\label{eq:EAT_instruments3}
\M^{abxy}(\rho_{R_AR_B})=p(xy) (\mathcal{E}^{xa}\otimes\mathcal{F}^{yb})(\rho_{R_AR_B}),
\end{align}
where, for each $x$, $\{\mathcal{E}^{xa}\}_a$ is an instrument acting on $R_A$ and, for each $y$, $\{\mathcal{F}^{yb}\}_b$ is an instrument acting on $R_B$.

We use $\mathfrak{N}$ to denote the set of all no-signalling channels consistent with some specified $p(xy)$.
\end{definition}

\begin{lemma}\label{lem:Qconvex}
Given a distribution $p(xy)$, the set, $\Q_{\mathfrak{N}}$, of achievable distributions on $U$ using no-signalling channels is convex.
\begin{proof}
It is well-known that the set of distributions $p(ab|xy)$ that can arise as the result of performing measurements on two remote sites, is convex (see, e.g.,~\cite{Cirelson93}). 
We are interested in the convexity of the probability distribution on the (classical) $U$ register. This register contains the outcome of a deterministic function on $A$, $B$, $X$ and $Y$. The elements of the probability distribution $p(u)$ will hence be sums of elements of $p(abxy)$ so that convexity is inherited.
\end{proof}
\end{lemma}

\begin{lemma}[Convexity of optimal rate functions]
\label{lem:ConvexRate}
The function $\rateopt(q):=\inf_{(\omega,\M)\in \Gamma_{\mathfrak{N}}(q)}H(AB|XY R)_{\M(\omega)}$ is a convex function over the set $\mathcal{Q}_\mathfrak{N}$.
\begin{proof}
  For sufficiently small $\epsilon>0$, there exist states $\omega_{R E}$ and $\omega'_{R E}$, channels $\M ,\M'\in\mathfrak{N}$, with distributions on $U$ given by $q=\M(\omega)_U$ and $q'=\M'(\omega')_U$ respectively, such that $H(AB|XYE)_{\M(\omega)}=\rateopt(q)+\epsilon$ and likewise $H(AB|XYE)_{\M'(\omega')}=\rateopt(q')+\epsilon$. 
Let
\begin{align*}
\rho_{REF}=\lambda\,\omega_{R E}\otimes |000\rangle\langle 000|_{F_A F_B F_E} +(1-\lambda)\, \omega'_{R E}\otimes |111\rangle\langle 111|_{F_A F_B F_E}\,,
\end{align*}
where $\lambda\in[0,1]$ and where we have defined the joint system $F=F_AF_BF_E$.  We can then define a composite channel $\N$ of the same form of Equation~\eqref{eq:EAT-instruments2} for which
\begin{align}
\N(\rho_{REF})=\lambda\M(\omega_{RE})\otimes\ketbra{000}{000}_F+(1-\lambda) \M'(\omega'_{RE})\otimes\ketbra{111}{111}_F\,.
\end{align}
The auxiliary flags act ensure that the correct instruments for $\M$ and $\M'$ are applied only to the appropriate state, and we obtain


Letting $E'=EF_E$, it follows that
\begin{align*}
H(AB|XYE')_{\N(\rho)}&= \lambda H(AB|XYE)_{\M(\omega)} + (1-\lambda) H(AB|XYE)_{\M'(\omega')}\\
&= \lambda \,\rateopt(q) + (1-\lambda)\,\rateopt(q')+\epsilon\,,
\end{align*}
where we have expanded in the second line by conditioning on the classical system $F_E$.  Recall that in the statement of the definition of $\rate$ functions, $E$ is permitted to be an arbitrary auxiliary system, and thus $E'$ also fulfils this role. Additionally, $R_A$ and $R_B$ are permitted to be arbitrary in the device-independent case. By letting $R'_A= R_A F_A$ and similarly for $R'_B$ we can see that the tensor product structure of measurements is maintained and $\N$ remains in the set $\mathfrak{N}$. Since the $U$ register on $\rho$ will be distributed according to $q''=\lambda q +(1-\lambda)q'$, we have shown that
\begin{align*}
\inf_{\omega\in \Gamma_{\mathfrak{N}}(q'')}H(AB|XY R)_{\M(\omega)}\leq \lambda \,\rate(q)+(1-\lambda)\, \rate(q')\,,
\end{align*}
since $\epsilon$ is arbitrary. Hence $\rateopt(q)$ is convex.
\end{proof}
\end{lemma}
\begin{corollary}
For a device-independent, spot-checking protocol, the infimum in Equation~\eqref{eq:infimum_bound} is a convex optimisation problem.
\begin{proof}
This result is a simple combination of Lemmas~\ref{lem:convexopt},~\ref{lem:Qconvex} and~\ref{lem:ConvexRate} along with the observation that $Q_\mathfrak{N}^\gamma$ inherits convexity from $Q_{\mathfrak{N}}$.
\end{proof}
\end{corollary}

In the CHSH game we have a binary score, i.e., $U \in \{0,1\}$ for test rounds. For quantum systems the probability of winning (scoring one) is bounded by $s\in\left[1/2(1-1\sqrt{2}), 1/2(1+1/\sqrt{2})\right]$. A rate function in terms of the CHSH score (without spot-checking) was first derived in~\cite{PABGMS}, and has been applied previously in the context of device-independent cryptography~\cite{ADFRV}. This rate function, with $q=\{1-s,s\}$~\footnote{Here we are specifying the distribution in the form $\{q(0),q(1)\}$.}, is defined as
\begin{align*}
   \rate_{\mathrm{CHSH}}(q)=\chi(s)=
\begin{cases}
    1-H_{\mathrm{bin}}\left(\frac{1}{2}+\frac{1}{2}\sqrt{16 s(s-1)+3}\right),& \text{if } 3/4\leq s\leq \frac{1}{2}(1+\frac{1}{\sqrt{2}}) \text{ or } \frac{1}{2}(1-\frac{1}{\sqrt{2}})\leq s\leq 1/4  \\
    0,              & \text{if } 1/4\leq s\leq 3/4\\
    \text{undefined} & \text{otherwise}
\end{cases}
\end{align*}
where $H_{\mathrm{bin}}(x):= -x\log(x) - (1-x)\log(1-x)$ is the binary Shannon entropy. We note that the rate function is undefined for winning probabilities that are not quantum-achievable. It also only takes into account the randomness obtained from one of Alice's devices and does not consider the full measurement statistics, so one could hope to gain more entropy if both parties' outputs are taken into account~\cite{BRC}. As $\chi$ is a convex function, we can obtain a family of min-tradeoff functions by taking the tangent to this rate function at any point $t \in \left(1/2(1-1/\sqrt{2}), 1/2(1+1/\sqrt{2})\right)$. Denoting these functions by $g_t$, we have
\begin{align}
\label{eq:Gp}
g_t(\{1-s,s\})=\chi(t)+(s-t)\,\chi'(t)\,,
\end{align}
where the prime indicates the derivative. As these functions are affine we can uniquely extend their domains to include all $s\in[0,1]$. Then, we can define our min-tradeoff function as
\begin{equation}
f_t(\delta_u)= \begin{cases}
	\frac{1}{\gamma}g_{t}(\delta_u)+(1-\frac{1}{\gamma})c_{\perp} &\text{ if }u\in\{0,1\} \\
	c_\perp &\text{ if }u=\perp
\end{cases}.
\end{equation}
Note that due to the spot checking structure, for a fixed $\gamma \in (0,1)$, $\chi$ also defines a rate function on $\Q_{\mathfrak{N}}$ via $\rate(\{\gamma(1-s),\gamma s,1-\gamma\})=\chi(s)$.  Recall also that we have the freedom to choose both the tangent point $t$ and the value of $c_{\perp}$; we later optimize these.  The family of functions $f_t$ can then be applied to the entropy accumulation theorem. We now state the main theorem.
\begin{theorem}[Entropy accumulation for the CHSH game]
\label{Thm:EAT}
Let $\rho_{\mathbf{A}\mathbf{B}\mathbf{XYU}E}$ be a CQ state (classical on
$\mathbf{A}\mathbf{B}\mathbf{XYU}$) produced by the CHSH protocol from the main
text with test probability $\gamma \in (0,1)$. Let $\wexp \in (3/4,1/2(1+1/\sqrt{2})]$, $\delta > 0$ and $n \in \mathbb{N}$. Let $\Omega$ refer to the event
specified by
$|\lbrace U_i|U_i=0\rbrace|\leq n\gamma(1-(\wexp-\delta))$,
$p_\Omega$ the probability of this event occurring for
$\rho_{\mathbf{A}\mathbf{B}\mathbf{XYU}E}$, and
$\rho_{\mathbf{A}\mathbf{B}\mathbf{XYU}E|\Omega}$ the state conditioned on this
occurrence. Let $\epsmooth\in(0,1)$, $\alpha \in (1,2)$ and let $\chi$ and $f_{t}$ be defined as
above. Then for any $r$ such that $f_{t}(\freq_{\mathbf{U}})\geq r$ for
every event in $\Omega$ we have
\begin{align}
H_{\min}^{\epsmooth}(\mathbf{AB}|\mathbf{XY}E)_{\rho_{\mathbf{A}\mathbf{B}\mathbf{XYU}E|\Omega}}>n r-&\frac{\alpha}{\alpha-1}\log\left(\frac{1}{p_{\Omega}(1-\sqrt{1-\epsmooth^2})}\right) \nonumber\\
&+n\inf_{p \in \Qgamma_{\mathfrak{N}}}\bigl(\Delta(f_t,p)-(\alpha-1)V(f_t,p)\bigr)-n(\alpha-1)^2 K_\alpha(f_t),\label{eq:mainEAT}
\end{align}
where\footnote{Note that when using~\eqref{eq:V} and~\eqref{eq:K}, because of our earlier substitution $\mathbf{A}\to\mathbf{AB}$, we have replaced the $d_A$ in~\eqref{eq:V} and~\eqref{eq:K} by $d_Ad_B=4$.}
\begin{align*}
\Delta(f_t,p)&= \chi(p(1)/\gamma)-f_t(p) \\
V(f_t,p)&= \frac{\ln 2}{2} \left(\log(9) +\sqrt{2+\varq(f_{t})}\right)^2\\
K_\alpha(f_t)&= \frac{1}{6(2-\alpha)^3 \ln 2}2^{(\alpha-1)(2+\Max(f_{t})-\Min_{\Qgamma_{\mathfrak{N}}}(f_{t}))}\ln^3 \left(2^{2+\Max(f_{t})-\Min_{\Qgamma_{\mathfrak{N}}}(f_{t})}+\mathrm{e}^2\right).
\end{align*}
\end{theorem}
To use this theorem we need to find a suitable $r$.  To do so we consider the following:
\begin{align*}
  f_t(\freq_{\mathbf{U}})&=\freq_{\mathbf{U}}(0)f_t(\delta_0)+\freq_{\mathbf{U}}(1)f_t(\delta_1)+\freq_{\mathbf{U}}(\bot)f_t(\delta_\bot)\\
                         &=\freq_{\mathbf{U}}(0)\left(\frac{g_t(\delta_0)}{\gamma}+\left(1-\frac{1}{\gamma}\right)c_\bot\right)+\freq_{\mathbf{U}}(1)\left(\frac{g_t(\delta_1)}{\gamma}+\left(1-\frac{1}{\gamma}\right)c_\bot\right)+\freq_{\mathbf{U}}(\bot)c_{\bot}\\
                         &=\frac{\freq_{\mathbf{U}}(0)}{\gamma}\left(g_t(\delta_0)-c_\bot\right)+\frac{\freq_{\mathbf{U}}(1)}{\gamma}\left(g_t(\delta_1)-c_\bot\right)+c_\bot
                                 \,.
\end{align*}

For events in $\Omega$ we have
$\freq_{\mathbf{U}}(0)\leq\gamma(1-(\omega_{\mathrm{exp}}-\delta))$.  Thus, if we choose $c_\bot\geq g_t(\delta_0)$ we find
\begin{align}\label{eq:fru}
f_t(\freq_{\mathbf{U}})&\geq(1-(\omega_{\mathrm{exp}}-\delta))\left(g_t(\delta_0)-c_\bot\right)+\frac{\freq_{\mathbf{U}}(1)}{\gamma}\left(g_t(\delta_1)-c_\bot\right)+c_\bot\,.
\end{align}
Furthermore, if $c_\bot\leq g_t(\delta_1)$ then the second term is non negative and we get the bound
\begin{align*}
  f_t(\freq_{\mathbf{U}})&\geq(1-(\omega_{\mathrm{exp}}-\delta))\left(g_t(\delta_0)-c_\bot\right)+c_\bot\,.
\end{align*}
This is increasing with $c_\bot$, so it is best to take $c_\bot=g_t(\delta_1)$ giving
\begin{align*}
f_t(\freq_{\mathbf{U}})&\geq(1-(\omega_{\mathrm{exp}}-\delta))g_t(\delta_0)+(\omega_{\mathrm{exp}}-\delta)g_t(\delta_1)\\
&=g_t(\{1-(\omega_{\mathrm{exp}}-\delta),\omega_{\mathrm{exp}}-\delta\})\\
&=\chi(t)+(\omega_{\mathrm{exp}}-\delta-t)\chi'(t)\,.
\end{align*}
It follows that we can take $r=\chi(t)+(\omega_{\mathrm{exp}}-\delta-t)\chi'(t)$ in Theorem~\ref{Thm:EAT}.

Note that Theorem~\ref{Thm:EAT} holds for any $t$ and $\alpha$ and we optimize these to improve the bound. In other words, in order to evaluate our output entropy we take $p_{\Omega}=\epeat$ (cf.\ the discussion in Section~\ref{app:Error}), substitute the above expression for $r$ and take the supremum over $t$ and $\alpha$ on the right hand side of~\eqref{eq:mainEAT}.  Once a lower bound for the smooth min-entropy has been obtained, a randomness extraction procedure ensures that the security definitions are satisfied.

\begin{remark}
  We could alternatively have defined the protocol with the abort condition $|\{U_i:U_i=1\}|<n\gamma(\omega_{\mathrm{exp}}-\delta)$. In this case, taking $c_\bot=g_t(\delta_1)$ gives a similar bound to above, but the error parameters turn out to be worse.

  Similarly, if the abort condition was changed such that the protocol aborts if either $|\{U_i:U_i=0\}|>n\gamma(1-(\omega_{\mathrm{exp}}-\delta))$ or  $|\{U_i:U_i=1\}|<n\gamma(\omega_{\mathrm{exp}}-\delta)$, then we could take $g_t(\delta_0)\leq c_\bot\leq g_t(\delta_1)$ and get the same expression for $f_t(\freq_{\mathbf{U}})$.
\end{remark}

\section{System characterization}

\subsection{Determination of single photon efficiency}\label{efficiency}

We define the single photon heralding efficiency as $\eta_A=C/N_B$ and $\eta_B=C/N_A$ for Alice and Bob, in which two-photon coincidence events $C$ and single photon detection events for Alice $N_A$ and Bob $N_B$ are measured in the experiment. The heralding efficiency is listed in Table~\ref{tab:OptEffAB}, where $\eta^{\mathrm{sc}}$ is the efficiency of coupling entangled photons into single mode optical fibre, $\eta^{\mathrm{so}}$ the optical efficiency due to limited transmittance of optical elements in the source, $\eta^{\mathrm{fibre}}$ the transmittance of fibre linking source to measurement station, $\eta^{\mathrm{m}}$ the efficiency for light passing through the measurement station, and $\eta^{\mathrm{det}}$ the single photon detector efficiency.  The heralding efficiency and the transmittance of individual optical elements are listed in Table~\ref{tab:OptEffAB}, where $\eta^{\mathrm{so}}$, $\eta^{\mathrm{fibre}}$, $\eta^{\mathrm{m}}$, $\eta^{\mathrm{det}}$ can be measured with classical light beams and NIST-traceable power meters, and $\eta^{\mathrm{fibre}}$, $\eta^{\mathrm{m}}$ only exist in the case of the space-like separated setup.

\begin{table}[htbp]
\centering
\begin{tabular}{cc|c|ccccc}
\hline
\multicolumn{2}{c|}{Setup}                                & heralding efficiency ($\eta$) & $\eta^{\mathrm{sc}}$    & $\eta^{\mathrm{so}}$    & $\eta^{\mathrm{fibre}}$ & $\eta^{\mathrm{m}}$ & $\eta^{\mathrm{det}}$ \\ \hline
\multicolumn{1}{c|}{\multirow{2}{*}{space-like}}  & Alice & 80.41\%                       & \multirow{2}{*}{93.5\%} & \multirow{2}{*}{95.9\%} & \multirow{2}{*}{99.0\%} & 94.8\%              & 95.5\% \\
\multicolumn{1}{c|}{}                             & Bob   & 82.24\%                       &                         &                         &                         & 95.2\%              & 97.3\% \\ \hline
\multicolumn{1}{c|}{\multirow{2}{*}{main}}        & Alice & 83.40\%                       & \multirow{2}{*}{91.2\%} & 93.8\%                  & \multirow{2}{*}{--}     &\multirow{2}{*}{--}  & 97.5\% \\
\multicolumn{1}{c|}{}                             & Bob   & 84.80\%                       &                         & 94.3\%                  &                         &                     & 98.6\% \\ \hline
\end{tabular}
\caption{\bf Characterization of optical efficiencies in the experiments. }
\label{tab:OptEffAB}
\end{table}

\subsection{Quantum state and measurement bases}
To maximally violate the Bell inequality for our main experiment, we aim to create a non-maximally entangled two-photon state~\cite{CHSH} $\cos(\alpha)\ket{HV}+\sin(\alpha)\ket{VH}$, where $\alpha = 27.76^\circ$ and set measurement bases to be $A_1=-82.30^\circ$ and $A_2=-118.72^\circ$ for Alice, and $B_1=7.70^\circ$ and $B_2=-28.72^\circ$ for Bob, respectively. We also optimize the mean photon number to be 0.52 to maximize CHSH score.

We measure diagonal/anti-diagonal visibility in the basis set ($45^\circ, -\alpha$), ($90^\circ+\alpha, 45^\circ$) for minimum coincidence, and in the basis set ($45^\circ, 90^\circ-\alpha$), ($\alpha, 45^\circ$) for maximum coincidence, where the angles represent measurement basis  $\cos(\theta)\ket{H}+\sin(\theta)\ket{V}$ for Alice and Bob. By setting the mean photon number to $\mu=0.0034$ to suppress the multi-photon effect, we measure the visibility to be $99.4\%$ and $98.5\%$ in horizontal/vertical basis and diagonal/anti-diagonal basis.

We perform quantum state tomography measurement of the non-maximally entangled state, with result shown in Figure~\ref{Fig:Tomo}. The state fidelity is $99.06\%$. We attribute the imperfection to multi-photon components, imperfect optical elements, and imperfect spatial/spectral mode matching.

As for the experiment with space-like separation, we use the same analysis method and perform the same tests. The parameters used and a comparison of the results are listed in the Table~\ref{tab:state}.

\begin{figure}[htbp]
\centering
\resizebox{16cm}{!}{\includegraphics{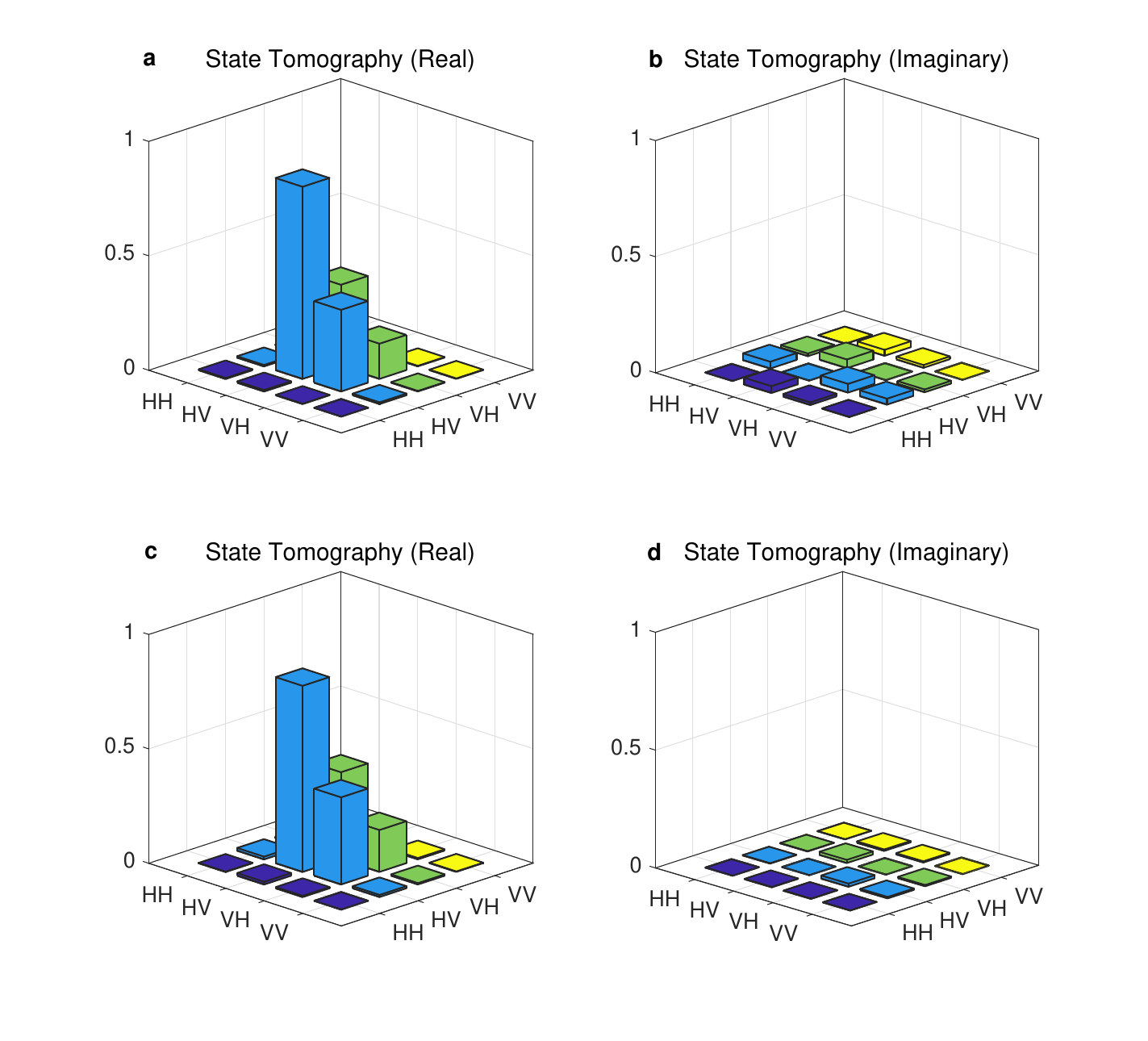}}
\caption{{\bf Tomography of the produced two-photon state in the experiments,} with real and imaginary components shown in {\bf a} and {\bf b} for the space-like experiment and real and imaginary components shown in {\bf c} and {\bf d} for the main experiment, respectively.}
\label{Fig:Tomo}
\end{figure}

\begin{table}[htbp]
\centering
\begin{tabular}{c|c|cc}
\hline
\multicolumn{2}{c|}{experiment}                             & space-like                              & main                           \\ \hline
\multicolumn{2}{c|}{entangled state: $\alpha$}              & $24.30^\circ$                           & $27.76^\circ$                           \\ \hline
\multirow{2}{*}{measurement bases} & Alice                  & $A_1=-83.08^\circ$; $A_2=-118.59^\circ$ & $A_1=-82.30^\circ$; $A_2=-118.72^\circ$ \\
                                   & Bob                    & $B_1=6.92^\circ$; $B_2=-28.59^\circ$    & $B_1=7.70^\circ$; $B_2=-28.72^\circ$    \\ \hline
\multirow{2}{*}{visibility}        & horizontal/vertical    & $99.4$\%                                & $99.4$\%                                  \\
                                   & diagonal/anti-diagonal & $98.5$\%                                & $98.5$\%                                  \\ \hline
\multicolumn{2}{c|}{fidelity}                               & $99.10$\%                               & $99.06$\%                                 \\ \hline
\end{tabular}
\caption{\bf Quantum states and measurement bases in the experiment. }
\label{tab:state}
\end{table}

\subsection{Spacetime configuration of the experiment}\label{spacetime}
To close the locality loophole in the experiment with space-like separation, space-like separation must be satisfied between relevant events at Alice and Bob's measurement stations: the state measurement events by Alice and Bob, measurement event at one station and the setting choice event at the other station (Figure~\ref{Fig:SpaceTimeSupp}). We then obtain
\begin{equation}
	\begin{cases}
(|SA| + |SB|) / c > T_{\mathrm{E}} - (L_{\mathrm{SA}} - L_{\mathrm{SB}}) / c + T_{\mathrm{QRNG}}^{A} + T_{\mathrm{Delay}}^{A} + T_{\mathrm{PC}}^{A} +T_{\mathrm{M}}^{A}, \\
(|SA| + |SB|) / c > T_{\mathrm{E}} + (L_{\mathrm{SA}} - L_{\mathrm{SB}}) / c + T_{\mathrm{QRNG}}^{B} + T_{\mathrm{Delay}}^{B} + T_{\mathrm{PC}}^{B} +T_{\mathrm{M}}^{B},
	\end{cases}
\label{Eq:SC1}
\end{equation}
where $|SA|$ = 93 m ($|SB|$ = 90 m) is the free space distance between entanglement source and Alice's (Bob's) measurement station, $T_{\mathrm{E}}$ = 10 ns is the generation time for entangled photon pairs, which is mainly contributed by the 10 ns pump pulse duration, $L_{\mathrm{SA}}$ = 191 m ($L_{\mathrm{SB}}$ = 173.5 m) is the effective optical path which is mainly contributed by the long fibre (130 m, 118 m) between source and Alice/Bob's measurement station, $T_{\mathrm{QRNG}}^{A}=T_{\mathrm{QRNG}}^{B}$ = 96 ns is the time elapse for QRNG to generate a random bit, $T_{\mathrm{Delay}}^{A}$ = 270 ns ($T_{\mathrm{Delay}}^{B}$ = 230 ns) is the delay between QRNG and Pockels cells, $T_{\mathrm{PC}}^A$ = 112 ns ($T_{\mathrm{PC}}^B$ = 100 ns) including the internal delay of the Pockels Cells (62 ns, 50 ns) and the time for Pockels cell to stabilize before performing single photon polarization state projection after switching which is 50 ns, which implies that the experimental time is able to be shortened by increasing the repetition rate of the experiment because the low testing probability reduced the impact of the modulation rate of the Pockels cells, $T_{\mathrm{M}}^A$ = 55 ns ($T_{\mathrm{M}}^B$ = 100 ns) is the time elapse for SNSPD to output an electronic signal, including the delay due to fibre and cable length.

The measurement independence requirement is satisfied by space-like separation between entangled-pair creation event and setting choice events, so we can have

\begin{equation}
	\begin{cases}
|SA| / c > L_{\mathrm{SA}} / c  - T_{\mathrm{Delay}}^A - T_{\mathrm{PC}}^A\\
|SB| / c > L_{\mathrm{SB}} / c  - T_{\mathrm{Delay}}^B- T_{\mathrm{PC}}^B
	\end{cases}
\label{Eq:SC2}
\end{equation}

As shown in Figure~\ref{Fig:SpaceTimeSupp}, Alice's and Bob's random bit generation events for input setting choices are outside the future light cone (green shade) of entanglement creation event at the source.

\begin{figure}[htbp]
\centering
\resizebox{10cm}{!}{\includegraphics{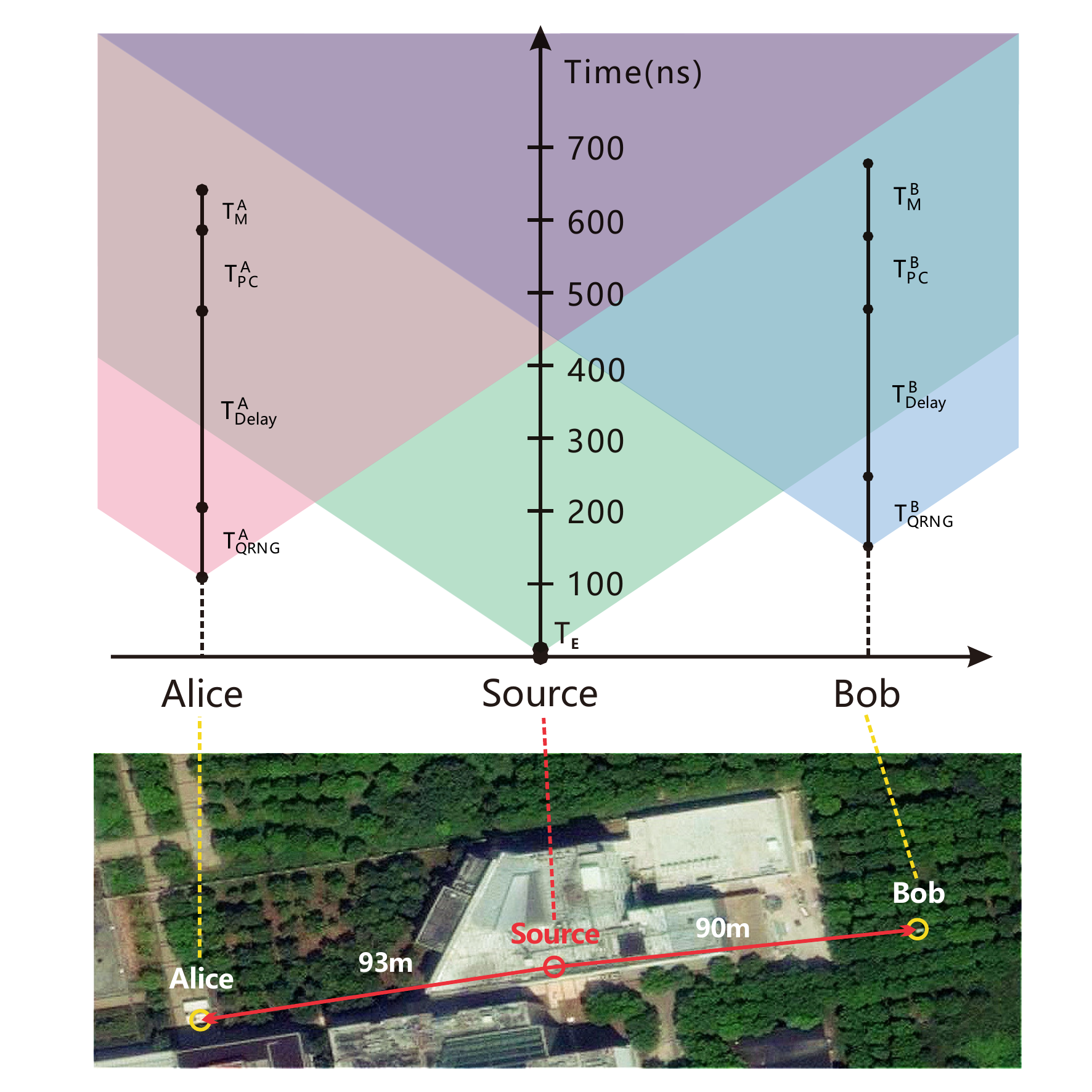}}
\caption{{\bf Spacetime analysis of the test rounds (space-like experiment).} $T_{\mathrm{E}}=10$ ns is the time elapse to generate a pair of entangled photons.  $T_{\mathrm{QRNG}}^{A,B}$ is the time elapse to generate random bits to switch the Pockels cell. $T_{\mathrm{Delay}}^{A,B}$ is the delay between QRNG and the Pockels cell. $T_{\mathrm{PC}}^{A,B}$ is the time elapse for the Pockels cell to be ready to perform state measurements after receiving the random bits from the QRNG. $T_{\mathrm{M}}^{A,B}$ is the time elapse for the SNSPD to output an electronic signal. For $T_{\mathrm{QRNG}}^{A}=T_{\mathrm{QRNG}}^{B}=96$ ns, $T_{\mathrm{Delay}}^{A}=270$ ns and $T_{\mathrm{Delay}}^{B}=230$ ns, $T_{\mathrm{PC}}^{A}=112$ ns and $T_{\mathrm{PC}}^{B}=100$ ns, $T_{\mathrm{M}}^{A}= 55$ ns and $T_{\mathrm{M}}^{B}= 100$ ns, we place Alice's measurement station and Bob's measurement station on the opposite side of the source and $93\pm1$ ($90\pm1$) meter from the source, and set the effective optical length between Alice's (Bob's) station and the source to be 130 m (118 m). This arrangement ensures space-like separation between measurement event and distant base setting event and between base setting event and photon pair emission event.}
\label{Fig:SpaceTimeSupp}
\end{figure}

\section{Parameter determination and Experimental Results}

\subsection{Parameter determination}

The implementation of the protocol depends on several parameters which we can choose in advance. Firstly, we pre-determine the testing probability $\gamma$ based on a $10$ minute-Bell test at a repetition frequency of $200$~ KHz for both experiments. The counts are summarized in Table~\ref{tab:training}. With $\wexp=0.750809$ and $\wexp=0.752487$, the optimal testing probability $\gamma_{\mathrm{opt}}$ is basically independent of soundness error $\epsound$ and completeness error $\epcomplete$ as shown in Table~\ref{tab:optgamma}. The amount of randomness expansion we achieve depends on the testing probability, $\gamma$. In Figure~\ref{Fig:fix} we show this behaviour, indicating the values chosen in our experiments (black plus, $\gamma=3.264\times10^{-4}$ and black cross, $\gamma=1.194\times10^{-4}$). Although they were chosen slightly sub-optimally, it is sufficient for our purposes.

With the corresponding $\wexp$ values above and choosing the soundness and completeness errors to be $\epsound=3.09\times10^{-12}$ and $\epcomplete=1\times10^{-6}$, randomness expansion is expected to be witnessed within $200$ hours for the space-like experiment and $12.5$ hours for the main experiment, respectively. To ensure success we run the experiments for slightly longer than this ($220$ hours and $19.2$ hours), corresponding to $3.168\times10^{12}$ and $1.3824\times10^{11}$ rounds, respectively. The randomness generation and consumption with time are shown in Figure~\ref{fig:expansion}. In Table~\ref{Tab:err} we show how the amount of randomness would vary by adjusting the completeness and soundness errors.

\begin{figure}[tbp]
\centering
\resizebox{8cm}{!}{\includegraphics{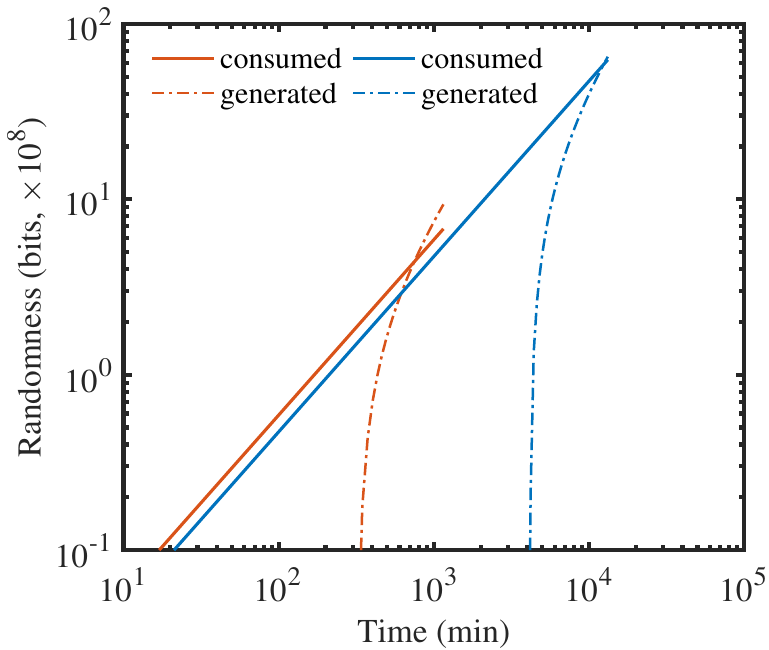}}
\caption{{\bf Randomness generated and consumed versus experimental time.} Red and blue represent the result of main and space-like experiment, respectively. The generated (dashed line) and consumed (smooth line) randomness in the experiment.}
\label{fig:expansion}
\end{figure}

\begin{table}[htbp]
\centering
\begin{tabular}{c|c|cccc}
\hline
Experiment                   & Basis settings & $ab=00$  & $ab=01$ & $ab=10$ & $ab=11$ \\ \hline
\multirow{4}{*}{space-like}  & $xy=00$        & 29172431 & 214574  & 181730  & 422697  \\
                             & $xy=01$        & 28732552 & 654944  & 134767  & 471554  \\
                             & $xy=10$        & 28711770 & 154239  & 637085  & 483897  \\
                             & $xy=11$        & 27868246 & 1043044 & 1033950 & 82496   \\ \hline
\multirow{4}{*}{main}        & $xy=00$        & 27591258 & 582051  & 566988  & 1262270 \\
                             & $xy=01$        & 26620499 & 1555007 & 341276  & 1488787 \\
                             & $xy=10$        & 26667606 & 366045  & 1485890 & 1478620 \\
                             & $xy=11$        & 24394111 & 2635885 & 2558367 & 405340  \\ \hline
\end{tabular}
\caption{{\bf Counts of training rounds.} Recorded number of two-photon detection events for four sets of polarization state measurement bases $x = 0$ or $1$ indicates ``$0-$'' or ``$1/2-$'' wave voltages for Pockels cell respectively and $a=1$ or $0$ indicates that Alice detects a photon or not, the same applies for $y$ and $b$ on Bob's side. The expected CHSH scores are $0.750809$ and $0.752487$ for space-like and main experiment, respectively.}
\label{tab:training}
\end{table}

\begin{table}[htbp]
\centering
\begin{tabular}{cc|cc|cc}
\hline
\multicolumn{2}{c|}{error parameters} &\multicolumn{2}{c|}{space-like experiment}      & \multicolumn{2}{c}{main experiment}            \\ \hline
$\epsound$        & $\epcomplete$     & $n_{\mathrm{min}}$   & $\gamma_{\mathrm{opt}}$ & $n_{\mathrm{min}}$   & $\gamma_{\mathrm{opt}}$ \\ \hline
$1\times10^{-6}$  & $1\times10^{-6}$  & $1.861\times10^{12}$ & $9.840\times10^{-5}$    & $5.761\times10^{10}$ & $3.393\times10^{-4}$    \\
$1\times10^{-6}$  & $1\times10^{-9}$  & $2.202\times10^{12}$ & $9.813\times10^{-5}$    & $6.808\times10^{10}$ & $3.395\times10^{-4}$    \\
$1\times10^{-6}$  & $1\times10^{-12}$ & $2.508\times10^{12}$ & $9.850\times10^{-5}$    & $7.747\times10^{10}$ & $3.394\times10^{-4}$    \\
$1\times10^{-9}$  & $1\times10^{-6}$  & $2.432\times10^{12}$ & $9.850\times10^{-5}$    & $7.533\times10^{10}$ & $3.393\times10^{-4}$    \\
$1\times10^{-9}$  & $1\times10^{-9}$  & $2.819\times10^{12}$ & $9.833\times10^{-5}$    & $8.724\times10^{10}$ & $3.393\times10^{-4}$    \\
$1\times10^{-9}$  & $1\times10^{-12}$ & $3.164\times10^{12}$ & $9.793\times10^{-5}$    & $9.782\times10^{10}$ & $3.394\times10^{-4}$    \\
$1\times10^{-12}$ & $1\times10^{-6}$  & $2.975\times10^{12}$ & $9.833\times10^{-5}$    & $9.221\times10^{10}$ & $3.394\times10^{-4}$    \\
$1\times10^{-12}$ & $1\times10^{-9}$  & $3.402\times10^{12}$ & $9.844\times10^{-5}$    & $1.053\times10^{11}$ & $3.393\times10^{-4}$    \\
$1\times10^{-12}$ & $1\times10^{-12}$ & $3.779\times10^{12}$ & $9.857\times10^{-5}$    & $1.169\times10^{11}$ & $3.393\times10^{-4}$    \\ \hline
\end{tabular}
\caption{{\bf Parameters for achieving randomness expansion with different expected scores.} For various values of the soundness error $\epsound$ and completeness error $\epcomplete$, $n_{\mathrm{min}}$ represents the minimal expected number of rounds witness expansion with the $\wexp$ for each setup using the optimized testing probability, $\gamma_{\mathrm{opt}}$.}
\label{tab:optgamma}
\end{table}

\begin{figure}[htbp]
\centering
\resizebox{16cm}{!}{\includegraphics{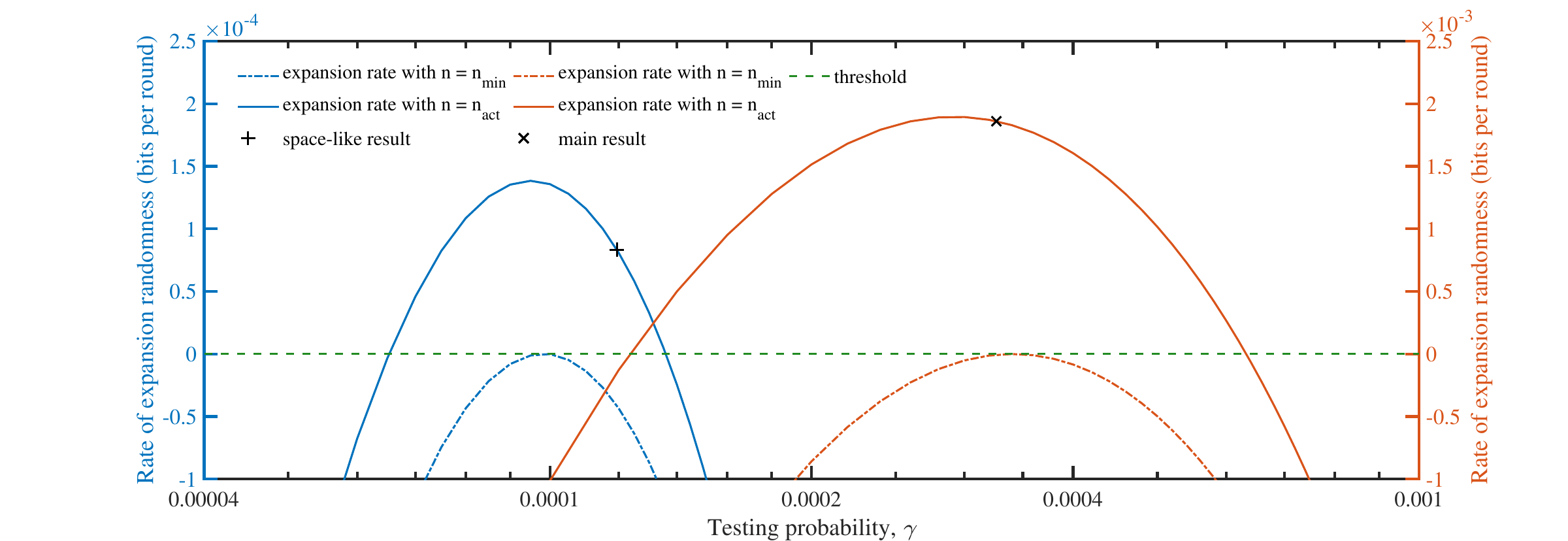}}
\caption{{\bf Expansion rate versus the value of $\gamma$.} With a fixed number of rounds and $\omega_{\mathrm{exp}}=0.750809$ (blue) or $\omega_{\mathrm{exp}}=0.752487$ (red), the expansion rate changes slowly around the optimal value of $\gamma$ (fixing the other parameters to those used in the protocol). The green dashed line is the threshold to witness expansion. The blue and red dashed lines represent the expansion rate with expected number of rounds $n_{\mathrm{min}}=2.888\times10^{12}$ and $n_{\mathrm{min}}=8.951\times10^{10}$, respectively. The blue and red smooth lines represent the expansion rate with actual number of rounds $n_{\mathrm{act}}=3.168\times10^{12}$ and $n_{\mathrm{act}}=1.3824\times10^{11}$, respectively. The black plus and cross indicate the value of $\gamma=1.194\times10^{-4}$ and $\gamma=3.264\times10^{-4}$ used in two different experiments, respectively.}
\label{Fig:fix}
\end{figure}

\begin{table}[htbp]
\centering
\begin{tabular}{c|c|ccccccc}
\hline
Experiment                  & \diagbox{$\epcomplete$}{$\epsound$}   & $10^{-8}$         & $10^{-9}$         & $10^{-10}$        & $10^{-11}$        & $10^{-12}$        & $10^{-13}$        & $10^{-14}$        \\ \hline
\multirow{7}{*}{space-like} & $10^{-3}$                             & $2.031\times10^9$ & $1.732\times10^9$ & $1.449\times10^9$ & $1.179\times10^9$ & $9.214\times10^8$ & $6.737\times10^8$ & $4.352\times10^8$ \\
                            & $10^{-4}$                             & $1.735\times10^9$ & $1.436\times10^9$ & $1.153\times10^9$ & $8.834\times10^8$ & $6.255\times10^8$ & $3.778\times10^8$ & $1.393\times10^8$ \\
                            & $10^{-5}$                             & $1.478\times10^9$ & $1.179\times10^9$ & $8.961\times10^8$ & $6.265\times10^8$ & $3.685\times10^8$ & $1.209\times10^8$ & --                \\
                            & $10^{-6}$                             & $1.248\times10^9$ & $9.492\times10^8$ & $6.661\times10^8$ & $3.965\times10^8$ & $1.386\times10^8$ & --                & --                \\
                            & $10^{-7}$                             & $1.038\times10^9$ & $7.393\times10^8$ & $4.563\times10^8$ & $1.866\times10^8$ & --                & --                & --                \\
                            & $10^{-8}$                             & $8.438\times10^8$ & $5.451\times10^8$ & $2.621\times10^8$ & --                & --                & --                & --                \\
                            & $10^{-9}$                             & $6.622\times10^8$ & $3.635\times10^8$ & $8.042\times10^7$ & --                & --                & --                & --                \\ \hline
\multirow{7}{*}{main}       & $10^{-3}$                             & $4.826\times10^8$ & $4.443\times10^8$ & $4.080\times10^8$ & $3.735\times10^8$ & $3.405\times10^8$ & $3.088\times10^8$ & $2.782\times10^8$ \\
                            & $10^{-4}$                             & $4.450\times10^8$ & $4.067\times10^8$ & $3.705\times10^8$ & $3.360\times10^8$ & $3.029\times10^8$ & $2.712\times10^8$ & $2.407\times10^8$ \\
                            & $10^{-5}$                             & $4.124\times10^8$ & $3.741\times10^8$ & $3.379\times10^8$ & $3.034\times10^8$ & $2.703\times10^8$ & $2.386\times10^8$ & $2.081\times10^8$ \\
                            & $10^{-6}$                             & $3.832\times10^8$ & $3.449\times10^8$ & $3.087\times10^8$ & $2.742\times10^8$ & $2.412\times10^8$ & $2.095\times10^8$ & $1.790\times10^8$ \\
                            & $10^{-7}$                             & $3.565\times10^8$ & $3.183\times10^8$ & $2.821\times10^8$ & $2.476\times10^8$ & $2.146\times10^8$ & $1.829\times10^8$ & $1.524\times10^8$ \\
                            & $10^{-8}$                             & $3.319\times10^8$ & $2.937\times10^8$ & $2.574\times10^8$ & $2.229\times10^8$ & $1.899\times10^8$ & $1.583\times10^8$ & $1.278\times10^8$ \\
                            & $10^{-9}$                             & $3.088\times10^8$ & $2.706\times10^8$ & $2.344\times10^8$ & $1.999\times10^8$ & $1.669\times10^8$ & $1.352\times10^8$ & $1.047\times10^8$ \\ \hline
\end{tabular}
\caption{Values of the net output that our experiment would achieve with varying completeness and soundness errors.}
\label{Tab:err}
\end{table}

\subsection{Experimental results}
The recorded experimental data are listed in Table~\ref{tab:output}. 
We assign $1$ for a detection event and $0$ for no detection. The CHSH score $\omega_{\mathrm{CHSH}}$ is given by
\begin{equation}
\omega_{\mathrm{CHSH}} = \sum_{k,l}{\sum_{i=1}^{n_{x_iy_i=kl}}{(1+(-1)^{a_i\oplus b_i\oplus(x_i\cdot y_i)})/n_{x_iy_i=kl}}},
\end{equation}
with $(k,l)\in(0,1)\times(0,1)$ is computed to be $0.750805$ and $0.752484$ in the space-like and main experiments, respectively. In Figure~\ref{Fig:vio}a, we show the measured CHSH violation value versus time. Due to the short time required for the experiment and the well stability of the whole system, we continuously carried out $19.2$ hours data collection without any break for calibrations. As a comparison, in Figure~\ref{Fig:vio}b, we also show the CHSH performance in the experiment with space-like separation. The experiment took a total about $2$ weeks to collect data, but only contains $220$ hours of valid data. Since the optical fibre links exposed in the field are easily affected by the environment, which will reduce the visibility of the whole system, we need to periodically recalibrate the system based on the observed CHSH scores. 

After randomness extraction using a $6.496\text{Gb} \times 3.17\text{Tb}$ (in this case because of the increased size of the data, we perform the extraction on the {\sc Viking} supercomputer) or a $0.935\text{Gb} \times 0.138\text{Tb}$ Toeplitz matrix, we obtain $6.496\times10^9$ and $9.350\times10^8$ genuinely quantum-certified random bits with the same uniformity within $3.09\times10^{-12}$, which are equivalent to $2.63\times10^8$ and $2.57\times10^8$ net bits after subtracting the randomness consumed, respectively. The two streams of random bits both pass the NIST statistical test suite (see Table~\ref{Tab:NISTtest} for details).

We have made all the final random output available at \href{https://www.dropbox.com/sh/hae9ht1cc426i5g/AABcuGgGyuNJMC0zOxIhOQBGa?dl=0} {{\tt https://tinyurl.com/qssxxaq}} so that it may be used for testing (the output has been split between several files for convenience of downloading).

\begin{figure}[htbp]
\centering
\resizebox{16cm}{!}{\includegraphics{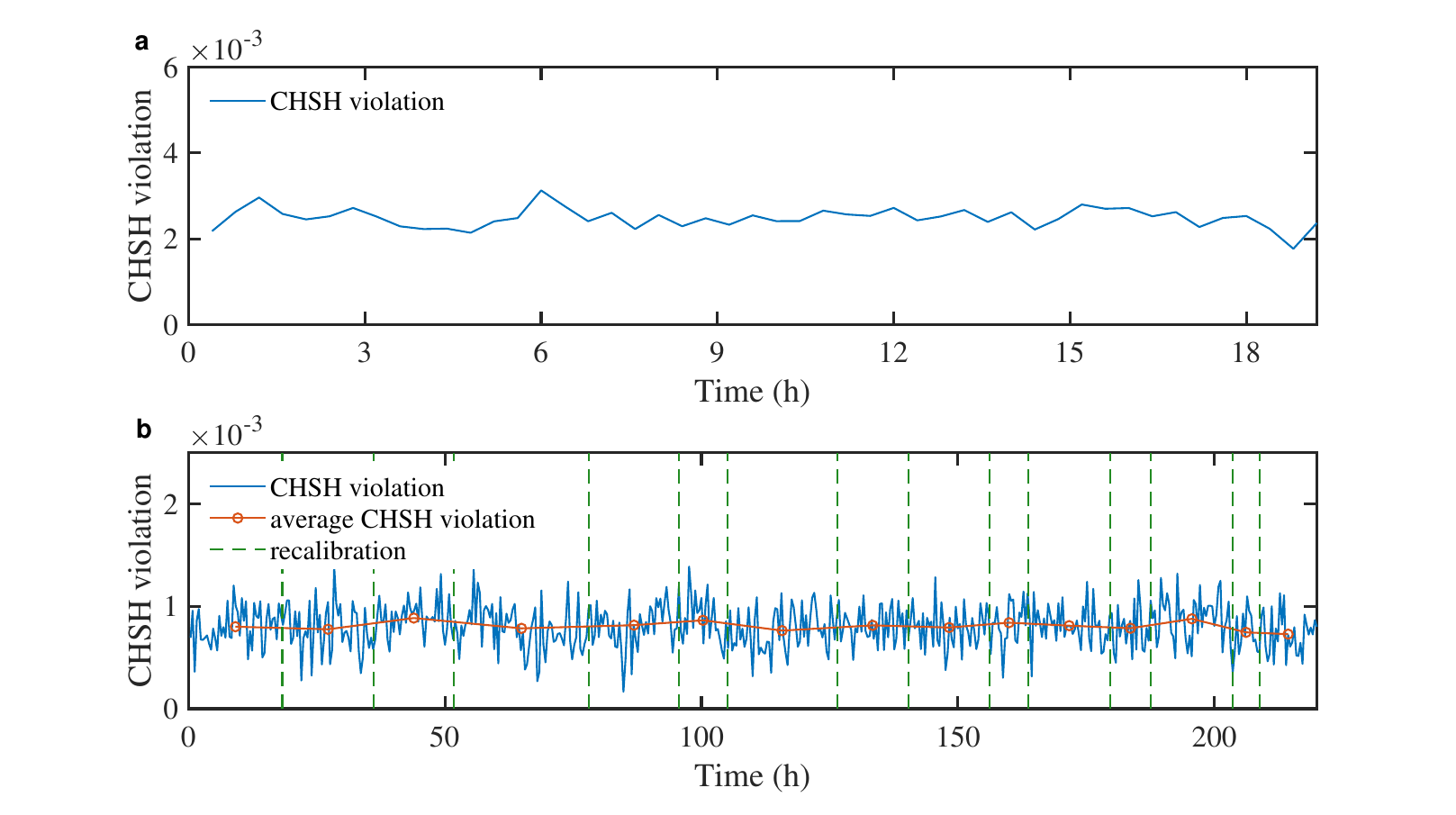}}
\caption{{\bf CHSH violation versus time.} For two experiments, we both choose the average value of CHSH violations for each $24$min data as a point to observe its performance over experimental time. {\bf a}: The main experiment. {\bf b}: The space-like experiment. The green dashed lines divide the experiment into a few segments, inbetween a recalibration was performed. The red circles represent the average value of CHSH violation between two adjacent recalibrations.}
\label{Fig:vio}
\end{figure}

\begin{table}[htbp]
\centering
\begin{tabular}{c|c|cccc}
\hline
Setup                        & Basis settings               & $ab=00$       & $ab=01$     & $ab=10$     & $ab=11$    \\\hline
\multirow{5}{*}{space-like}  & $xy=00\mathrm{(generation)}$ & 3079174741623 & 22815515154 & 19498193776 & 46131693660 \\
                             & $xy=00$                      & 91961904 	    & 681402 	  & 582328 	    & 1377758    \\
                             & $xy=01$                      & 90500314 	    & 2120124 	  & 457496 	    & 1503804    \\
                             & $xy=10$                      & 90460255 	    & 502741 	  & 2061782 	& 1557955    \\
                             & $xy=11$                      & 87602250 	    & 3356060 	  & 3353628 	& 263743     \\\hline
\multirow{5}{*}{main}        & $xy=00\mathrm{(generation)}$ & 126900146278  & 2775664347  & 2645179400  & 5873897572 \\
                             & $xy=00$                      & 10358320 	    & 226153 	  & 215067 	    & 480849     \\
                             & $xy=01$                      & 9982527 	    & 600342 	  & 128778 	    & 566330     \\
                             & $xy=10$                      & 10001267 	    & 141618 	  & 567085 	    & 564216     \\
                             & $xy=11$                      & 9130907 	    & 1013382 	  & 978197 	    & 157365     \\\hline

\end{tabular}
\caption{{\bf Counts of experimental rounds.} Recorded number of two-photon detection events for four sets of polarization state measurement bases $x=0$ or $1$ indicates ``$0-$'' or ``$1/2-$'' wave voltages for Pockels cell respectively and $a=1$ or $0$ indicates that Alice detects a photon or not, the same applies for $y$ and $b$ on Bob's side. The CHSH score are $0.750805$ and $0.752484$ for space-like and main experiment, respectively.}
\label{tab:output}
\end{table}

To check the statistical properties of our output, we run them through the NIST test suite~\cite{NIST_Tests}. To do so, we set the section length to $1$ Mbits for our $6.496\times10^9$ and $9.350\times10^8$ random output bits. As shown in Table~\ref{Tab:NISTtest}, the random bits successfully pass the tests.

\begin{table}
\centering
\begin{tabular}{c|ccc|ccc}
\hline
Experiment              & \multicolumn{3}{c}{space-like}                 & \multicolumn{3}{|c}{main}         \\\hline
Statistical   tests     & P value       & Proportion      & Result       & P value   & Proportion  & Result  \\\hline
Frequency               & 0.836586      & 0.990           & Success      & 0.363700  & 0.994       & Success \\
BlockFrequency          & 0.680822      & 0.991           & Success      & 0.672894  & 0.988       & Success \\
CumulativeSums          & 0.375465      & 0.990           & Success      & 0.121862  & 0.993       & Success \\
Runs                    & 0.085139      & 0.992           & Success      & 0.383642  & 0.994       & Success \\
LongestRun              & 0.535389      & 0.990           & Success      & 0.337586  & 0.993       & Success \\
Rank                    & 0.351484      & 0.991           & Success      & 0.891287  & 0.994       & Success \\
FFT                     & 0.170935      & 0.989           & Success      & 0.695112  & 0.986       & Success \\
NonOverlappingTemplate  & 0.142187      & 0.990           & Success      & 0.494828  & 0.990       & Success \\
OverlappingTemplate     & 0.374441      & 0.989           & Success      & 0.039521  & 0.986       & Success \\
Universal               & 0.996501      & 0.988           & Success      & 0.203199  & 0.988       & Success \\
ApproximateEntropy      & 0.525783      & 0.991           & Success      & 0.244667  & 0.990       & Success \\
RandomExcursions        & 0.120519      & 0.988           & Success      & 0.358633  & 0.986       & Success \\
RandomExcursionsVariant & 0.282519      & 0.991           & Success      & 0.500535  & 0.993       & Success \\
Serial                  & 0.261111      & 0.990           & Success      & 0.212314  & 0.993       & Success \\
LinearComplexity        & 0.669968      & 0.990           & Success      & 0.903456  & 0.985       & Success \\\hline
\end{tabular}
\caption{Results of the NIST test suite after dividing our output into $1$~Mbit sections.}
\label{Tab:NISTtest}
\end{table}

\section*{References}
\bibliographystyle{apsrev4-1}
\bibliography{RandBib}

\begin{thebibliography}{56}%
\makeatletter
\providecommand \@ifxundefined [1]{%
 \@ifx{#1\undefined}
}%
\providecommand \@ifnum [1]{%
 \ifnum #1\expandafter \@firstoftwo
 \else \expandafter \@secondoftwo
 \fi
}%
\providecommand \@ifx [1]{%
 \ifx #1\expandafter \@firstoftwo
 \else \expandafter \@secondoftwo
 \fi
}%
\providecommand \natexlab [1]{#1}%
\providecommand \enquote  [1]{``#1''}%
\providecommand \bibnamefont  [1]{#1}%
\providecommand \bibfnamefont [1]{#1}%
\providecommand \citenamefont [1]{#1}%
\providecommand \href@noop [0]{\@secondoftwo}%
\providecommand \href [0]{\begingroup \@sanitize@url \@href}%
\providecommand \@href[1]{\@@startlink{#1}\@@href}%
\providecommand \@@href[1]{\endgroup#1\@@endlink}%
\providecommand \@sanitize@url [0]{\catcode `\\12\catcode `\$12\catcode
  `\&12\catcode `\#12\catcode `\^12\catcode `\_12\catcode `\%12\relax}%
\providecommand \@@startlink[1]{}%
\providecommand \@@endlink[0]{}%
\providecommand \url  [0]{\begingroup\@sanitize@url \@url }%
\providecommand \@url [1]{\endgroup\@href {#1}{\urlprefix }}%
\providecommand \urlprefix  [0]{URL }%
\providecommand \Eprint [0]{\href }%
\providecommand \doibase [0]{http://dx.doi.org/}%
\providecommand \selectlanguage [0]{\@gobble}%
\providecommand \bibinfo  [0]{\@secondoftwo}%
\providecommand \bibfield  [0]{\@secondoftwo}%
\providecommand \translation [1]{[#1]}%
\providecommand \BibitemOpen [0]{}%
\providecommand \bibitemStop [0]{}%
\providecommand \bibitemNoStop [0]{.\EOS\space}%
\providecommand \EOS [0]{\spacefactor3000\relax}%
\providecommand \BibitemShut  [1]{\csname bibitem#1\endcsname}%
\let\auto@bib@innerbib\@empty
\bibitem [{\citenamefont {Colbeck}(2007)}]{ColbeckThesis}%
  \BibitemOpen
  \bibfield  {author} {\bibinfo {author} {\bibfnamefont {R.~A.}\ \bibnamefont
  {Colbeck}},\ }\emph {\bibinfo {title} {Quantum and relativistic protocols for
  secure multi-party computation}},\ \href@noop {} {Ph.D. thesis},\ \bibinfo
  {school} {University of Cambridge} (\bibinfo {year} {2007}),\ \bibinfo {note}
  {also available as
  \href{https://arxiv.org/abs/0911.3814}{arXiv:0911.3814}.}\BibitemShut {Stop}%
\bibitem [{\citenamefont {Colbeck}\ and\ \citenamefont {Kent}(2011)}]{CK2}%
  \BibitemOpen
  \bibfield  {author} {\bibinfo {author} {\bibfnamefont {R.}~\bibnamefont
  {Colbeck}}\ and\ \bibinfo {author} {\bibfnamefont {A.}~\bibnamefont {Kent}},\
  }\href {\doibase 10.1088/1751-8113/44/9/095305} {\bibfield  {journal}
  {\bibinfo  {journal} {Journal of Physics A: Mathematical and Theoretical}\
  }\textbf {\bibinfo {volume} {44}},\ \bibinfo {pages} {095305} (\bibinfo
  {year} {2011})}\BibitemShut {NoStop}%
\bibitem [{\citenamefont {Pironio}\ \emph {et~al.}(2010)\citenamefont
  {Pironio}, \citenamefont {Ac{\'i}n}, \citenamefont {Massar}, \citenamefont
  {de~la Giroday}, \citenamefont {Matsukevich}, \citenamefont {Maunz},
  \citenamefont {Olmschenk}, \citenamefont {Hayes}, \citenamefont {Luo},
  \citenamefont {Manning},\ and\ \citenamefont {Monroe}}]{PAMBMMOHLMM}%
  \BibitemOpen
  \bibfield  {author} {\bibinfo {author} {\bibfnamefont {S.}~\bibnamefont
  {Pironio}}, \bibinfo {author} {\bibfnamefont {A.}~\bibnamefont {Ac{\'i}n}},
  \bibinfo {author} {\bibfnamefont {S.}~\bibnamefont {Massar}}, \bibinfo
  {author} {\bibfnamefont {A.~B.}\ \bibnamefont {de~la Giroday}}, \bibinfo
  {author} {\bibfnamefont {D.~N.}\ \bibnamefont {Matsukevich}}, \bibinfo
  {author} {\bibfnamefont {P.}~\bibnamefont {Maunz}}, \bibinfo {author}
  {\bibfnamefont {S.}~\bibnamefont {Olmschenk}}, \bibinfo {author}
  {\bibfnamefont {D.}~\bibnamefont {Hayes}}, \bibinfo {author} {\bibfnamefont
  {L.}~\bibnamefont {Luo}}, \bibinfo {author} {\bibfnamefont {T.~A.}\
  \bibnamefont {Manning}}, \ and\ \bibinfo {author} {\bibfnamefont
  {C.}~\bibnamefont {Monroe}},\ }\href {\doibase 10.1038/nature09008}
  {\bibfield  {journal} {\bibinfo  {journal} {Nature}\ }\textbf {\bibinfo
  {volume} {464}},\ \bibinfo {pages} {1021} (\bibinfo {year}
  {2010})}\BibitemShut {NoStop}%
\bibitem [{\citenamefont {Arnon-Friedman}\ \emph {et~al.}(2018)\citenamefont
  {Arnon-Friedman}, \citenamefont {Dupuis}, \citenamefont {Fawzi},
  \citenamefont {Renner},\ and\ \citenamefont {Vidick}}]{ADFRV}%
  \BibitemOpen
  \bibfield  {author} {\bibinfo {author} {\bibfnamefont {R.}~\bibnamefont
  {Arnon-Friedman}}, \bibinfo {author} {\bibfnamefont {F.}~\bibnamefont
  {Dupuis}}, \bibinfo {author} {\bibfnamefont {O.}~\bibnamefont {Fawzi}},
  \bibinfo {author} {\bibfnamefont {R.}~\bibnamefont {Renner}}, \ and\ \bibinfo
  {author} {\bibfnamefont {T.}~\bibnamefont {Vidick}},\ }\href {\doibase
  10.1038/s41467-017-02307-4} {\bibfield  {journal} {\bibinfo  {journal}
  {Nature Communications}\ }\textbf {\bibinfo {volume} {9}},\ \bibinfo {pages}
  {459} (\bibinfo {year} {2018})}\BibitemShut {NoStop}%
\bibitem [{\citenamefont {Liu}\ \emph {et~al.}(2018{\natexlab{a}})\citenamefont
  {Liu}, \citenamefont {Yuan}, \citenamefont {Li}, \citenamefont {Zhang},
  \citenamefont {Zhao}, \citenamefont {Zhong}, \citenamefont {Cao},
  \citenamefont {Li}, \citenamefont {Chen}, \citenamefont {Li}, \citenamefont
  {Peng}, \citenamefont {Chen}, \citenamefont {Peng}, \citenamefont {Shi},
  \citenamefont {Wang}, \citenamefont {You}, \citenamefont {Ma}, \citenamefont
  {Fan}, \citenamefont {Zhang},\ and\ \citenamefont {Pan}}]{Liu1}%
  \BibitemOpen
  \bibfield  {author} {\bibinfo {author} {\bibfnamefont {Y.}~\bibnamefont
  {Liu}}, \bibinfo {author} {\bibfnamefont {X.}~\bibnamefont {Yuan}}, \bibinfo
  {author} {\bibfnamefont {M.-H.}\ \bibnamefont {Li}}, \bibinfo {author}
  {\bibfnamefont {W.}~\bibnamefont {Zhang}}, \bibinfo {author} {\bibfnamefont
  {Q.}~\bibnamefont {Zhao}}, \bibinfo {author} {\bibfnamefont {J.}~\bibnamefont
  {Zhong}}, \bibinfo {author} {\bibfnamefont {Y.}~\bibnamefont {Cao}}, \bibinfo
  {author} {\bibfnamefont {Y.-H.}\ \bibnamefont {Li}}, \bibinfo {author}
  {\bibfnamefont {L.-K.}\ \bibnamefont {Chen}}, \bibinfo {author}
  {\bibfnamefont {H.}~\bibnamefont {Li}}, \bibinfo {author} {\bibfnamefont
  {T.}~\bibnamefont {Peng}}, \bibinfo {author} {\bibfnamefont {Y.-A.}\
  \bibnamefont {Chen}}, \bibinfo {author} {\bibfnamefont {C.-Z.}\ \bibnamefont
  {Peng}}, \bibinfo {author} {\bibfnamefont {S.-C.}\ \bibnamefont {Shi}},
  \bibinfo {author} {\bibfnamefont {Z.}~\bibnamefont {Wang}}, \bibinfo {author}
  {\bibfnamefont {L.}~\bibnamefont {You}}, \bibinfo {author} {\bibfnamefont
  {X.}~\bibnamefont {Ma}}, \bibinfo {author} {\bibfnamefont {J.}~\bibnamefont
  {Fan}}, \bibinfo {author} {\bibfnamefont {Q.}~\bibnamefont {Zhang}}, \ and\
  \bibinfo {author} {\bibfnamefont {J.-W.}\ \bibnamefont {Pan}},\ }\href
  {\doibase 10.1103/PhysRevLett.120.010503} {\bibfield  {journal} {\bibinfo
  {journal} {Phys. Rev. Lett.}\ }\textbf {\bibinfo {volume} {120}},\ \bibinfo
  {pages} {010503} (\bibinfo {year} {2018}{\natexlab{a}})}\BibitemShut
  {NoStop}%
\bibitem [{\citenamefont {Shen}\ \emph {et~al.}(2018)\citenamefont {Shen},
  \citenamefont {Lee}, \citenamefont {Thinh}, \citenamefont {Bancal},
  \citenamefont {Cer\`e}, \citenamefont {Lamas-Linares}, \citenamefont {Lita},
  \citenamefont {Gerrits}, \citenamefont {Nam}, \citenamefont {Scarani},\ and\
  \citenamefont {Kurtsiefer}}]{Shen}%
  \BibitemOpen
  \bibfield  {author} {\bibinfo {author} {\bibfnamefont {L.}~\bibnamefont
  {Shen}}, \bibinfo {author} {\bibfnamefont {J.}~\bibnamefont {Lee}}, \bibinfo
  {author} {\bibfnamefont {L.~P.}\ \bibnamefont {Thinh}}, \bibinfo {author}
  {\bibfnamefont {J.-D.}\ \bibnamefont {Bancal}}, \bibinfo {author}
  {\bibfnamefont {A.}~\bibnamefont {Cer\`e}}, \bibinfo {author} {\bibfnamefont
  {A.}~\bibnamefont {Lamas-Linares}}, \bibinfo {author} {\bibfnamefont
  {A.}~\bibnamefont {Lita}}, \bibinfo {author} {\bibfnamefont {T.}~\bibnamefont
  {Gerrits}}, \bibinfo {author} {\bibfnamefont {S.~W.}\ \bibnamefont {Nam}},
  \bibinfo {author} {\bibfnamefont {V.}~\bibnamefont {Scarani}}, \ and\
  \bibinfo {author} {\bibfnamefont {C.}~\bibnamefont {Kurtsiefer}},\ }\href
  {\doibase 10.1103/PhysRevLett.121.150402} {\bibfield  {journal} {\bibinfo
  {journal} {Phys. Rev. Lett.}\ }\textbf {\bibinfo {volume} {121}},\ \bibinfo
  {pages} {150402} (\bibinfo {year} {2018})}\BibitemShut {NoStop}%
\bibitem [{\citenamefont {Bierhorst}\ \emph {et~al.}(2018)\citenamefont
  {Bierhorst}, \citenamefont {Knill}, \citenamefont {Glancy}, \citenamefont
  {Zhang}, \citenamefont {Mink}, \citenamefont {Jordan}, \citenamefont
  {Rommal}, \citenamefont {Liu}, \citenamefont {Christensen}, \citenamefont
  {Nam}, \citenamefont {Stevens},\ and\ \citenamefont {Shalm}}]{QRNG1}%
  \BibitemOpen
  \bibfield  {author} {\bibinfo {author} {\bibfnamefont {P.}~\bibnamefont
  {Bierhorst}}, \bibinfo {author} {\bibfnamefont {E.}~\bibnamefont {Knill}},
  \bibinfo {author} {\bibfnamefont {S.}~\bibnamefont {Glancy}}, \bibinfo
  {author} {\bibfnamefont {Y.}~\bibnamefont {Zhang}}, \bibinfo {author}
  {\bibfnamefont {A.}~\bibnamefont {Mink}}, \bibinfo {author} {\bibfnamefont
  {S.}~\bibnamefont {Jordan}}, \bibinfo {author} {\bibfnamefont
  {A.}~\bibnamefont {Rommal}}, \bibinfo {author} {\bibfnamefont {Y.-K.}\
  \bibnamefont {Liu}}, \bibinfo {author} {\bibfnamefont {B.}~\bibnamefont
  {Christensen}}, \bibinfo {author} {\bibfnamefont {S.~W.}\ \bibnamefont
  {Nam}}, \bibinfo {author} {\bibfnamefont {M.~J.}\ \bibnamefont {Stevens}}, \
  and\ \bibinfo {author} {\bibfnamefont {L.~K.}\ \bibnamefont {Shalm}},\ }\href
  {\doibase 10.1038/s41586-018-0019-0} {\bibfield  {journal} {\bibinfo
  {journal} {Nature}\ }\textbf {\bibinfo {volume} {556}},\ \bibinfo {pages}
  {223} (\bibinfo {year} {2018})}\BibitemShut {NoStop}%
\bibitem [{\citenamefont {Liu}\ \emph {et~al.}(2018{\natexlab{b}})\citenamefont
  {Liu}, \citenamefont {Zhao}, \citenamefont {Li}, \citenamefont {Guan},
  \citenamefont {Zhang}, \citenamefont {Bai}, \citenamefont {Zhang},
  \citenamefont {Liu}, \citenamefont {Wu}, \citenamefont {Yuan}, \citenamefont
  {Li}, \citenamefont {Munro}, \citenamefont {Wang}, \citenamefont {You},
  \citenamefont {Zhang}, \citenamefont {Ma}, \citenamefont {Fan}, \citenamefont
  {Zhang},\ and\ \citenamefont {Pan}}]{Liu2}%
  \BibitemOpen
  \bibfield  {author} {\bibinfo {author} {\bibfnamefont {Y.}~\bibnamefont
  {Liu}}, \bibinfo {author} {\bibfnamefont {Q.}~\bibnamefont {Zhao}}, \bibinfo
  {author} {\bibfnamefont {M.-H.}\ \bibnamefont {Li}}, \bibinfo {author}
  {\bibfnamefont {J.-Y.}\ \bibnamefont {Guan}}, \bibinfo {author}
  {\bibfnamefont {Y.}~\bibnamefont {Zhang}}, \bibinfo {author} {\bibfnamefont
  {B.}~\bibnamefont {Bai}}, \bibinfo {author} {\bibfnamefont {W.}~\bibnamefont
  {Zhang}}, \bibinfo {author} {\bibfnamefont {W.-Z.}\ \bibnamefont {Liu}},
  \bibinfo {author} {\bibfnamefont {C.}~\bibnamefont {Wu}}, \bibinfo {author}
  {\bibfnamefont {X.}~\bibnamefont {Yuan}}, \bibinfo {author} {\bibfnamefont
  {H.}~\bibnamefont {Li}}, \bibinfo {author} {\bibfnamefont {W.~J.}\
  \bibnamefont {Munro}}, \bibinfo {author} {\bibfnamefont {Z.}~\bibnamefont
  {Wang}}, \bibinfo {author} {\bibfnamefont {L.}~\bibnamefont {You}}, \bibinfo
  {author} {\bibfnamefont {J.}~\bibnamefont {Zhang}}, \bibinfo {author}
  {\bibfnamefont {X.}~\bibnamefont {Ma}}, \bibinfo {author} {\bibfnamefont
  {J.}~\bibnamefont {Fan}}, \bibinfo {author} {\bibfnamefont {Q.}~\bibnamefont
  {Zhang}}, \ and\ \bibinfo {author} {\bibfnamefont {J.-W.}\ \bibnamefont
  {Pan}},\ }\href {\doibase 10.1038/s41586-018-0559-3} {\bibfield  {journal}
  {\bibinfo  {journal} {Nature}\ }\textbf {\bibinfo {volume} {562}},\ \bibinfo
  {pages} {548} (\bibinfo {year} {2018}{\natexlab{b}})}\BibitemShut {NoStop}%
\bibitem [{\citenamefont {Zhang}\ \emph
  {et~al.}(2020{\natexlab{a}})\citenamefont {Zhang}, \citenamefont {Shalm},
  \citenamefont {Bienfang}, \citenamefont {Stevens}, \citenamefont {Mazurek},
  \citenamefont {Nam}, \citenamefont {Abell\'an}, \citenamefont {Amaya},
  \citenamefont {Mitchell}, \citenamefont {Fu}, \citenamefont {Miller},
  \citenamefont {Mink},\ and\ \citenamefont {Knill}}]{QRNG2}%
  \BibitemOpen
  \bibfield  {author} {\bibinfo {author} {\bibfnamefont {Y.}~\bibnamefont
  {Zhang}}, \bibinfo {author} {\bibfnamefont {L.~K.}\ \bibnamefont {Shalm}},
  \bibinfo {author} {\bibfnamefont {J.~C.}\ \bibnamefont {Bienfang}}, \bibinfo
  {author} {\bibfnamefont {M.~J.}\ \bibnamefont {Stevens}}, \bibinfo {author}
  {\bibfnamefont {M.~D.}\ \bibnamefont {Mazurek}}, \bibinfo {author}
  {\bibfnamefont {S.~W.}\ \bibnamefont {Nam}}, \bibinfo {author} {\bibfnamefont
  {C.}~\bibnamefont {Abell\'an}}, \bibinfo {author} {\bibfnamefont
  {W.}~\bibnamefont {Amaya}}, \bibinfo {author} {\bibfnamefont {M.~W.}\
  \bibnamefont {Mitchell}}, \bibinfo {author} {\bibfnamefont {H.}~\bibnamefont
  {Fu}}, \bibinfo {author} {\bibfnamefont {C.~A.}\ \bibnamefont {Miller}},
  \bibinfo {author} {\bibfnamefont {A.}~\bibnamefont {Mink}}, \ and\ \bibinfo
  {author} {\bibfnamefont {E.}~\bibnamefont {Knill}},\ }\href {\doibase
  10.1103/PhysRevLett.124.010505} {\bibfield  {journal} {\bibinfo  {journal}
  {Phys. Rev. Lett.}\ }\textbf {\bibinfo {volume} {124}},\ \bibinfo {pages}
  {010505} (\bibinfo {year} {2020}{\natexlab{a}})}\BibitemShut {NoStop}%
\bibitem [{\citenamefont {Fehr}\ \emph {et~al.}(2013)\citenamefont {Fehr},
  \citenamefont {Gelles},\ and\ \citenamefont {Schaffner}}]{Fehr13}%
  \BibitemOpen
  \bibfield  {author} {\bibinfo {author} {\bibfnamefont {S.}~\bibnamefont
  {Fehr}}, \bibinfo {author} {\bibfnamefont {R.}~\bibnamefont {Gelles}}, \ and\
  \bibinfo {author} {\bibfnamefont {C.}~\bibnamefont {Schaffner}},\ }\href
  {\doibase 10.1103/PhysRevA.87.012335} {\bibfield  {journal} {\bibinfo
  {journal} {Phys. Rev. A}\ }\textbf {\bibinfo {volume} {87}},\ \bibinfo
  {pages} {012335} (\bibinfo {year} {2013})}\BibitemShut {NoStop}%
\bibitem [{\citenamefont {Coudron}\ and\ \citenamefont
  {Yuen}(2014)}]{coudron2014infinite}%
  \BibitemOpen
  \bibfield  {author} {\bibinfo {author} {\bibfnamefont {M.}~\bibnamefont
  {Coudron}}\ and\ \bibinfo {author} {\bibfnamefont {H.}~\bibnamefont {Yuen}},\
  }in\ \href {\doibase 10.1145/2591796.2591873} {\emph {\bibinfo {booktitle}
  {Proceedings of the Forty-Sixth Annual ACM Symposium on Theory of
  Computing}}}\ (\bibinfo {year} {2014})\ p.\ \bibinfo {pages}
  {427–436}\BibitemShut {NoStop}%
\bibitem [{\citenamefont {Miller}\ and\ \citenamefont {Shi}(2014)}]{MS1}%
  \BibitemOpen
  \bibfield  {author} {\bibinfo {author} {\bibfnamefont {C.~A.}\ \bibnamefont
  {Miller}}\ and\ \bibinfo {author} {\bibfnamefont {Y.}~\bibnamefont {Shi}},\
  }in\ \href {\doibase 10.1145/2591796.2591843} {\emph {\bibinfo {booktitle}
  {Proceedings of the Forty-Sixth Annual ACM Symposium on Theory of
  Computing}}}\ (\bibinfo {year} {2014})\ p.\ \bibinfo {pages}
  {417–426}\BibitemShut {NoStop}%
\bibitem [{\citenamefont {Miller}\ and\ \citenamefont
  {Shi}(2017)}]{miller2017universal}%
  \BibitemOpen
  \bibfield  {author} {\bibinfo {author} {\bibfnamefont {C.~A.}\ \bibnamefont
  {Miller}}\ and\ \bibinfo {author} {\bibfnamefont {Y.}~\bibnamefont {Shi}},\
  }\href {\doibase 10.1137/15M1044333} {\bibfield  {journal} {\bibinfo
  {journal} {SIAM Journal on Computing}\ }\textbf {\bibinfo {volume} {46}},\
  \bibinfo {pages} {1304} (\bibinfo {year} {2017})}\BibitemShut {NoStop}%
\bibitem [{\citenamefont {Vazirani}\ and\ \citenamefont {Vidick}(2012)}]{VV}%
  \BibitemOpen
  \bibfield  {author} {\bibinfo {author} {\bibfnamefont {U.}~\bibnamefont
  {Vazirani}}\ and\ \bibinfo {author} {\bibfnamefont {T.}~\bibnamefont
  {Vidick}},\ }in\ \href {\doibase 10.1145/2213977.2213984} {\emph {\bibinfo
  {booktitle} {Proceedings of the Forty-Fourth Annual ACM Symposium on Theory
  of Computing}}}\ (\bibinfo {year} {2012})\ p.\ \bibinfo {pages}
  {61–76}\BibitemShut {NoStop}%
\bibitem [{\citenamefont {{Brown}}\ \emph {et~al.}(2020)\citenamefont
  {{Brown}}, \citenamefont {{Ragy}},\ and\ \citenamefont {{Colbeck}}}]{BRC}%
  \BibitemOpen
  \bibfield  {author} {\bibinfo {author} {\bibfnamefont {P.~J.}\ \bibnamefont
  {{Brown}}}, \bibinfo {author} {\bibfnamefont {S.}~\bibnamefont {{Ragy}}}, \
  and\ \bibinfo {author} {\bibfnamefont {R.}~\bibnamefont {{Colbeck}}},\ }\href
  {\doibase 10.1109/TIT.2019.2960252} {\bibfield  {journal} {\bibinfo
  {journal} {IEEE Transactions on Information Theory}\ }\textbf {\bibinfo
  {volume} {66}},\ \bibinfo {pages} {2964} (\bibinfo {year}
  {2020})}\BibitemShut {NoStop}%
\bibitem [{\citenamefont {Dupuis}\ \emph {et~al.}(2020)\citenamefont {Dupuis},
  \citenamefont {Fawzi},\ and\ \citenamefont {Renner}}]{DFR}%
  \BibitemOpen
  \bibfield  {author} {\bibinfo {author} {\bibfnamefont {F.}~\bibnamefont
  {Dupuis}}, \bibinfo {author} {\bibfnamefont {O.}~\bibnamefont {Fawzi}}, \
  and\ \bibinfo {author} {\bibfnamefont {R.}~\bibnamefont {Renner}},\ }\href
  {\doibase 10.1007/s00220-020-03839-5} {\bibfield  {journal} {\bibinfo
  {journal} {Communications in Mathematical Physics}\ }\textbf {\bibinfo
  {volume} {379}},\ \bibinfo {pages} {867} (\bibinfo {year}
  {2020})}\BibitemShut {NoStop}%
\bibitem [{\citenamefont {{Dupuis}}\ and\ \citenamefont {{Fawzi}}(2019)}]{DF}%
  \BibitemOpen
  \bibfield  {author} {\bibinfo {author} {\bibfnamefont {F.}~\bibnamefont
  {{Dupuis}}}\ and\ \bibinfo {author} {\bibfnamefont {O.}~\bibnamefont
  {{Fawzi}}},\ }\href {\doibase 10.1109/TIT.2019.2929564} {\bibfield  {journal}
  {\bibinfo  {journal} {IEEE Transactions on Information Theory}\ }\textbf
  {\bibinfo {volume} {65}},\ \bibinfo {pages} {7596} (\bibinfo {year}
  {2019})}\BibitemShut {NoStop}%
\bibitem [{\citenamefont {Ac{\'i}n}\ and\ \citenamefont
  {Masanes}(2016)}]{acin2016certified}%
  \BibitemOpen
  \bibfield  {author} {\bibinfo {author} {\bibfnamefont {A.}~\bibnamefont
  {Ac{\'i}n}}\ and\ \bibinfo {author} {\bibfnamefont {L.}~\bibnamefont
  {Masanes}},\ }\href {\doibase 10.1038/nature20119} {\bibfield  {journal}
  {\bibinfo  {journal} {Nature}\ }\textbf {\bibinfo {volume} {540}},\ \bibinfo
  {pages} {213} (\bibinfo {year} {2016})}\BibitemShut {NoStop}%
\bibitem [{\citenamefont {Herrero-Collantes}\ and\ \citenamefont
  {Garcia-Escartin}(2017)}]{herrero2017quantum}%
  \BibitemOpen
  \bibfield  {author} {\bibinfo {author} {\bibfnamefont {M.}~\bibnamefont
  {Herrero-Collantes}}\ and\ \bibinfo {author} {\bibfnamefont {J.~C.}\
  \bibnamefont {Garcia-Escartin}},\ }\href {\doibase
  10.1103/RevModPhys.89.015004} {\bibfield  {journal} {\bibinfo  {journal}
  {Rev. Mod. Phys.}\ }\textbf {\bibinfo {volume} {89}},\ \bibinfo {pages}
  {015004} (\bibinfo {year} {2017})}\BibitemShut {NoStop}%
\bibitem [{\citenamefont {Gerhardt}\ \emph {et~al.}(2011)\citenamefont
  {Gerhardt}, \citenamefont {Liu}, \citenamefont {Lamas-Linares}, \citenamefont
  {Skaar}, \citenamefont {Kurtsiefer},\ and\ \citenamefont {Makarov}}]{GLLSKM}%
  \BibitemOpen
  \bibfield  {author} {\bibinfo {author} {\bibfnamefont {I.}~\bibnamefont
  {Gerhardt}}, \bibinfo {author} {\bibfnamefont {Q.}~\bibnamefont {Liu}},
  \bibinfo {author} {\bibfnamefont {A.}~\bibnamefont {Lamas-Linares}}, \bibinfo
  {author} {\bibfnamefont {J.}~\bibnamefont {Skaar}}, \bibinfo {author}
  {\bibfnamefont {C.}~\bibnamefont {Kurtsiefer}}, \ and\ \bibinfo {author}
  {\bibfnamefont {V.}~\bibnamefont {Makarov}},\ }\href {\doibase
  10.1038/ncomms1348} {\bibfield  {journal} {\bibinfo  {journal} {Nature
  Communications}\ }\textbf {\bibinfo {volume} {2}},\ \bibinfo {pages} {349}
  (\bibinfo {year} {2011})}\BibitemShut {NoStop}%
\bibitem [{\citenamefont {Clauser}\ \emph {et~al.}(1969)\citenamefont
  {Clauser}, \citenamefont {Horne}, \citenamefont {Shimony},\ and\
  \citenamefont {Holt}}]{CHSH}%
  \BibitemOpen
  \bibfield  {author} {\bibinfo {author} {\bibfnamefont {J.~F.}\ \bibnamefont
  {Clauser}}, \bibinfo {author} {\bibfnamefont {M.~A.}\ \bibnamefont {Horne}},
  \bibinfo {author} {\bibfnamefont {A.}~\bibnamefont {Shimony}}, \ and\
  \bibinfo {author} {\bibfnamefont {R.~A.}\ \bibnamefont {Holt}},\ }\href
  {\doibase 10.1103/PhysRevLett.23.880} {\bibfield  {journal} {\bibinfo
  {journal} {Phys. Rev. Lett.}\ }\textbf {\bibinfo {volume} {23}},\ \bibinfo
  {pages} {880} (\bibinfo {year} {1969})}\BibitemShut {NoStop}%
\bibitem [{\citenamefont {Murta}\ \emph {et~al.}(2019)\citenamefont {Murta},
  \citenamefont {van Dam}, \citenamefont {Ribeiro}, \citenamefont {Hanson},\
  and\ \citenamefont {Wehner}}]{murta2019towards}%
  \BibitemOpen
  \bibfield  {author} {\bibinfo {author} {\bibfnamefont {G.}~\bibnamefont
  {Murta}}, \bibinfo {author} {\bibfnamefont {S.~B.}\ \bibnamefont {van Dam}},
  \bibinfo {author} {\bibfnamefont {J.}~\bibnamefont {Ribeiro}}, \bibinfo
  {author} {\bibfnamefont {R.}~\bibnamefont {Hanson}}, \ and\ \bibinfo {author}
  {\bibfnamefont {S.}~\bibnamefont {Wehner}},\ }\href {\doibase
  10.1088/2058-9565/ab2819} {\bibfield  {journal} {\bibinfo  {journal} {Quantum
  Science and Technology}\ }\textbf {\bibinfo {volume} {4}},\ \bibinfo {pages}
  {035011} (\bibinfo {year} {2019})}\BibitemShut {NoStop}%
\bibitem [{\citenamefont {{Konig}}\ and\ \citenamefont
  {{Renner}}(2011)}]{KR2011}%
  \BibitemOpen
  \bibfield  {author} {\bibinfo {author} {\bibfnamefont {R.}~\bibnamefont
  {{Konig}}}\ and\ \bibinfo {author} {\bibfnamefont {R.}~\bibnamefont
  {{Renner}}},\ }\href {\doibase 10.1109/TIT.2011.2146730} {\bibfield
  {journal} {\bibinfo  {journal} {IEEE Transactions on Information Theory}\
  }\textbf {\bibinfo {volume} {57}},\ \bibinfo {pages} {4760} (\bibinfo {year}
  {2011})}\BibitemShut {NoStop}%
\bibitem [{\citenamefont {Hensen}\ \emph {et~al.}(2015)\citenamefont {Hensen},
  \citenamefont {Bernien}, \citenamefont {Dr{\'e}au}, \citenamefont {Reiserer},
  \citenamefont {Kalb}, \citenamefont {Blok}, \citenamefont {Ruitenberg},
  \citenamefont {Vermeulen}, \citenamefont {Schouten}, \citenamefont
  {Abell{\'a}n}, \citenamefont {Amaya}, \citenamefont {Pruneri}, \citenamefont
  {Mitchell}, \citenamefont {Markham}, \citenamefont {Twitchen}, \citenamefont
  {Elkouss}, \citenamefont {Wehner}, \citenamefont {Taminiau},\ and\
  \citenamefont {Hanson}}]{hensen2015loophole}%
  \BibitemOpen
  \bibfield  {author} {\bibinfo {author} {\bibfnamefont {B.}~\bibnamefont
  {Hensen}}, \bibinfo {author} {\bibfnamefont {H.}~\bibnamefont {Bernien}},
  \bibinfo {author} {\bibfnamefont {A.~E.}\ \bibnamefont {Dr{\'e}au}}, \bibinfo
  {author} {\bibfnamefont {A.}~\bibnamefont {Reiserer}}, \bibinfo {author}
  {\bibfnamefont {N.}~\bibnamefont {Kalb}}, \bibinfo {author} {\bibfnamefont
  {M.~S.}\ \bibnamefont {Blok}}, \bibinfo {author} {\bibfnamefont
  {J.}~\bibnamefont {Ruitenberg}}, \bibinfo {author} {\bibfnamefont {R.~F.~L.}\
  \bibnamefont {Vermeulen}}, \bibinfo {author} {\bibfnamefont {R.~N.}\
  \bibnamefont {Schouten}}, \bibinfo {author} {\bibfnamefont {C.}~\bibnamefont
  {Abell{\'a}n}}, \bibinfo {author} {\bibfnamefont {W.}~\bibnamefont {Amaya}},
  \bibinfo {author} {\bibfnamefont {V.}~\bibnamefont {Pruneri}}, \bibinfo
  {author} {\bibfnamefont {M.~W.}\ \bibnamefont {Mitchell}}, \bibinfo {author}
  {\bibfnamefont {M.}~\bibnamefont {Markham}}, \bibinfo {author} {\bibfnamefont
  {D.~J.}\ \bibnamefont {Twitchen}}, \bibinfo {author} {\bibfnamefont
  {D.}~\bibnamefont {Elkouss}}, \bibinfo {author} {\bibfnamefont
  {S.}~\bibnamefont {Wehner}}, \bibinfo {author} {\bibfnamefont {T.~H.}\
  \bibnamefont {Taminiau}}, \ and\ \bibinfo {author} {\bibfnamefont
  {R.}~\bibnamefont {Hanson}},\ }\href {\doibase 10.1038/nature15759}
  {\bibfield  {journal} {\bibinfo  {journal} {Nature}\ }\textbf {\bibinfo
  {volume} {526}},\ \bibinfo {pages} {682} (\bibinfo {year}
  {2015})}\BibitemShut {NoStop}%
\bibitem [{\citenamefont {Shalm}\ \emph {et~al.}(2015)\citenamefont {Shalm},
  \citenamefont {Meyer-Scott}, \citenamefont {Christensen}, \citenamefont
  {Bierhorst}, \citenamefont {Wayne}, \citenamefont {Stevens}, \citenamefont
  {Gerrits}, \citenamefont {Glancy}, \citenamefont {Hamel}, \citenamefont
  {Allman}, \citenamefont {Coakley}, \citenamefont {Dyer}, \citenamefont
  {Hodge}, \citenamefont {Lita}, \citenamefont {Verma}, \citenamefont
  {Lambrocco}, \citenamefont {Tortorici}, \citenamefont {Migdall},
  \citenamefont {Zhang}, \citenamefont {Kumor}, \citenamefont {Farr},
  \citenamefont {Marsili}, \citenamefont {Shaw}, \citenamefont {Stern},
  \citenamefont {Abell\'an}, \citenamefont {Amaya}, \citenamefont {Pruneri},
  \citenamefont {Jennewein}, \citenamefont {Mitchell}, \citenamefont {Kwiat},
  \citenamefont {Bienfang}, \citenamefont {Mirin}, \citenamefont {Knill},\ and\
  \citenamefont {Nam}}]{shalm2015strong}%
  \BibitemOpen
  \bibfield  {author} {\bibinfo {author} {\bibfnamefont {L.~K.}\ \bibnamefont
  {Shalm}}, \bibinfo {author} {\bibfnamefont {E.}~\bibnamefont {Meyer-Scott}},
  \bibinfo {author} {\bibfnamefont {B.~G.}\ \bibnamefont {Christensen}},
  \bibinfo {author} {\bibfnamefont {P.}~\bibnamefont {Bierhorst}}, \bibinfo
  {author} {\bibfnamefont {M.~A.}\ \bibnamefont {Wayne}}, \bibinfo {author}
  {\bibfnamefont {M.~J.}\ \bibnamefont {Stevens}}, \bibinfo {author}
  {\bibfnamefont {T.}~\bibnamefont {Gerrits}}, \bibinfo {author} {\bibfnamefont
  {S.}~\bibnamefont {Glancy}}, \bibinfo {author} {\bibfnamefont {D.~R.}\
  \bibnamefont {Hamel}}, \bibinfo {author} {\bibfnamefont {M.~S.}\ \bibnamefont
  {Allman}}, \bibinfo {author} {\bibfnamefont {K.~J.}\ \bibnamefont {Coakley}},
  \bibinfo {author} {\bibfnamefont {S.~D.}\ \bibnamefont {Dyer}}, \bibinfo
  {author} {\bibfnamefont {C.}~\bibnamefont {Hodge}}, \bibinfo {author}
  {\bibfnamefont {A.~E.}\ \bibnamefont {Lita}}, \bibinfo {author}
  {\bibfnamefont {V.~B.}\ \bibnamefont {Verma}}, \bibinfo {author}
  {\bibfnamefont {C.}~\bibnamefont {Lambrocco}}, \bibinfo {author}
  {\bibfnamefont {E.}~\bibnamefont {Tortorici}}, \bibinfo {author}
  {\bibfnamefont {A.~L.}\ \bibnamefont {Migdall}}, \bibinfo {author}
  {\bibfnamefont {Y.}~\bibnamefont {Zhang}}, \bibinfo {author} {\bibfnamefont
  {D.~R.}\ \bibnamefont {Kumor}}, \bibinfo {author} {\bibfnamefont {W.~H.}\
  \bibnamefont {Farr}}, \bibinfo {author} {\bibfnamefont {F.}~\bibnamefont
  {Marsili}}, \bibinfo {author} {\bibfnamefont {M.~D.}\ \bibnamefont {Shaw}},
  \bibinfo {author} {\bibfnamefont {J.~A.}\ \bibnamefont {Stern}}, \bibinfo
  {author} {\bibfnamefont {C.}~\bibnamefont {Abell\'an}}, \bibinfo {author}
  {\bibfnamefont {W.}~\bibnamefont {Amaya}}, \bibinfo {author} {\bibfnamefont
  {V.}~\bibnamefont {Pruneri}}, \bibinfo {author} {\bibfnamefont
  {T.}~\bibnamefont {Jennewein}}, \bibinfo {author} {\bibfnamefont {M.~W.}\
  \bibnamefont {Mitchell}}, \bibinfo {author} {\bibfnamefont {P.~G.}\
  \bibnamefont {Kwiat}}, \bibinfo {author} {\bibfnamefont {J.~C.}\ \bibnamefont
  {Bienfang}}, \bibinfo {author} {\bibfnamefont {R.~P.}\ \bibnamefont {Mirin}},
  \bibinfo {author} {\bibfnamefont {E.}~\bibnamefont {Knill}}, \ and\ \bibinfo
  {author} {\bibfnamefont {S.~W.}\ \bibnamefont {Nam}},\ }\href {\doibase
  10.1103/PhysRevLett.115.250402} {\bibfield  {journal} {\bibinfo  {journal}
  {Phys. Rev. Lett.}\ }\textbf {\bibinfo {volume} {115}},\ \bibinfo {pages}
  {250402} (\bibinfo {year} {2015})}\BibitemShut {NoStop}%
\bibitem [{\citenamefont {Giustina}\ \emph {et~al.}(2015)\citenamefont
  {Giustina}, \citenamefont {Versteegh}, \citenamefont {Wengerowsky},
  \citenamefont {Handsteiner}, \citenamefont {Hochrainer}, \citenamefont
  {Phelan}, \citenamefont {Steinlechner}, \citenamefont {Kofler}, \citenamefont
  {Larsson}, \citenamefont {Abell\'an}, \citenamefont {Amaya}, \citenamefont
  {Pruneri}, \citenamefont {Mitchell}, \citenamefont {Beyer}, \citenamefont
  {Gerrits}, \citenamefont {Lita}, \citenamefont {Shalm}, \citenamefont {Nam},
  \citenamefont {Scheidl}, \citenamefont {Ursin}, \citenamefont {Wittmann},\
  and\ \citenamefont {Zeilinger}}]{giustina2015Significant}%
  \BibitemOpen
  \bibfield  {author} {\bibinfo {author} {\bibfnamefont {M.}~\bibnamefont
  {Giustina}}, \bibinfo {author} {\bibfnamefont {M.~A.~M.}\ \bibnamefont
  {Versteegh}}, \bibinfo {author} {\bibfnamefont {S.}~\bibnamefont
  {Wengerowsky}}, \bibinfo {author} {\bibfnamefont {J.}~\bibnamefont
  {Handsteiner}}, \bibinfo {author} {\bibfnamefont {A.}~\bibnamefont
  {Hochrainer}}, \bibinfo {author} {\bibfnamefont {K.}~\bibnamefont {Phelan}},
  \bibinfo {author} {\bibfnamefont {F.}~\bibnamefont {Steinlechner}}, \bibinfo
  {author} {\bibfnamefont {J.}~\bibnamefont {Kofler}}, \bibinfo {author}
  {\bibfnamefont {J.-A.}\ \bibnamefont {Larsson}}, \bibinfo {author}
  {\bibfnamefont {C.}~\bibnamefont {Abell\'an}}, \bibinfo {author}
  {\bibfnamefont {W.}~\bibnamefont {Amaya}}, \bibinfo {author} {\bibfnamefont
  {V.}~\bibnamefont {Pruneri}}, \bibinfo {author} {\bibfnamefont {M.~W.}\
  \bibnamefont {Mitchell}}, \bibinfo {author} {\bibfnamefont {J.}~\bibnamefont
  {Beyer}}, \bibinfo {author} {\bibfnamefont {T.}~\bibnamefont {Gerrits}},
  \bibinfo {author} {\bibfnamefont {A.~E.}\ \bibnamefont {Lita}}, \bibinfo
  {author} {\bibfnamefont {L.~K.}\ \bibnamefont {Shalm}}, \bibinfo {author}
  {\bibfnamefont {S.~W.}\ \bibnamefont {Nam}}, \bibinfo {author} {\bibfnamefont
  {T.}~\bibnamefont {Scheidl}}, \bibinfo {author} {\bibfnamefont
  {R.}~\bibnamefont {Ursin}}, \bibinfo {author} {\bibfnamefont
  {B.}~\bibnamefont {Wittmann}}, \ and\ \bibinfo {author} {\bibfnamefont
  {A.}~\bibnamefont {Zeilinger}},\ }\href {\doibase
  10.1103/PhysRevLett.115.250401} {\bibfield  {journal} {\bibinfo  {journal}
  {Phys. Rev. Lett.}\ }\textbf {\bibinfo {volume} {115}},\ \bibinfo {pages}
  {250401} (\bibinfo {year} {2015})}\BibitemShut {NoStop}%
\bibitem [{\citenamefont {Rosenfeld}\ \emph {et~al.}(2017)\citenamefont
  {Rosenfeld}, \citenamefont {Burchardt}, \citenamefont {Garthoff},
  \citenamefont {Redeker}, \citenamefont {Ortegel}, \citenamefont {Rau},\ and\
  \citenamefont {Weinfurter}}]{rosenfeld2017event}%
  \BibitemOpen
  \bibfield  {author} {\bibinfo {author} {\bibfnamefont {W.}~\bibnamefont
  {Rosenfeld}}, \bibinfo {author} {\bibfnamefont {D.}~\bibnamefont
  {Burchardt}}, \bibinfo {author} {\bibfnamefont {R.}~\bibnamefont {Garthoff}},
  \bibinfo {author} {\bibfnamefont {K.}~\bibnamefont {Redeker}}, \bibinfo
  {author} {\bibfnamefont {N.}~\bibnamefont {Ortegel}}, \bibinfo {author}
  {\bibfnamefont {M.}~\bibnamefont {Rau}}, \ and\ \bibinfo {author}
  {\bibfnamefont {H.}~\bibnamefont {Weinfurter}},\ }\href {\doibase
  10.1103/PhysRevLett.119.010402} {\bibfield  {journal} {\bibinfo  {journal}
  {Phys. Rev. Lett.}\ }\textbf {\bibinfo {volume} {119}},\ \bibinfo {pages}
  {010402} (\bibinfo {year} {2017})}\BibitemShut {NoStop}%
\bibitem [{\citenamefont {Li}\ \emph {et~al.}(2018)\citenamefont {Li},
  \citenamefont {Wu}, \citenamefont {Zhang}, \citenamefont {Liu}, \citenamefont
  {Bai}, \citenamefont {Liu}, \citenamefont {Zhang}, \citenamefont {Zhao},
  \citenamefont {Li}, \citenamefont {Wang}, \citenamefont {You}, \citenamefont
  {Munro}, \citenamefont {Yin}, \citenamefont {Zhang}, \citenamefont {Peng},
  \citenamefont {Ma}, \citenamefont {Zhang}, \citenamefont {Fan},\ and\
  \citenamefont {Pan}}]{LiPRL2018}%
  \BibitemOpen
  \bibfield  {author} {\bibinfo {author} {\bibfnamefont {M.-H.}\ \bibnamefont
  {Li}}, \bibinfo {author} {\bibfnamefont {C.}~\bibnamefont {Wu}}, \bibinfo
  {author} {\bibfnamefont {Y.}~\bibnamefont {Zhang}}, \bibinfo {author}
  {\bibfnamefont {W.-Z.}\ \bibnamefont {Liu}}, \bibinfo {author} {\bibfnamefont
  {B.}~\bibnamefont {Bai}}, \bibinfo {author} {\bibfnamefont {Y.}~\bibnamefont
  {Liu}}, \bibinfo {author} {\bibfnamefont {W.}~\bibnamefont {Zhang}}, \bibinfo
  {author} {\bibfnamefont {Q.}~\bibnamefont {Zhao}}, \bibinfo {author}
  {\bibfnamefont {H.}~\bibnamefont {Li}}, \bibinfo {author} {\bibfnamefont
  {Z.}~\bibnamefont {Wang}}, \bibinfo {author} {\bibfnamefont {L.}~\bibnamefont
  {You}}, \bibinfo {author} {\bibfnamefont {W.~J.}\ \bibnamefont {Munro}},
  \bibinfo {author} {\bibfnamefont {J.}~\bibnamefont {Yin}}, \bibinfo {author}
  {\bibfnamefont {J.}~\bibnamefont {Zhang}}, \bibinfo {author} {\bibfnamefont
  {C.-Z.}\ \bibnamefont {Peng}}, \bibinfo {author} {\bibfnamefont
  {X.}~\bibnamefont {Ma}}, \bibinfo {author} {\bibfnamefont {Q.}~\bibnamefont
  {Zhang}}, \bibinfo {author} {\bibfnamefont {J.}~\bibnamefont {Fan}}, \ and\
  \bibinfo {author} {\bibfnamefont {J.-W.}\ \bibnamefont {Pan}},\ }\href
  {\doibase 10.1103/PhysRevLett.121.080404} {\bibfield  {journal} {\bibinfo
  {journal} {Phys. Rev. Lett.}\ }\textbf {\bibinfo {volume} {121}},\ \bibinfo
  {pages} {080404} (\bibinfo {year} {2018})}\BibitemShut {NoStop}%
\bibitem [{\citenamefont {Barrett}\ \emph {et~al.}(2013)\citenamefont
  {Barrett}, \citenamefont {Colbeck},\ and\ \citenamefont {Kent}}]{bckone}%
  \BibitemOpen
  \bibfield  {author} {\bibinfo {author} {\bibfnamefont {J.}~\bibnamefont
  {Barrett}}, \bibinfo {author} {\bibfnamefont {R.}~\bibnamefont {Colbeck}}, \
  and\ \bibinfo {author} {\bibfnamefont {A.}~\bibnamefont {Kent}},\ }\href
  {\doibase 10.1103/PhysRevLett.110.010503} {\bibfield  {journal} {\bibinfo
  {journal} {Phys. Rev. Lett.}\ }\textbf {\bibinfo {volume} {110}},\ \bibinfo
  {pages} {010503} (\bibinfo {year} {2013})}\BibitemShut {NoStop}%
\bibitem [{\citenamefont {Colbeck}\ and\ \citenamefont
  {Renner}(2012)}]{CR_free}%
  \BibitemOpen
  \bibfield  {author} {\bibinfo {author} {\bibfnamefont {R.}~\bibnamefont
  {Colbeck}}\ and\ \bibinfo {author} {\bibfnamefont {R.}~\bibnamefont
  {Renner}},\ }\href {\doibase 10.1038/nphys2300} {\bibfield  {journal}
  {\bibinfo  {journal} {Nature Physics}\ }\textbf {\bibinfo {volume} {8}},\
  \bibinfo {pages} {450} (\bibinfo {year} {2012})}\BibitemShut {NoStop}%
\bibitem [{\citenamefont {Canetti}(2000)}]{Canetti}%
  \BibitemOpen
  \bibfield  {author} {\bibinfo {author} {\bibfnamefont {R.}~\bibnamefont
  {Canetti}},\ }\href {\doibase 10.1007/s001459910006} {\bibfield  {journal}
  {\bibinfo  {journal} {Journal of Cryptology}\ }\textbf {\bibinfo {volume}
  {13}},\ \bibinfo {pages} {143} (\bibinfo {year} {2000})}\BibitemShut
  {NoStop}%
\bibitem [{\citenamefont {Ben-Or}\ and\ \citenamefont
  {Mayers}(2004)}]{Ben-OrMayers}%
  \BibitemOpen
  \bibfield  {author} {\bibinfo {author} {\bibfnamefont {M.}~\bibnamefont
  {Ben-Or}}\ and\ \bibinfo {author} {\bibfnamefont {D.}~\bibnamefont
  {Mayers}},\ }\href@noop {} {\enquote {\bibinfo {title} {General security
  definition and composability for quantum \& classical protocols},}\ }\bibinfo
  {howpublished} {e-print
  \href{https://arxiv.org/abs/quant-ph/0409062}{quant-ph/0409062}} (\bibinfo
  {year} {2004})\BibitemShut {NoStop}%
\bibitem [{\citenamefont {Portmann}\ and\ \citenamefont
  {Renner}(2014)}]{PR2014}%
  \BibitemOpen
  \bibfield  {author} {\bibinfo {author} {\bibfnamefont {C.}~\bibnamefont
  {Portmann}}\ and\ \bibinfo {author} {\bibfnamefont {R.}~\bibnamefont
  {Renner}},\ }\href@noop {} {\enquote {\bibinfo {title} {Cryptographic
  security of quantum key distribution},}\ }\bibinfo {howpublished} {e-print
  \href{https://arxiv.org/abs/1409.3525}{arXiv:1409.3525}} (\bibinfo {year}
  {2014})\BibitemShut {NoStop}%
\bibitem [{\citenamefont {Renner}(2005)}]{Renner}%
  \BibitemOpen
  \bibfield  {author} {\bibinfo {author} {\bibfnamefont {R.}~\bibnamefont
  {Renner}},\ }\emph {\bibinfo {title} {Security of Quantum Key
  Distribution}},\ \href@noop {} {Ph.D. thesis},\ \bibinfo  {school} {Swiss
  Federal Institute of Technology, Zurich} (\bibinfo {year} {2005}),\ \bibinfo
  {note} {also available as
  \href{https://arxiv.org/abs/quant-ph/0512258}{quant-ph/0512258}}\BibitemShut
  {NoStop}%
\bibitem [{\citenamefont {{Konig}}\ \emph {et~al.}(2009)\citenamefont
  {{Konig}}, \citenamefont {{Renner}},\ and\ \citenamefont
  {{Schaffner}}}]{KRS}%
  \BibitemOpen
  \bibfield  {author} {\bibinfo {author} {\bibfnamefont {R.}~\bibnamefont
  {{Konig}}}, \bibinfo {author} {\bibfnamefont {R.}~\bibnamefont {{Renner}}}, \
  and\ \bibinfo {author} {\bibfnamefont {C.}~\bibnamefont {{Schaffner}}},\
  }\href {\doibase 10.1109/TIT.2009.2025545} {\bibfield  {journal} {\bibinfo
  {journal} {IEEE Transactions on Information Theory}\ }\textbf {\bibinfo
  {volume} {55}},\ \bibinfo {pages} {4337} (\bibinfo {year}
  {2009})}\BibitemShut {NoStop}%
\bibitem [{\citenamefont {{Tomamichel}}\ \emph {et~al.}(2010)\citenamefont
  {{Tomamichel}}, \citenamefont {{Colbeck}},\ and\ \citenamefont
  {{Renner}}}]{TCR2}%
  \BibitemOpen
  \bibfield  {author} {\bibinfo {author} {\bibfnamefont {M.}~\bibnamefont
  {{Tomamichel}}}, \bibinfo {author} {\bibfnamefont {R.}~\bibnamefont
  {{Colbeck}}}, \ and\ \bibinfo {author} {\bibfnamefont {R.}~\bibnamefont
  {{Renner}}},\ }\href {\doibase 10.1109/TIT.2010.2054130} {\bibfield
  {journal} {\bibinfo  {journal} {IEEE Transactions on Information Theory}\
  }\textbf {\bibinfo {volume} {56}},\ \bibinfo {pages} {4674} (\bibinfo {year}
  {2010})}\BibitemShut {NoStop}%
\bibitem [{\citenamefont {{Te Sun Hao}}\ and\ \citenamefont
  {{Hoshi}}(1997)}]{HaoHoshi}%
  \BibitemOpen
  \bibfield  {author} {\bibinfo {author} {\bibnamefont {{Te Sun Hao}}}\ and\
  \bibinfo {author} {\bibfnamefont {M.}~\bibnamefont {{Hoshi}}},\ }\href
  {\doibase 10.1109/18.556116} {\bibfield  {journal} {\bibinfo  {journal} {IEEE
  Transactions on Information Theory}\ }\textbf {\bibinfo {volume} {43}},\
  \bibinfo {pages} {599} (\bibinfo {year} {1997})}\BibitemShut {NoStop}%
\bibitem [{\citenamefont {Colbeck}\ and\ \citenamefont
  {Renner}(2011)}]{CR_ext}%
  \BibitemOpen
  \bibfield  {author} {\bibinfo {author} {\bibfnamefont {R.}~\bibnamefont
  {Colbeck}}\ and\ \bibinfo {author} {\bibfnamefont {R.}~\bibnamefont
  {Renner}},\ }\href {\doibase 10.1038/ncomms1416} {\bibfield  {journal}
  {\bibinfo  {journal} {Nature Communications}\ }\textbf {\bibinfo {volume}
  {2}},\ \bibinfo {pages} {411} (\bibinfo {year} {2011})}\BibitemShut {NoStop}%
\bibitem [{\citenamefont {Colbeck}\ and\ \citenamefont
  {Renner}(2016)}]{CR_book1}%
  \BibitemOpen
  \bibfield  {author} {\bibinfo {author} {\bibfnamefont {R.}~\bibnamefont
  {Colbeck}}\ and\ \bibinfo {author} {\bibfnamefont {R.}~\bibnamefont
  {Renner}},\ }in\ \href@noop {} {\emph {\bibinfo {booktitle} {Quantum Theory:
  Informational Foundations and Foils}}}\ (\bibinfo {year} {2016})\ pp.\
  \bibinfo {pages} {497--528}\BibitemShut {NoStop}%
\bibitem [{\citenamefont {Nisan}\ and\ \citenamefont {Zuckerman}(1996)}]{NZ96}%
  \BibitemOpen
  \bibfield  {author} {\bibinfo {author} {\bibfnamefont {N.}~\bibnamefont
  {Nisan}}\ and\ \bibinfo {author} {\bibfnamefont {D.}~\bibnamefont
  {Zuckerman}},\ }\href {\doibase https://doi.org/10.1006/jcss.1996.0004}
  {\bibfield  {journal} {\bibinfo  {journal} {Journal of Computer and System
  Sciences}\ }\textbf {\bibinfo {volume} {52}},\ \bibinfo {pages} {43 }
  (\bibinfo {year} {1996})}\BibitemShut {NoStop}%
\bibitem [{\citenamefont {De}\ \emph {et~al.}(2012)\citenamefont {De},
  \citenamefont {Portmann}, \citenamefont {Vidick},\ and\ \citenamefont
  {Renner}}]{DPVR}%
  \BibitemOpen
  \bibfield  {author} {\bibinfo {author} {\bibfnamefont {A.}~\bibnamefont
  {De}}, \bibinfo {author} {\bibfnamefont {C.}~\bibnamefont {Portmann}},
  \bibinfo {author} {\bibfnamefont {T.}~\bibnamefont {Vidick}}, \ and\ \bibinfo
  {author} {\bibfnamefont {R.}~\bibnamefont {Renner}},\ }\href {\doibase
  10.1137/100813683} {\bibfield  {journal} {\bibinfo  {journal} {SIAM Journal
  on Computing}\ }\textbf {\bibinfo {volume} {41}},\ \bibinfo {pages} {915}
  (\bibinfo {year} {2012})}\BibitemShut {NoStop}%
\bibitem [{\citenamefont {Ma}\ \emph {et~al.}(2013)\citenamefont {Ma},
  \citenamefont {Xu}, \citenamefont {Xu}, \citenamefont {Tan}, \citenamefont
  {Qi},\ and\ \citenamefont {Lo}}]{Ma13}%
  \BibitemOpen
  \bibfield  {author} {\bibinfo {author} {\bibfnamefont {X.}~\bibnamefont
  {Ma}}, \bibinfo {author} {\bibfnamefont {F.}~\bibnamefont {Xu}}, \bibinfo
  {author} {\bibfnamefont {H.}~\bibnamefont {Xu}}, \bibinfo {author}
  {\bibfnamefont {X.}~\bibnamefont {Tan}}, \bibinfo {author} {\bibfnamefont
  {B.}~\bibnamefont {Qi}}, \ and\ \bibinfo {author} {\bibfnamefont {H.-K.}\
  \bibnamefont {Lo}},\ }\href {\doibase 10.1103/PhysRevA.87.062327} {\bibfield
  {journal} {\bibinfo  {journal} {Phys. Rev. A}\ }\textbf {\bibinfo {volume}
  {87}},\ \bibinfo {pages} {062327} (\bibinfo {year} {2013})}\BibitemShut
  {NoStop}%
\bibitem [{\citenamefont {Mansour}\ \emph {et~al.}(1993)\citenamefont
  {Mansour}, \citenamefont {Nisan},\ and\ \citenamefont {Tiwari}}]{MNT}%
  \BibitemOpen
  \bibfield  {author} {\bibinfo {author} {\bibfnamefont {Y.}~\bibnamefont
  {Mansour}}, \bibinfo {author} {\bibfnamefont {N.}~\bibnamefont {Nisan}}, \
  and\ \bibinfo {author} {\bibfnamefont {P.}~\bibnamefont {Tiwari}},\ }\href
  {\doibase https://doi.org/10.1016/0304-3975(93)90257-T} {\bibfield  {journal}
  {\bibinfo  {journal} {Theoretical Computer Science}\ }\textbf {\bibinfo
  {volume} {107}},\ \bibinfo {pages} {121} (\bibinfo {year}
  {1993})}\BibitemShut {NoStop}%
\bibitem [{\citenamefont {Krawczyk}(1994)}]{Krawczyk}%
  \BibitemOpen
  \bibfield  {author} {\bibinfo {author} {\bibfnamefont {H.}~\bibnamefont
  {Krawczyk}},\ }in\ \href@noop {} {\emph {\bibinfo {booktitle} {Proceedings of
  the 14th Annual International Cryptology Conference on Advances in
  Cryptology}}}\ (\bibinfo {year} {1994})\ p.\ \bibinfo {pages}
  {129–139}\BibitemShut {NoStop}%
\bibitem [{\citenamefont {Gohberg}\ and\ \citenamefont
  {Olshevsky}(1994)}]{gohberg1994complexity}%
  \BibitemOpen
  \bibfield  {author} {\bibinfo {author} {\bibfnamefont {I.}~\bibnamefont
  {Gohberg}}\ and\ \bibinfo {author} {\bibfnamefont {V.}~\bibnamefont
  {Olshevsky}},\ }\href {\doibase https://doi.org/10.1016/0024-3795(94)90189-9}
  {\bibfield  {journal} {\bibinfo  {journal} {Linear Algebra and its
  Applications}\ }\textbf {\bibinfo {volume} {202}},\ \bibinfo {pages} {163}
  (\bibinfo {year} {1994})}\BibitemShut {NoStop}%
\bibitem [{\citenamefont {{Frigo}}\ and\ \citenamefont
  {{Johnson}}(1998)}]{FrigoJo98}%
  \BibitemOpen
  \bibfield  {author} {\bibinfo {author} {\bibfnamefont {M.}~\bibnamefont
  {{Frigo}}}\ and\ \bibinfo {author} {\bibfnamefont {S.~G.}\ \bibnamefont
  {{Johnson}}},\ }in\ \href {\doibase 10.1109/ICASSP.1998.681704} {\emph
  {\bibinfo {booktitle} {Proceedings of the 1998 IEEE International Conference
  on Acoustics, Speech and Signal Processing}}}\ (\bibinfo {year} {1998})\ pp.\
  \bibinfo {pages} {1381--1384}\BibitemShut {NoStop}%
\bibitem [{\citenamefont {Trevisan}(2001)}]{trevisan}%
  \BibitemOpen
  \bibfield  {author} {\bibinfo {author} {\bibfnamefont {L.}~\bibnamefont
  {Trevisan}},\ }\href {\doibase 10.1145/502090.502099} {\bibfield  {journal}
  {\bibinfo  {journal} {Journal of the ACM}\ }\textbf {\bibinfo {volume}
  {48}},\ \bibinfo {pages} {860–879} (\bibinfo {year} {2001})}\BibitemShut
  {NoStop}%
\bibitem [{\citenamefont {Alfers}\ and\ \citenamefont {Dinges}(1984)}]{AD}%
  \BibitemOpen
  \bibfield  {author} {\bibinfo {author} {\bibfnamefont {D.}~\bibnamefont
  {Alfers}}\ and\ \bibinfo {author} {\bibfnamefont {H.}~\bibnamefont
  {Dinges}},\ }\href {\doibase 10.1007/BF00533744} {\bibfield  {journal}
  {\bibinfo  {journal} {Zeitschrift f{\"u}r Wahrscheinlichkeitstheorie und
  Verwandte Gebiete}\ }\textbf {\bibinfo {volume} {65}},\ \bibinfo {pages}
  {399} (\bibinfo {year} {1984})}\BibitemShut {NoStop}%
\bibitem [{\citenamefont {Zubkov}\ and\ \citenamefont {Serov}(2013)}]{SZ}%
  \BibitemOpen
  \bibfield  {author} {\bibinfo {author} {\bibfnamefont {A.~M.}\ \bibnamefont
  {Zubkov}}\ and\ \bibinfo {author} {\bibfnamefont {A.~A.}\ \bibnamefont
  {Serov}},\ }\href {\doibase 10.1137/S0040585X97986138} {\bibfield  {journal}
  {\bibinfo  {journal} {Theory of Probability \& Its Applications}\ }\textbf
  {\bibinfo {volume} {57}},\ \bibinfo {pages} {539} (\bibinfo {year}
  {2013})}\BibitemShut {NoStop}%
\bibitem [{\citenamefont {Müller-Lennert}\ \emph {et~al.}(2013)\citenamefont
  {Müller-Lennert}, \citenamefont {Dupuis}, \citenamefont {Szehr},
  \citenamefont {Fehr},\ and\ \citenamefont {Tomamichel}}]{MDSFT}%
  \BibitemOpen
  \bibfield  {author} {\bibinfo {author} {\bibfnamefont {M.}~\bibnamefont
  {Müller-Lennert}}, \bibinfo {author} {\bibfnamefont {F.}~\bibnamefont
  {Dupuis}}, \bibinfo {author} {\bibfnamefont {O.}~\bibnamefont {Szehr}},
  \bibinfo {author} {\bibfnamefont {S.}~\bibnamefont {Fehr}}, \ and\ \bibinfo
  {author} {\bibfnamefont {M.}~\bibnamefont {Tomamichel}},\ }\href {\doibase
  10.1063/1.4838856} {\bibfield  {journal} {\bibinfo  {journal} {Journal of
  Mathematical Physics}\ }\textbf {\bibinfo {volume} {54}},\ \bibinfo {pages}
  {122203} (\bibinfo {year} {2013})}\BibitemShut {NoStop}%
\bibitem [{\citenamefont {{Tomamichel}}\ \emph {et~al.}(2009)\citenamefont
  {{Tomamichel}}, \citenamefont {{Colbeck}},\ and\ \citenamefont
  {{Renner}}}]{TCR}%
  \BibitemOpen
  \bibfield  {author} {\bibinfo {author} {\bibfnamefont {M.}~\bibnamefont
  {{Tomamichel}}}, \bibinfo {author} {\bibfnamefont {R.}~\bibnamefont
  {{Colbeck}}}, \ and\ \bibinfo {author} {\bibfnamefont {R.}~\bibnamefont
  {{Renner}}},\ }\href {\doibase 10.1109/TIT.2009.2032797} {\bibfield
  {journal} {\bibinfo  {journal} {IEEE Transactions on Information Theory}\
  }\textbf {\bibinfo {volume} {55}},\ \bibinfo {pages} {5840} (\bibinfo {year}
  {2009})}\BibitemShut {NoStop}%
\bibitem [{\citenamefont {Tomamichel}(2015)}]{Tomamichel_book}%
  \BibitemOpen
  \bibfield  {author} {\bibinfo {author} {\bibfnamefont {M.}~\bibnamefont
  {Tomamichel}},\ }\href@noop {} {\emph {\bibinfo {title} {Quantum Information
  Processing with Finite Resources: Mathematical Foundations}}}\ (\bibinfo
  {year} {2015})\BibitemShut {NoStop}%
\bibitem [{\citenamefont {Zhang}\ \emph
  {et~al.}(2020{\natexlab{b}})\citenamefont {Zhang}, \citenamefont {Fu},\ and\
  \citenamefont {Knill}}]{QPE}%
  \BibitemOpen
  \bibfield  {author} {\bibinfo {author} {\bibfnamefont {Y.}~\bibnamefont
  {Zhang}}, \bibinfo {author} {\bibfnamefont {H.}~\bibnamefont {Fu}}, \ and\
  \bibinfo {author} {\bibfnamefont {E.}~\bibnamefont {Knill}},\ }\href
  {\doibase 10.1103/PhysRevResearch.2.013016} {\bibfield  {journal} {\bibinfo
  {journal} {Phys. Rev. Research}\ }\textbf {\bibinfo {volume} {2}},\ \bibinfo
  {pages} {013016} (\bibinfo {year} {2020}{\natexlab{b}})}\BibitemShut
  {NoStop}%
\bibitem [{\citenamefont {Tsirelson}(1993)}]{Cirelson93}%
  \BibitemOpen
  \bibfield  {author} {\bibinfo {author} {\bibfnamefont {B.~S.}\ \bibnamefont
  {Tsirelson}},\ }\href@noop {} {\bibfield  {journal} {\bibinfo  {journal}
  {Hadronic Journal Supplement}\ }\textbf {\bibinfo {volume} {8}},\ \bibinfo
  {pages} {329} (\bibinfo {year} {1993})}\BibitemShut {NoStop}%
\bibitem [{\citenamefont {Pironio}\ \emph {et~al.}(2009)\citenamefont
  {Pironio}, \citenamefont {Ac{\'{\i}}n}, \citenamefont {Brunner},
  \citenamefont {Gisin}, \citenamefont {Massar},\ and\ \citenamefont
  {Scarani}}]{PABGMS}%
  \BibitemOpen
  \bibfield  {author} {\bibinfo {author} {\bibfnamefont {S.}~\bibnamefont
  {Pironio}}, \bibinfo {author} {\bibfnamefont {A.}~\bibnamefont
  {Ac{\'{\i}}n}}, \bibinfo {author} {\bibfnamefont {N.}~\bibnamefont
  {Brunner}}, \bibinfo {author} {\bibfnamefont {N.}~\bibnamefont {Gisin}},
  \bibinfo {author} {\bibfnamefont {S.}~\bibnamefont {Massar}}, \ and\ \bibinfo
  {author} {\bibfnamefont {V.}~\bibnamefont {Scarani}},\ }\href {\doibase
  10.1088/1367-2630/11/4/045021} {\bibfield  {journal} {\bibinfo  {journal}
  {New Journal of Physics}\ }\textbf {\bibinfo {volume} {11}},\ \bibinfo
  {pages} {045021} (\bibinfo {year} {2009})}\BibitemShut {NoStop}%
\bibitem [{NIS()}]{NIST_Tests}%
  \BibitemOpen
  \href@noop {} {\enquote {\bibinfo {title} {{NIST} statistical test suite},}\
  }\bibinfo {howpublished}
  {\url{http://csrc.nist.gov/groups/ST/toolkit/rng/stats_tests.html}}\BibitemShut
  {NoStop}%
\end{thebibliography}%
\pagenumbering{gobble} 
\end{document}